\documentclass[12pt,american]{article}
\newcommand{\cH}{\mathcal{H}}
\usepackage{lmodern}
\usepackage[T1]{fontenc}
\usepackage[latin9]{inputenc}
\usepackage{geometry}
\usepackage[inactive]{srcltx}
\usepackage{babel}
\usepackage{threeparttable}
\usepackage{booktabs}
\usepackage{comment}
\usepackage{amsfonts}
\usepackage{amsmath}
\DeclareMathOperator{\tr}{tr}
\DeclareMathOperator{\supp}{supp}
\DeclareMathOperator{\sgn}{sgn}
\usepackage{amsthm}
\usepackage{amssymb}
\usepackage{setspace}
\usepackage[authoryear]{natbib}
\usepackage[unicode=true,pdfusetitle,
  bookmarks=true,bookmarksnumbered=false,bookmarksopen=false,
breaklinks=false,pdfborder={0 0 1},backref=false,colorlinks=false]
{hyperref}
\hypersetup{
colorlinks,linkcolor={blue!70!black},citecolor={blue!50!black},urlcolor={blue!60!black}}

\makeatletter

\theoremstyle{definition}
\newtheorem{example}{\protect\examplename}
\theoremstyle{plain}
\newtheorem{assumption}{\protect\assumptionname}
\theoremstyle{definition}
\newtheorem*{example*}{\protect\examplename}
\theoremstyle{plain}
\newtheorem{lem}{\protect\lemmaname}
\theoremstyle{plain}
\newtheorem{thm}{\protect\theoremname}
\theoremstyle{plain}
\newtheorem{cor}{\protect\corollaryname}
\theoremstyle{plain}
\newtheorem{prop}{\protect\propositionname}
\newtheorem{rem}{\protect\remarkname}
\theoremstyle{plain}

\usepackage{babel}
\usepackage{float}
\usepackage{mathrsfs}
\usepackage{breakurl}
\usepackage{bbold}
\usepackage{bbm}
\usepackage{tikz}
\usepackage{enumitem}
\allowdisplaybreaks

\newcommand{\R}{\mathbb{R}}
\newcommand{\E}{\mathbb{E}}
\newcommand{\Var}{\mathrm{Var}}
\newcommand{\X}{\mathcal{X}}

\newcommand{\dHd}{d\mathcal{H}^{d-1}}

\newcommand{\norm}[1]{\left\| #1 \right\|}
\newcommand{\abs}[1]{\left| #1 \right|}

\newcommand{\customlabel}[2]{%
  \protected@write \@auxout {}{\string \newlabel {#1}{{#2}{\thepage}{#2}{#1}{}} }%
  \hypertarget{#1}{}
}

\providecommand{\assumptionname}{Assumption}
\providecommand{\examplename}{Example}
\providecommand{\lemmaname}{Lemma}
\providecommand{\theoremname}{Theorem}
\providecommand{\remarkname}{Remark}

\makeatother

\providecommand{\assumptionname}{Assumption}
\providecommand{\corollaryname}{Corollary}
\providecommand{\examplename}{Example}
\providecommand{\lemmaname}{Lemma}
\providecommand{\propositionname}{Proposition}
\providecommand{\theoremname}{Theorem}

\begin{document}
\global\long\def\a{\alpha}%
\global\long\def\b{\beta}%
\global\long\def\g{\gamma}%
\global\long\def\d{\delta}%
\global\long\def\e{\epsilon}%
\global\long\def\l{\lambda}%
\global\long\def\t{\theta}%
\global\long\def\o{\omega}%
\global\long\def\s{\sigma}%
\global\long\def\G{\Gamma}%
\global\long\def\D{\Delta}%
\global\long\def\L{\Lambda}%
\global\long\def\T{\Theta}%
\global\long\def\O{\Omega}%
\global\long\def\H{\mathbb{H}}%
\global\long\def\R{\mathbb{R}}%
\global\long\def\N{\mathbb{N}}%
\global\long\def\Q{\mathbb{Q}}%
\global\long\def\I{\mathbb{I}}%
\global\long\def\P{\mathbb{P}}%
\global\long\def\E{\mathbb{E}}%
\global\long\def\B{\mathbb{B}}%
\global\long\def\S{\mathbb{S}}%
\global\long\def\V{\mathbb{V}\text{ar}}%
\global\long\def\GG{\mathbb{G}}%
\global\long\def\TT{\mathbb{T}}%
\global\long\def\X{{\bf X}}%
\global\long\def\cX{\mathscr{X}}%
\global\long\def\cY{\mathscr{Y}}%
\global\long\def\cA{\mathscr{A}}%
\global\long\def\cB{\mathscr{B}}%
\global\long\def\cF{\mathscr{F}}%
\global\long\def\cM{\mathcal{M}}%
\global\long\def\cN{\mathcal{N}}%
\global\long\def\cG{\mathcal{G}}%
\global\long\def\cC{\mathcal{C}}%
\global\long\def\cT{\mathcal{T}}%
\global\long\def\sp{\,}%
\global\long\def\es{\emptyset}%
\global\long\def\mc#1{\mathscr{#1}}%
\global\long\def\ind{\mathbb 1}%
\global\long\def\indep{\perp}%
\global\long\def\any{\forall}%
\global\long\def\ex{\exists}%
\global\long\def\p{\partial}%
\global\long\def\cd{\cdot}%
\global\long\def\Dif{\nabla}%
\global\long\def\imp{\Rightarrow}%
\global\long\def\iff{\Leftrightarrow}%
\global\long\def\up{\uparrow}%
\global\long\def\down{\downarrow}%
\global\long\def\arrow{\rightarrow}%
\global\long\def\rlarrow{\leftrightarrow}%
\global\long\def\lrarrow{\leftrightarrow}%
\global\long\def\abs#1{\left|#1\right|}%
\global\long\def\norm#1{\left\Vert #1\right\Vert }%
\global\long\def\rest#1{\left.#1\right|}%
\global\long\def\bracket#1#2{\left\langle #1\middle\vert#2\right\rangle }%
\global\long\def\sandvich#1#2#3{\left\langle #1\middle\vert#2\middle\vert#3\right\rangle }%
\global\long\def\third#1{\frac{#1}{3}}%
\global\long\def\ellipsis{\textellipsis}%
\global\long\def\sand#1{\left\lceil #1\right\vert }%
\global\long\def\wich#1{\left\vert #1\right\rfloor }%
\global\long\def\sandwich#1#2#3{\left\lceil #1\middle\vert#2\middle\vert#3\right\rfloor }%
\global\long\def\inprod#1{\left\langle #1\right\rangle }%
\global\long\def\ol#1{\overline{#1}}%
\global\long\def\ul#1{\underline{#1}}%
\global\long\def\td#1{\tilde{#1}}%
\global\long\def\bs#1{\boldsymbol{#1}}%
\global\long\def\upto{\nearrow}%
\global\long\def\downto{\searrow}%
\global\long\def\pto{\overset{p}{\longrightarrow}}%
\global\long\def\dto{\overset{d}{\longrightarrow}}%
\global\long\def\asto{\overset{a.s.}{\longrightarrow}}%
\setlength{\abovedisplayskip}{6pt} \setlength{\belowdisplayskip}{6pt}

\sloppy

\title{
Thin Sets Are Not Equally Thin:\\ Minimax Learning of Submanifold Integrals\footnote{First version: arXiv:2507.12673v1. We acknowledge helpful comments from T. Armstrong, F. Bugni, T. Cai, X. Cheng, B. Deaner, Y. Fan, B. Honore, S. Khan, P. Kline, S. Kwon, O. Linton, M. Masten, D. Pouzo, J. Powell, A. Rosen, Y. Sun, and numerous conference and seminar participants.}}
\author{Xiaohong Chen\footnote{Chen: Department of Economics and Cowles Foundation for Research in Economics,
    Yale University, USA.
  xiaohong.chen@yale.edu. Chen thanks Cowles Foundation for research support.}$\ \ $and Wayne Yuan Gao\footnote{Gao: Department of Economics, University of Pennsylvania, USA.
waynegao@upenn.edu.}}
\maketitle

\begin{abstract}
  \noindent Many economic parameters are identified by ``thin sets'' (submanifolds with Lebesgue measure zero)  and hence difficult to recover from data in an ambient space. This paper provides a unified theory for estimation and inference of such ``thin-set'' identified functionals. We show that thin sets are \emph{not} equally thin: their intrinsic dimensionality $m$ matters in a precise manner. For a nonparametric regression $h_0$ with H\"{o}lder smoothness $s$ and $d$-dimensional covariates in the ambient space, we show that $n^{-\frac{s}{2s+d-m}}$ is the minimax optimal rate of estimating linear and nonlinear (e.g., quadratic, upper contour) integrals of $h_0$ on an $m$-dimensional submanifold ($0\leq m < d$), which is the fastest possible attainable rate among all estimators. The minimax lower bound rate result is generalized to estimating submanifold integrals when $h_0$ is a nonparametric density and a nonparametric instrumental variable function. The asymptotic normality of t statistics is established via sieve Riesz representation, and the corresponding inference is computed using Sobol points. \medskip

\end{abstract}

\section{\label{sec:Intro}Introduction}

Many parameters of interest in economics are identified by information contained in lower-dimensional \emph{thin sets}, i.e., subsets of the covariate space that have Lebesgue measure zero yet may still carry economic meanings. \citet*{khan2010irregular} coined the term \emph{thin-set identification} to describe settings in which identifying information is concentrated on such measure-zero sets, and showed that the resulting parameters are \emph{irregular} in the sense that they cannot be estimated at the parametric $n^{-1/2}$-rate, where $n$ denotes the sample size.

In this paper, we provide a more nuanced and, in some sense, refined view about \emph{thin sets}. While all thin sets possess Lebesgue measure zero in their ambient space, we show that they may differ substantively in terms of their intrinsic dimensions and geometric structures. These refined differences lead to quantitatively different convergence rate of estimation and different forms of first-order expansions for inference on aggregate parameters (integrals) defined on such thin sets.

Specifically, this paper considers semiparametric estimation and inference for general integral functionals on submanifolds of the following form:
\begin{equation}
  \G(h_0):=\int_{{\cal M}}\phi(h_{0}(x),x)w(x)d{\cal H}^{m}\left(x\right),\label{eq:SM_nonlinear}
\end{equation}
where $\phi:\R \times \mathcal{X} \mapsto \R$ is a known transformation of an unknown function $h_0$ that can be estimated nonparametrically from data, $w:\mathcal{X} \mapsto \R$ is a known (weight) function, and $\mathcal{X}$ is a known closed convex subset (with positive Lebesgue measure) in $\R^d$. Here ${\cal H}^{m}$
denotes the $m$-dimensional Hausdorff measure in $\R^{d}$, which coincides with the Lebesgue measure in $\R^m$ (and hence has Lebesgue measure zero in $\R^d$ whenever $m<d$).\footnote{See Section \ref{sec:Examples} for a precise definition of ${\cal H}^{m}$. For $m<d$, the $m$-dimensional Hausdorff measure ${\cal H}^{m}$ can be thought as a generalization of a variety of ``uniform'' measures for lower-dimensional subsets in $\R^d$ such as point, line, area, surface, and volume measures.} In this paper,
\begin{equation}\label{eq:M-defn}
  {\cal M}:=\left\{ x\in{\cal X}:g\left(x\right)={\bf 0}\right\}
\end{equation}
denotes an $m$-dimensional ($0\leq m <d$) submanifold, where $g:\mathcal{X} \mapsto \R^{d-m}$ may be known, or unknown but can be estimated parametrically, semiparametrically or nonparametrically from data. By definition,
the $m$-dimensional submanifold ${\cal M}$ has positive $m$-dimensional Hausdorff volume but zero $d$-dimensional Lebesgue volume, and hence is a thin set in the covariate ambient space $\mathcal{X}\subset \R^{d}$.

Submanifold integrals emerge naturally in many economic  settings where we take derivatives of a standard Lebesgue integral ($dx$ in $\R^d$),
such as an expectation, with respect to a changing domain of integration:
\begin{equation}
  \rest{\frac{d}{dt}\int_{\O_{t}}\omega_{t}\left(x\right)dx}_{t=0}=\rest{\frac{d}{dt}\int_{\R^d}\ind\left\{ x\in \O_{t}\right\}\omega_{t}\left(x\right)dx}_{t=0}.\label{eq:time_deriv}
\end{equation}
Mathematically, the derivative of an integral with respect
to its domain $\O_{t}$ translates into an integral over the \emph{boundary} $\p\O_t$ of its domain, and the boundary is often a lower-dimensional submanifold of the original domain. As we show in Section \ref{sec:Examples}, many parameters of interest are identified by solutions to first-order condition for optimization over subpopulation, which often take the form \eqref{eq:time_deriv}.  

Differentiation of an integral with respect to its domain also shows
up in asymptotic analysis. For example, consider estimation of the following integral functional on upper contour set of the unknown $h_0$
\begin{equation}
  V\left(h_{0}\right):=\int \ind\left\{ h_{0}(x)\geq0\right\}w\left(x\right)dx=\int_{\left\{x\in \mathcal{X}:~ h_{0}(x)\geq0\right\} }w\left(x\right)dx,\label{eq:WFB}
\end{equation}
where $w(x)$ is the marginal density of $x$ with support $\mathcal{X}\subset \R^d$. When $h_{0}\left(x\right)=\text{CATE}\left(x\right):=\E\left[\rest{Y_{i}\left(1\right)-Y_{i}\left(0\right)}X_{i}=x\right]$
is the conditional average treatment effect (CATE) for type $x \in \mathcal{X}$, we call $V\left(h_{0}\right)$ a value functional, which is a policy parameter of interest in many applications.
If $h_{0}$ is nonparametrically estimated by $\hat{h}$, the impact of the estimation error
on the plug-in estimation of $V(h_{0})$ can be analyzed using the
pathwise derivative of $V(h)-V(h_{0})$ with respect to $h_0$ in the direction $[h-h_{0}]$, which
becomes a submanifold integral of the form:
\[
  DV(h_0)[h-h_0]:=\rest{\frac{dV\left(h_0 +t[h-h_0]\right)}{dt}}_{t=0}=\int_{\mathcal{M}_0}\frac{h\left(x\right)-h_{0}\left(x\right)}{\norm{\Dif_{x}h_{0}\left(x\right)}}w\left(x\right)d{\cal H}^{d-1}\left(x\right),
\]
where the level set\footnote{In this paper we use ${\cal M}_0$ to highlight the level set when it is an unknown function $h_0$.} ${\cal M}_0:=\{x\in \mathcal{X}:~ h_{0}\left(x\right)=0\}$ is a
$m=(d-1)$-dimensional submanifold in $\R^{d}$.

In this paper, we assume that a data set $\mathcal{D}_n:=\{(Y_i,X_i)\}_{i=1}^n$ is a random sample of size $n$ drawn from an unknown probability distribution of $(Y,X)$, in which the unknown density of $X$ is bounded away from zero and infinity on its support $\mathcal{X}$ in $\R^d$. We also assume that the unknown function $h_0:\mathcal{X} \mapsto \R$ and/or the unknown submanifold mapping $g:\mathcal{X} \mapsto \R^{d-m}$ can be identified and estimated from the data $\mathcal{D}_n$ in the ambient space.\footnote{Our framework is different from the literature that assumes the data is sampled directly from lower dimensional submanifolds.}  Our parameters of interest are the $m$-dimensional submanifold integrals $\G(h_0)$ and the upper contour integrals $V(h_0)$. In Section \ref{sec:Examples} we show that many economic functionals of interest can be represented as integrals of the forms $\G(h_0)$ and $V(h_0)$.


We first establish the minimax lower bound rates of estimation for the linear integral functional on a submanifold:
\begin{equation}
  L(h_0):=\int_{{\cal M}}h_{0}\left(x\right)w\left(x\right)d{\cal H}^{m}\left(x\right),\label{eq:SM_Int}
\end{equation}
a core special case of $\G(h_0)$ in \eqref{eq:SM_nonlinear} with $\phi=h_0$ and ${\cal M}$ a known $m$-dimensional submanifold ($m < d$). For the sake of concreteness we assume that the unknown function $h_0:\mathcal{X} \mapsto \R$ belongs to a H\"{o}lder class of finite smoothness $s>0$. When $h_{0}$ is respectively a nonparametric regression $\E[Y|X]$, a nonparametric density of $X$ and a nonparametric instrumental variables (NPIV) regression $\E[Y-h_0 (X)|Z]=0$, we establish the minimax lower bound rates for estimating $L(h_0)$, which are the \textit{fastest} possible convergence rates among all estimators for $L(h_0)$. These lower bound rates are all \textit{slower} than the parametric convergence rate of $n^{-1/2}$ whenever $m<d$, confirming that $L(h_0)$ is an \textit{irregular} functional whenever $m<d$. Specifically, we show that $r^*_n:=n^{-\frac{s}{2s+d-m}}$ is the minimax lower bound rate for estimating $L(h_0)$ when $h_0$ is a nonparametric regression and a density. Interestingly, this rate coincides with the famous lower bound rate of \cite{stone1980optimal} for the pointwise estimation of a nonparametric regression with $(d-m)$-dimensional covariates. In Subsection \ref{subsec:Minimax_NPIV-LB} we also establish that $r_{NPIV,n}$ is the minimax lower bound rate for estimating $L(h_0)$ when $h_0$ satisfies a NPIV restriction but does not need to be point-identified.
Importantly, the rate $r_{NPIV,n}$ for $L(h_0)$ coincides with the previously established minimax lower bound rate for estimating a $(d-m)$-dimensional point-identified NPIV function (\cite{chen2011rate}).

Under some mild regularity conditions, we show that the minimax lower bound rates for $L(h_0)$ are also the minimax lower bounds rates for nonlinear integrals $\G(h_0)$ with possibly \textit{unknown} submanifolds, upper contour integrals $V(h_0)$, and surface integrals of the form
\begin{equation}
  S(h_0):=\int_{\left\{x\in \mathcal{X}:~ h_{0}\left(x\right)=c\right\} } w\left(x\right)d{\cal H}^{d-1}\left(x\right).\label{eq:S}
\end{equation}
For example, when $h_{0}$ is a nonparametric regression or a density, the rate $r^*_n=n^{-\frac{s}{2s+d-m}}$ is the minimax lower bound rate for estimating nonlinear integral $\G(h_0)$ on $m$-dimensional possibly unknown submanifolds. In particular, when $m=d-1$ the minimax rate $r^*_n$ becomes $n^{-\frac{s}{2s+1}}$, which is the \textit{fastest} possible convergence rate for estimating $V(h_0)$, $S(h_0)$ and $\G(h_0)$ on $m=(d-1)$-dimensional $\mathcal{M}$ among all estimators. This result recovers the minimax lower bound of \cite{horowitz1993et} for smooth maximum score estimation that corresponds to an $m=(d-1)$-dimensional parametric submanifold $\{x\in\mathcal{X}:~x'\beta_0=0\}$.\footnote{Our minimax lower bound proof follows the nonparametric literature such as \cite{Tsybakov2009}, which differs from the proof of \cite{horowitz1993et} for his smooth maximum score model.}

We then show that the above lower bound rates are \textit{attainable}, and hence \textit{minimax-optimal}, by presenting sieve-based estimators for $\t_0=L(h_0),~\G(h_0)$ and $V(h_0)$ for the case when $h_0$ is a nonparametric regression with H\"{o}lder smoothness $s$ and $d$-dimensional covariates. For the linear integral $L(h_0)$, we consider the plug-in sieve estimator. For the nonlinear integrals $\G(h_0)$ and $V(h_0)$, we consider plug-in, split-sample and leave-one-out sieve estimators. We provide low level sufficient conditions under which they achieve the optimal rate $r^*_n=n^{-\frac{s}{2s+d-m}}$, with the smoothness requirement for the split-sample and leave-one-out debiased estimators weaker than that for the plug-in estimators for $\G(h_0)$ and $V(h_0)$.


Given the irregularity of submanifold integral functionals, they do not admit well-defined Riesz representers; however, the \textit{sieve} Riesz representers remain well-defined and computable in closed form. Following \citet*{chen2014sieve,chen2014sieveM} and \cite*{chenpouzo2015sieve}, we construct valid confidence intervals via sieve student-$t$ statistics. By exploiting the submanifold structure, we characterize the growth rate of the sieve Riesz representer norm and obtain tighter control of nonlinear remainders. For the upper contour integral $V(h_0)$, its pathwise derivatives are computed using the calculus of moving submanifolds.

Monte Carlo simulations confirm our theoretical results: the sieve estimators produce reasonably small RMSEs that shrink with sample sizes, and the realized confidence interval coverage is close to the nominal 95\% level. The submanifold integrals in the simulations are numerically computed using Sobol quasi-random sequences \citep{sobol1967distribution} for its better numerical performance than uniform random sampling.  

In a companion paper \cite*{CCG2025}, we apply the theory developed here to inference on value functionals of the CATE under first-best nonparametric treatment assignment. Using the Job Training Partnership Act (JTPA) data set, we compute confidence intervals for the nonparametric  first-best welfare and the treatment share for the JTPA job training program, two parameters estimated in \cite{kitagawa2018should} without reported confidence intervals.

\subsection*{Related Literature}\label{sec:review}

Our paper contributes to the literature on semiparametric
estimation and inference on irregular integral functionals.
To our best knowledge, our paper is the first to
provide a unified theory on general
submanifold integral functionals of an unknown function $h_0$ with H\"{o}lder smoothness $s>0$. We establish the minimax-optimal estimation rate for linear and nonlinear submanifold integral functionals of $h_0$ and for integrals of upper contour set on $h_0$, in which $h_0$ could be a regression, a density, and a NPIV function. In addition, we provide simple asymptotic normality based confidence intervals for these irregular integral functionals using sieve Riesz representation.

Our minimax lower bound estimation rate results can be viewed as quantitative refinements of the famous singular semiparametric information bound results of \cite{chamberlain1986asymptotic} and \citet*{khan2010irregular} on ``thin-set'' identified parameters. In Section \ref{sec:Examples} we show many semiparametric ``thin-set'' identified parameters can be viewed as integral functionals on $m$-dimensional submanifolds (for $m<d$). Previously, \citet*{kim1990cube} shows that the first-order condition for the maximum score criterion of \citet*{manski1975maximum} is a submanifold integral with a $m=(d-1)$-dimensional hyperplane, and derives the famous $n^{-1/3}$ convergence rate using empirical process theory.\footnote{\cite{horowitz1992smoothed,horowitz1993et} obtains the upper and lower bound rates of $n^{-s/(2s+1)}$ for the smoothed maximum score estimator without mentioning $(d-1)$-dimensional submanifold.} \citet*{sasaki2015quantile} also notes that differentiation
with respect to the domain of integration produces a $m=(d-1)$-dimensional
submanifold integrals in his identification paper on quantile functions in nonseparable structural models.

We view our minimax lower bound results for the large class of submanifold functionals $L(h_0)$, $\G (h_0)$, $V(h_0)$ and $S(h_0)$ \textit{important}, as they provide a minimax informational criterion to compare many different machine learning estimators. Our paper presents various sieve estimators for $L(h_0)$, $\G(h_0)$ and $V(h_0)$ and establishes that they can attain the minimax lower bound rate of $r^*_n=n^{-\frac{s}{2s+d-m}}$ when $h_0$ is a nonparametric regression. Other estimators for these functionals can also be presented and be verified if they are minimax rate optimal. For example, \citet*{qiao2021nonparametric} proposes a kernel plug-in density estimator of the surface integral $S(h_0)$ of the form \eqref{eq:S} and establishes a convergence rate $n^{-s/(2s+1)}$ when $s\geq d+1$. \citet*{qiao2021nonparametric} concludes his paper by stating that ``Another open problem is the minimax rates of estimating the surface integrals on level sets''. Notice that the surface integral $S(h_0)$ is a $m=(d-1)$-dimensional nonlinear submanifold functional, his kernel estimator matches our minimax lower bound rate $r^*_n$ provided $s\geq d+1$.
In works that are concurrent to ours, \citet*{cattaneo2025dist,cattaneo2025loc} consider estimation of linear integrals over $m=1$-dimensional known submanifolds that arise in the boundary discontinuity designs using local polynomial regressions and obtain a minimax optimal convergence rate of $n^{-\frac{s}{2s+d-1}}$ in their contexts (they mostly consider $d=2,m=1$). Due to the lack of space, we leave it to future work to compare finite sample performance of different minimax rate-optimal estimators for these submanifold integrals.

Technically, \citet*{qiao2021nonparametric}, \citet*{cattaneo2025dist,cattaneo2025loc} and our paper all utilize some mathematical tools in differential geometry and geometric measure theory to establish the asymptotic properties of different estimators of different \textit{irregular} submanifold integrals. Differential geometry tools have also been used in asymptotic analysis of \textit{regular} functionals (i.e., the ones that can be estimated at the parametric convergence rate of $n^{-1/2}$). \citet*{chernozhukov2018sorted} uses integrals on level sets to study sorted partial effects in heterogeneous coefficient models and establishes the convergence rate of $n^{-1/2}$ for their regular functionals. \citet*{feng2024statistical} shows how Hausdorff integrals can be used in the analysis of regular integral functionals.

\paragraph{Organization of the Paper}
Section \ref{sec:Examples} provides motivating examples for submanifold integrals in econometrics. Section \ref{sec:Minimax-LB} establishes minimax lower bounds rates for estimating linear and nonlinear submanifold integrals $L(h_0),~\G(h_0)$ when $h_0$ is a nonparametric regression, a nonparametric density and a NPIV function respectively. Section  \ref{sec:Minimax-UB} shows that the minimax lower bound rate is attainable by sieve estimators for estimating $L(h_0),~\G(h_0)$ and $V(h_0)$ when $h_0$ is a nonparametric regression. Section \ref{sec:Nonlinear} provides the asymptotic normality of the sieve estimators proposed in Section \ref{sec:Minimax-UB}, along with consistency sieve variance estimators, for inference on both linear and nonlinear submanifold integrals.
Section \ref{sec:Simulation} presents Monte Carlo simulation results.
The Appendix contains sections on mathematical tools in differential geometry used in this paper, technical lemmas, as well as all the proofs.

\section{\label{sec:Examples} Setup and Examples}

\subsection{Setup}\label{subsec:Setup}

Throughout this paper, we let $\left\{Y_{i},X_{i}\right\}_{i=1}^{n}$ be a random sample of size $n$ drawn from a unknown joint probability distribution $P_{\left(Y,X\right)}$, where $Y_{i}$ is
a scalar-valued outcome variable, and $X_{i}$ is a vector of observed
covariates with a convex and compact support ${\cal X}\subseteq\R^{d}$. Let $h_{0}:{\cal X} \to \R$
be a nonparametric function\footnote{More generally, $h_0$ may be a vector of nonparametric functions.} that is directly identified
from the data and can be estimated using standard nonparametric estimation methods. A leading example of $h_0$ is the conditional expectation (nonparametric regression) function $h_0(x) = \E[\rest{Y_i}X_i=x]$, which will be our main focus in the paper. That said, $h_{0}$ may also take the form of density functions, conditional quantiles and structural regression functions in NPIV models.

We consider submanifolds ${\cal M}:=\left\{ x\in{\cal X}:g\left(x\right)={\bf 0}\right\}$ that take the form of level sets of functions $g$. Depending on the problem setup, $g$ may be known or unknown, parametric or nonparametric, and it may be taken to be different from or the same as $h_0$, which will be illustrated in the examples below and treated in subsequent sections.
We maintain the following standard regularity condition in this paper:

\begin{assumption}[Regular Level Set]
  \label{assu:RegLevelSet} (i) $g:{\cal X}\to\R^{d-m}$ is a continuously differentiable
  function with ${\bf 0}\in\text{int}\left(g\left({\cal X}\right)\right)\subseteq\R^{d-m}$; (ii) $\Dif_{x}g\left(x\right)$ has full rank $d-m$ for every $x\in{\cal M}:=\left\{ x\in{\cal X}:g\left(x\right)={\bf 0}\right\}$.\footnote{Note that ${\bf 0}$ may be replaced with any other constant vector
  ${\bf c} \in\text{int}\left(g\left({\cal X}\right)\right)\subseteq\R^{d-m}$
without affecting the results in this paper.}
\end{assumption}

Let $\mathcal{J}_g\left(x\right)$
denote the Jacobian of $g:\R^{d}\to\R^{d-m}$ defined by
\[
  \mathcal{J}_g\left(x\right):=\sqrt{\sum_{B(x)}\text{det}\left(B(x)\right)^{2}}=\sqrt{\det(\Dif_x g(x)\Dif_x g(x)')},
\]
where $B$ indexes all $\left(d-m\right)\times\left(d-m\right)$ minors
of $\Dif_x g\left(x\right)$. Under Assumption \ref{assu:RegLevelSet} and compactness of ${\cal X}$, we have $\mathcal{J}_g\left(x\right)\geq Const.>0$ for all $x\in\mathcal{M}$.

Under Assumption \ref{assu:RegLevelSet}, the set ${\cal M}$ defined in \eqref{eq:M-defn} is an $m$-dimensional submanifold of $\R^{d}$ (e.g. by Theorem 12.1 of \citealp{loomis2014advanced}). We note that ${\cal M}$ has zero Lebesgue measure on $\R^{d}$ for all $0\leq m <d$, although it has positive $m$-dimensional Hausdorff measure. In this paper we use ${\cal H}^{m}\left(x\right)$
to denote the $m$-dimensional Hausdorff measure on $\R^{d}$, which is defined as follows: For a set $A\subseteq \R^d$, define ${\cal H}^m(A) := \lim_{\d \to 0} {\cal H}^m_\d$, where, for any $\d\in(0,\infty)$, $${\cal H}^m_\d(A) := \inf\left\{\sum_{j=1}^\infty \alpha(m) \left(\frac{\text{diam}(C_j)}{2}\right)^m: A \subseteq \cup_{j=1}^\infty C_j, \text{diam}(C_j) \leq \d\right\},$$ with $\alpha(m)=\frac{\pi^{m/2}}{\text{Gamma}((m/2)+1)}=\frac{\pi^{m/2}}{\int_{0}^{\infty} e^{-x} x^{m/2}dx}$. The $m$-dimensional Hausdorff measure becomes the standard $m$-dimensional Lebesgue measure in the lower dimensional $\R^m$; see, for example, \cite{evans2015measure} for more details on the Hausdorff measure.

\subsection{Motivating Examples}\label{subsec:Examples}

In this subsection, we provide some econometric examples for submanifold integrals, 
which roughly belong to two categories.

\paragraph{Category 1:} Examples in which researchers are interested in some aggregate parameters
of \textbf{estimated or optimized subpopulations}. Then, either the
first-order expansion in asymptotic analysis (of estimators), or the
first-order condition for optimality, often takes the form of the
time derivative of an integral with a changing region of integration
$\O_{t}$, which produces a submanifold integral term by the generalized
Leibniz rule:\footnote{See, for example, Theorem 4.2 in Chapter 9 of \cite{delfour2001shapes}.}  under regularity conditions, \eqref{eq:time_deriv} can be expressed as
\begin{align}
  \rest{\frac{d}{dt}\int_{\O_{t}}\omega_{t}\left(x\right)dx}_{t=0} & =\underset{({\bf I})}{\underbrace{\rest{\int_{\O_{t}}\frac{\p}{\p t}\omega_{t}\left(x\right)dx}_{t=0}}}+\underset{\left({\bf II}\right)}{\underbrace{\rest{\int_{\p\O_{t}}\left\langle {\bf n}_{t}\left(x\right),{\bf v}_{t}\left(x\right)\right\rangle \omega_{t}\left(x\right)d\mathcal{H}^{d-1}\left(x\right)}_{t=0}}}\label{eq:Reynolds}
\end{align}
where term (I) captures the effect of the change in the integrand
$w_{t}\left(x\right)$ with the region of integration $\O_{t}$ held
fixed, while term (II) captures the effect of the change in the region
of integration $\O_{t}$ with the integrand $w_{t}\left(x\right)$
held fixed. Importantly, $\p\O_{t}$, the boundary of $\O_{t}$, is often a submanifold of dimension $d-1$, and thus term (II) takes the form of an integral over the submanifold $\p\O_t$ with respect to the $(d-1)$-dimensional Hausdorff measure, which can also be viewed as the surface measure on the boundary submanifold $\p\O_t$. For the integrand terms,
${\bf n}_{t}\left(x\right)$ is the outward-pointing unit normal vector,
${\bf v}_{t}\left(x\right)$ is the velocity vector associated with
the time movement in the $\p\O_{t}$. For example,
consider a simple domain $\O_{t}=\left[a_{t},b_{t}\right]\subset\R$, its boundary $\p\O_t$ is a set of two points $\{a_t,b_t\}$, which is the $0$-dimensional submanifold, and the $0$-dimensional Hausdorff measure is simply the point counting measure. When $x$ is one-dimensional, \eqref{eq:Reynolds} becomes the standard Leibniz rule for univariate calculus:
\[
  \rest{\frac{d}{dt}\int_{a_{t}}^{b_{t}}\omega_{t}\left(x\right)dx}_{t=0}=\rest{\int_{a_{t}}^{b_{t}}\frac{\p}{\p t}\omega_{t}\left(x\right)dx}_{t=0}+\rest{\left(\omega_{t}\left(b_{t}\right)\frac{d}{dt}b_{t}-\omega_{t}\left(a_{t}\right)\frac{d}{dt}a_{t}\right)}_{t=0}.
\]
Throughout this paper we focus on situations where term (II) in \eqref{eq:Reynolds} is not vanishing.
\paragraph{Category 2:} Researchers are interested (for some other reasons than above) in some aggregate parameters of certain \textbf{boundary or marginal} \textbf{subpopulations}  with the boundary or margin characterized by a lower-dimensional submanifold.

~

Roughly speaking, Examples \ref{exa: MaxScore}-\ref{exa:RMS} are of Category 1, Examples \ref{exa:EY_gX}-\ref{exa:AvgDer} are of Category 2, while Examples \ref{exa:NPIV}-\ref{exa:NPqt} can be of both categories. We emphasize that we do not intend the categorization above to be exact nor exhaustive, but more to provide a high-level summary of the origins of submanifold integrals.

\begin{example}[Maximum Score Estimation of Binary Choice Models]
  \label{exa: MaxScore} Consider any model that satisfies the following
  sign alignment restriction
  \begin{equation}
    h_{0}\left(x\right)\gtrless0\ \iff\ x^{'}\b_{0}\gtrless0\label{eq:SignAlign}
  \end{equation}
  where $h_{0}$ is a nonparametrically identified and estimable function
  of $x$ and $\b_{0}$ is a $d$-dimensional parameter normalized to
  lie on the unit sphere, i.e., $\b_{0}\in\S^{d-1}:=\left\{ \b\in\R^{d}:\norm{\b}=1\right\} $.
  For example, in the following binary choice model with a conditional
  median independence restriction as in \citet*{manski1975maximum},
  \[
    y_{i}=\ind\left\{ X_{i}^{'}\b_{0}+\e_{i}\geq0\right\} ,\quad\text{med}\left(\rest{\e_{i}}X_{i}\right)=0,
  \]
  the sign alignment restriction \eqref{eq:SignAlign} is satisfied
  with $h_{0}\left(x\right)=\E\left[\rest{Y_{i}}X_{i}=x\right]-\frac{1}{2}$.
  The population criterion function for maximum score estimator can
  then be written as
  \begin{align}
    W\left(\b\right) & :=\int h_{0}\left(x\right)\ind\left\{ x^{'}\b\geq0\right\} p\left(x\right)dx,\label{eq:Q_MS}
  \end{align}
  where $p\left(\cdot\right)$ denotes the density of $X_{i}$. Under
  appropriate conditions, $\b_{0}$ can be point identified under a scale
  normalization,
  $\b_{0}=\arg\max_{\b\in\S^{d-1}}W\left(\b\right).$
  By the generalized Leibniz rule, the first order condition (FOC) for the optimality of $\b_{0}$ is given by
  \begin{equation}
    {\bf 0}=\Dif_{\b}W\left(\b_{0}\right):=\int_{\mathcal{M}_0} h_{0}\left(x\right)xp\left(x\right)d{\cal H}^{d-1}\left(x\right),\label{eq:FOC_MS}
  \end{equation}
  where $\mathcal{M}=\left\{x\in\mathcal{X}:~ x^{'}\b_{0}=0\right\}$ is a $m=\left(d-1\right)$-dimensional
  hyperplane (see, e.g.,  \citet*{kim1990cube}). We note that the hyperplane $\mathcal{M}_0$ has Lebesgue measure $0$ in $\R^{d}$, which is why the identification of $\b_{0}$ is referred
  to as a type of ``thin-set identification'' in \citet*{khan2010irregular}.
  Consequently, the right-hand side of \eqref{eq:FOC_MS} cannot
  be represented by a regular Lebesgue integral, but instead by a Hausdorff
  integral over a $m=\left(d-1\right)$-dimensional manifold (hyperplane) in $\R^{d}$.
\end{example}
\begin{example}[Optimal Linear Treatment Assignment]\label{exa:LinTreat}
  It has been recognized, say in \citet*{kitagawa2018should}, that
  optimal linear treatment assignment problem shares some similarity with maximum
  score estimation. Specifically, consider the problem of optimizing
  over a parametric family of treatment assignment rules that assigns the
  treatment status $0/1$ according to $\ind\left\{ x^{'}\b\geq0\right\} $
  for a given observed type $x$, where $\b$ is a $d$-dimensional
  choice parameter. Then the welfare function of the assignment rule
  parameter $\b$ is given by
  \begin{equation}
    W\left(\b\right):=\int\ind\left\{ x^{'}\b\geq0\right\} h_{0}\left(x\right)p\left(x\right)dx,\label{eq:Treat_Welfare}
  \end{equation}
  where $h_{0}\left(x\right)=\E\left[\rest{Y_{i}\left(1\right)-Y_{i}\left(0\right)}X_{i}=x\right]$
  is the conditional average treatment effect (CATE) for type $x$.
  Hence \eqref{eq:Treat_Welfare} is of exactly the same form as the
  maximum score criterion function \eqref{eq:Q_MS}. As a result, the
  FOC of \eqref{eq:Treat_Welfare} for the optimal $\b_{0}$ (under scale normalization) is again
  given by the submanifold integral \eqref{eq:FOC_MS}. Often times,
  researchers conduct (costly) experiments on and estimate CATE
  from a sample of moderate sample size, but the target population on
  which the treatment in question might be implemented can be of a much
  larger scale. In such settings $p(\cdot)$ may be known or estimable
  using a much larger sample size than that used to estimate $h_{0}$,
  and thus we may focus on the estimation error for $h_{0}$ in the
  optimization of $W\left(\b\right)$.
\end{example}
\begin{example}[Aggregate Parameter over Estimated Subpopulation]
  \label{exa:EstSubpop} Consider the estimation of
  the following general welfare or value functional parameter
  \[
    W\left(h_{0},f_{0}\right):=\int\ind\left\{ h_{0}\left(x\right)\geq0\right\}f_{0}\left(x\right)p\left(x\right)dx,
  \]
  where $h_{0}$ and $f_{0}$ may be both unknown but nonparametrically
  estimable. For example, if we set $h_{0}(x)=\text{CATE}(x)$ and $f_{0}(x)=1$, then $W\left(h_{0},f_{0}\right)$ becomes the value functional $V(h_0)$ defined in \eqref{eq:WFB}. If we set $h_{0}(x)=\text{CATE}(x)$ and $f_{0}(x)=\lambda^{'}x$ for some fixed $\lambda \neq \textbf{0}$, then $W\left(h_{0},f_{0}\right)$
  becomes the average characteristics of the subpopulation with nonnegative
  CATE. If we set $h_{0}\left(x\right)=f_{0}\left(x\right)=\text{CATE}\left(x\right)$,
  then $W\left(h_{0},f_{0}\right)$ becomes $W\left(\text{CATE},\text{CATE}\right)$---the welfare under the ``first-best''
  treatment assignment. Alternatively, we may take $f_{0}\left(x\right)$
  to be any other value/cost function associated with type $x$ (see \citep*{CCG2025}). When $h_{0}$ and $f_{0}$ are real-valued functions and are nonparametrically estimated by $\hat{h}$ and $\hat{f}$,
  the pathwise derivative of $W\left(h_{0},f_{0}\right)$ with respect
  to $\left(h_{0},f_0\right)$ is a key object in the characterization
  of the asymptotic behaviors of the plug-in estimator $W\left(\hat{h},\hat{f}\right)$.
  Generally, the pathwise derivative of $W$ in the direction of $\left[h-h_{0},f-f_{0}\right]$
  is given by
  \begin{align}
    & DW\left(h_{0},f_{0}\right)\left[h-h_{0},f-f_{0}\right]\nonumber \\
    :=\  & \lim_{t\downto0}\frac{W\left(h_{0}+t\left(h-h_{0}\right),f_{0}+t\left(f-f_{0}\right)\right)-W\left(h_{0},f_{0}\right)}{t}\nonumber \\
    =\  & \int_{\left\{ h_{0}\left(x\right)\geq0\right\} }\left(f\left(x\right)-f_{0}\left(x\right)\right)p\left(x\right)dx+\int_{\left\{ h_{0}\left(x\right)=0\right\} }\left(h\left(x\right)-h_{0}\left(x\right)\right)\frac{f_{0}\left(x\right)p\left(x\right)}{\norm{\Dif_{x}h_{0}\left(x\right)}}d{\cal H}^{d-1}\left(x\right)\label{eq:NPcontour-1}
  \end{align}
  which consists of two terms: the first is the perturbation of the
  integrand $f_{0}$ in the direction of $[f-f_{0}]$ over the true region
  of integral $\left\{ h_{0}\left(x\right)\geq0\right\} $, while the
  second is the perturbation of (the boundary) region of integration
  induced by the perturbation of $h_{0}$ in the direction of $[h-h_{0}]$.
  While the first term is standard in the semiparametric estimation literature,
  the second term takes the nonstandard form of a Hausdorff integral over
  the $m=\left(d-1\right)$-dimensional submanifold ${\cal M}_0:=\left\{x\in \mathcal{X}:~ h_{0}\left(x\right)=0\right\}$.\footnote{In particular, $\norm{\Dif_{x}h_{0}\left(x\right)}$, which measures
    the ``thinness'' of the level set, enters into the derivative formula
    explicitly here. In the previous examples with hyperplane boundaries,
    we took $g(x)=x^{'}\b$ with $\b\in\S^{d-1}$, and
    thus $\norm{\Dif_{x}g(x)}=\norm{\b}=1$, which is
  why this term becomes implicit in formula \eqref{eq:FOC_MS}.} Hence,
  as long as $f_0(x)p(x)$ does not vanish on the submanifold ${\cal M}_0:=\left\{x\in \mathcal{X}:~ h_{0}\left(x\right)=0\right\}$, the second term in \eqref{eq:NPcontour-1}
  will remain as the leading term in the asymptotic
  behavior of $W\left(\hat{h},\hat{f}\right)-W(h_0,f_0)$.
  See our companion paper \citep*{CCG2025} for a detailed analysis of this example.
\end{example}

\begin{example}[Average Treatment Effects under Propensity Score or Density Trimming]\label{exa:TrimATE}\footnote{We thank Tim Armstrong for suggesting this example and for sending us his note on this.}
  \citet*[CHIM thereafter]{crump2009dealing}
  proposes as a systematic approach to deal with limited overlap problems
  in the estimation of ATEs, and shows that the optimal subpopulation
  that minimizes the asymptotic variance of ATE under homoskedastic errors
  takes the form of propensity score trimming:
  \[
    \text{ATE}_{p\text{-trimmed}}:=\E\left[\rest{\text{CATE}\left(X_{i}\right)}\a\leq p_{0}\left(X_{i}\right)\leq1-\a\right].
  \]
  In practice, the propensity score $p_{0}\left(x\right)$ might require
  nonparametric estimation. CHIM did not provide theoretical results
  on the asymptotic distribution of their proposed estimators with the
  first-stage nonparametric estimation error of $p_{0}\left(x\right)$
  taken into account. Clearly $\text{ATE}_{p\text{-trimmed}}$ is a
  two-sided version of Example \ref{exa:EstSubpop} with the region
  of integration defined by the two-sided inequality $\a\leq p_{0}\left(X_{i}\right)\leq1-\a$
  on the propensity score function, and the directional derivative of
  $\text{ATE}_{p\text{-trimmed}}$ with respect to $p_{0}$ will again
  features submanifold integrals on the level sets of $\left\{ x:p_{0}\left(x\right)=\a\right\} $
  and $\left\{ x:p_{0}\left(x\right)=1-\a\right\} $.
\end{example}

\begin{example}[ReLU-Based Generalized Maximum Score Estimator under Multi-Index Single-Crossing Conditions]\label{exa:RMS} In a follow-up paper \citep*{CGW2025}, we investigate a ReLU-based generalization of the maximum score estimator for multi-index single-crossing condition (MISC) models \citep{gao2019robust}.
  Formally, consider a random sample $\left(Y_{i},X_{i}\right)_{i=1}^{n}$
  where $Y_{i}$ is an outcome with support ${\cal Y}\subseteq\R^{d_{y}}$
  and $X_{i}:=\left(X_{i1},...,X_{iJ}\right)\in\R^{d\times J}$ is a $d\times J$ random matrix with support ${\cal X}\subset\R^{d\times J}$.
  Let $h_{0}:{\cal X}\to\R$ be a real-valued functional of the conditional
  distribution of $Y_{i}$ given $X_{i}$ that is directly identified
  and nonparametrically estimable from the data, and $\b_{0}\in \S^{d-1}$ be a finite-dimensional parameter.
  We say that $\left(h_{0},\b_{0}\right)$ satisfies the (\emph{multi-index
  single-crossing}) MISC condition if, for all $x=\left(x_{1},...,x_{J}\right)\in{\cal X}$,
  \begin{align}
    x_{j}^{'}\b_{0}>0,\ \forall j=1,...,J & \quad\imp\quad h_{0}\left(x\right)\geq0,\nonumber \\
    x_{j}^{'}\b_{0}<0,\ \forall j=1,...,J & \quad\imp\quad h_{0}\left(x\right)\leq0.\label{eq:MISC}
  \end{align}
  The condition is said to be \emph{strict} if the inequalities on the
  right-hand side of \eqref{eq:MISC} are strict, i.e., $h_{0}\left(x\right)>0$ whenever $x_{j}^{'}\b_{0}>0$
  for all $j$, and $h_{0}\left(x\right)<0$ whenever $x_{j}^{'}\b_{0}<0$
  for all $j$. Note that, when $J=1$, \eqref{eq:MISC} reduces exactly to the sign-alignment restriction in the original maximum score formulation \citep{manski1975maximum,manski1985semiparametric}. However, the maximum score criterion function in \citep{manski1975maximum,manski1985semiparametric} does not generalize to multi-index settings with $J\geq2$.

  \cite*{CGW2025} proposes a novel criterion function the encodes the sign restrictions in the MISC condition framework using ReLU functions (instead of indicator functions as Manski's maximum score criterion function). Specifically, given $h:{\cal X}\to\R$ and $\b\in \S^{d-1}$, define
  \begin{align}
    g_{+,\b,h}\left(x\right) & :=\left[h\left(x\right)-\min_{1\le j\le J}\left(-x_{j}^{'}\b\right)_{+}\right]_{+},\nonumber\\
    g_{-,\b,h}\left(x\right) & :=\left[-h\left(x\right)-\min_{1\le j\le J}\left(x_{j}^{'}\b\right)_{+}\right]_{+}.\label{eq:def_gminus_MISC}
  \end{align}
  \cite*{CGW2025} then constructs the following criterion function
  \[
    W\left(\b\right):=W_{+}\left(\b\right)+W_{-}\left(\b\right),\qquad W_{\pm}\left(\b\right):=\E\left[g_{\pm,\b,h_{0}}\left(X_{i}\right)\right].
  \]
  which is shown to ensure $\b_{0}\in\arg\max_{\b\in \S^{d-1}}W\left(\b\right)$ under the MISC condition. \cite*{CGW2025} then considers two estimation procedures: a two-step procedure in which a first-stage nonparametric estimator of $h_0$ is plugged into \eqref{eq:def_gminus_MISC} as well as a joint machine learning procedure by formulating the criterion $W$ as a special layer in a deep neural network (DNN) training problem. On a theoretical level, \cite*{CGW2025} shows that the MISC criterion $W(\b)$ also features ``thin-set'' identification with information concentrated on a $(d-1)$-dimensional submanifold that piecewisely takes the form of a hyperplane defined by $\{x_j'\b_0=0\}$ across $j=1,...,J$. Then the submanifold integral results in this paper are utilized to establish the convergence rate and asymptotic normality of the ReLU-based generalized maximum score estimator in \cite*{CGW2025}.

\end{example}

\begin{example}[Generalized Partial Means and Nonparametric Regression with Generated Covariates]
  \label{exa:EY_gX} Suppose that we are interested in the conditional
  mean of $\psi\left(Y_{i},X_{i}\right)$ given  $g\left(X_{i}\right)=c\in\R^{d-m}$,
  where $\psi$ is some known transformation of $Y_{i}$ and $X_{i}$:
  \[
    \E\left[\rest{\psi\left(Y_{i},X_{i}\right)}g\left(X_{i}\right)=c\right]=\int_{\left\{ g\left(x\right)=c\right\} }h_{0}\left(x\right)w_{0}\left(x\right)d{\cal H}^{m}\left(x\right)
  \]
  with $h_{0}\left(x\right):=\E\left[\rest{\psi\left(Y_{i},X_{i}\right)}X_{i}=x\right]$,
  $w_{0}\left(x\right):=p_{0}\left(\rest xg\left(x\right)=c\right)=\frac{p_{0}\left(x\right)}{\mathcal{J}_g(x)p_{g\left(X\right)}\left(c\right)}$,
  $p_{0}\left(x\right)$ being the density of $X_{i}$, and $p_{g\left(X\right)}$
  is the density of $g(X_i)$. In particular, when $g$ extracts
  a subvector of $x\in\R^{d}$, say, $g\left(x\right)=\left(x_{1},...,x_{d-m}\right)$
  for some $1\leq m<d$, the above specializes to the partial mean functional
  of the form $\E\left[\rest{Y_{i}}X_{i,1}=c_{1},...,X_{i,d-m}=c_{d-m}\right]$
  as studied in \citet*{newey1994kernel}. The weighted boundary average treatment effect considered in \cite*{cattaneo2025dist} is another example with the level set function $g$ given by a known distance function. When $g$ is nonparametrically estimated, the above can be also viewed as the nonparametric regression function with nonparametrically generated covariates as considered in \cite*{mammen2012nonparametric}.\footnote{\cite*{mammen2012nonparametric} develops the asymptotic theory for generated covariates using kernel first-stage estimators, and establishes that, under appropriate conditions, the estimation errors from the first-stage are ``smoothed out'' partially, which can be viewed as a special case of the rate acceleration result developed in our paper.}
\end{example}

\begin{example}[Marginal Treatment Effects and Policy Relevant Treatment Effects]
  \label{exa:MTE} \citet*{heckman2005structural,heckman2007econometric}
  proposes the marginal treatment effect (MTE) as a unifying concept
  that underlies a wide variety of treatment effects studied in causal
  inference and program evaluations. It is well-known that MTE can be
  written as a derivative with respect to the propensity score:
  \begin{align}
    \text{MTE}\left(x,p\right) & :=\E\left[\rest{Y_{i}\left(1\right)-Y_{i}\left(0\right)}X_{i}=x,p_{0}\left(Z_{i}\right)=p\right],\nonumber \\
    & = \frac{\p}{\p p}\E\left[\rest{Y_i} X_i= x, p_0(Z_i)=p\right]
    \label{eq:MTE}
  \end{align}
  where $p_{0}\left(z\right):=\E\left[\rest{D_{i}=1}Z_{i}=z\right]$
  is the propensity score (for the binary treatment $D\in\{0,1\}$) given the instruments $Z_{i}$, $f_{0}\left(z\right|x)$
  is the conditional density of $Z_{i}$ given the observed covariates $X_i=x$, and $Y_i$ denotes the observed outcome.
  It is clear from \eqref{eq:MTE} that MTE takes the form of the derivative of a submanifold integral,
  where the submanifold is given by the propensity score level set $p_{0}\left(z\right)=p$.
  Here, both the $\text{CATE}$ function and the propensity score function
  can be nonparametrically estimated.

  As argued in \citet*{carneiro2010evaluating}, in many scenarios it
  is useful to consider incremental policy reforms and focus on the
  analysis of marginal policy changes, whose effects are usually concentrated
  on individuals ``at the margin''. The average marginal treatment
  effect (AMTE) summarizes the mean benefit of the treatment for the
  subpopulation who is indifferent between participation in the treatment
  and nonparticipation. Formally, this ``at-the-margin'' population
  $\left\{ \left(z,u\right)\in\R^{d_{z}+1}:p_{0}\left(z\right)-u=0\right\} $ is a submanifold defined by the level set of the propensity score function $p$,
  and the average over this subpopulation can again be represented by
  a submanifold integral in the form of
  \[
    \text{AMTE}\left(x\right):=\int_{\left\{ p_{0}\left(z\right)=u\right\} }\text{MTE}\left(x,u\right)f_{0}\left(\rest zx\right)d{\cal H}_{\left(z,u\right)}^{d_{z}}.
  \]

  In addition, MTE is also used to define a wide range of causal parameters,
  notably the policy-relevant treatment effect (PRTE) proposed in \citet*{heckman2001policy}. Let $F_{P}$ denote the CDF of $P_{i}:=p_{0}\left(Z_{i}\right)$
  and suppose that a proposed policy changes this CDF from $F_{P}$
  to $\tilde{F}$. Then the PRTE is defined by
  \[
    \text{PRTE}\left(x,\tilde{F}\right):=\int_{0}^{1}\text{MTE}\left(x,u\right)\o_{PRTE}\left(u,\tilde{F}\right)du,\quad\o_{PRTE}\left(u,\tilde{F}\right):=\frac{F_{P}\left(u\right)-\tilde{F}\left(u\right)}{\E_{\tilde{F}}\left(P_{i}\right)-\E_{F_{P}}\left(P_{i}\right)}.
  \]
  Since $\text{MTE}$ is defined as a derivative, PRTE takes the form
  of an average derivative parameter, for which a parametric estimation rate of $n^{-1/2}$
  is possible under a certain ``vanishing-on-boundary'' condition (see Example \ref{exa:AvgDer} for a related discussion).
  \citet*{carneiro2010evaluating}
  pointed out that this type of ``vanishing-on-boundary'' condition
  may be violated in plausible economic policy scenarios in the analysis of PRTE, which
  makes it not estimable at a parametric rate $n^{-1/2}$ in general. Our theoretical results can be useful
  in scenarios where ``vanishing-on-boundary'' conditions are not
  imposed in estimation and testing PRTE in various applications.
\end{example}

\begin{example}[Weighted Average Derivatives]
  \label{exa:AvgDer} Related to Example \ref{exa:MTE} above, suppose
  that we are interested in the following weighted average derivative
  of $h_0$:
  \[
    \text{WAD}\left(h_{0}\right):=\int\Dif_x h_{0}\left(x\right)w\left(x\right)dx.
  \]
  For example, \citet*[PSS thereafter]{powell1989} considers a density-weighted
  average derivative, which is a special case of the above
  with $w\left(x\right):=p^{2}\left(x\right)$, with $p\left(x\right)$
  denoting the marginal density of $x$ on its support ${\cal X}$.
  The standard analysis of $\text{WAD}\left(\hat{h}\right)$
  exploits the following integration-by-parts formula
  (a special case of the Divergence Theorem):
  \begin{equation}
    \int_{{\cal X}}\Dif_x h_{0}\left(x\right)w\left(x\right)dx=\int_{\p{\cal X}}\overrightarrow{{\bf n}}\left(x\right) h_{0}\left(x\right)w\left(x\right)d{\cal H}^{d-1}\left(x\right)-\int_{{\cal X}}h_{0}\left(x\right)\Dif_x w\left(x\right)dx,\label{eq:IntByParts_AD}
  \end{equation}
  where the first term is a Hausdorff integral over the submanifold
  defined by the support boundary $\p{\cal X}$. The standard
  analysis of $\text{WAD}\left(h_{0}\right)$, such as in PSS and \citet*{newey1993efficiency},
  exploits a ``vanishing-on-boundary'' assumption, say, $h_{0}\left(x\right)w\left(x\right)=0$
  for all $x\in{\cal \p X}$, so that the first term in \eqref{eq:IntByParts_AD} degenerates to
  $0$.
  For example, PSS assumes that $p\left(x\right)=0$ on $\p{\cal X}$
  (Assumption 2 in their paper), and hence $\text{WAD}\left(h_{0}\right)=-2\int h_{0}\left(x\right)\Dif_x p\left(x\right)p\left(x\right)dx$. PSS then proceeds to establish that $\sqrt{n} \left(\text{WAD}\left(\hat{h}\right)-\text{WAD}\left(h_{0}\right)\right)=O_p (1)$ and is asymptotic normal. As clear
  from \eqref{eq:IntByParts_AD}, the parametric convergence rate of $n^{-1/2}$ relies
  crucially on the ``vanishing on boundary'' assumption, without which
  the leading asymptotic term in $\text{WAD}\left(\hat{h}\right)-\text{WAD}\left(h_{0}\right)$ would be the first term (the submanifold integral) that converges at a slower-than-$n^{-1/2}$ rate, as established
  in our paper.
\end{example}

\begin{example}[Structural Functions in NPIV Regression]
  \label{exa:NPIV} The previous examples focus on \emph{exogenous}
  nonparametric regression functions. Here we present a canonical example
  of a nonparametric function with endogeneity: the structural functions
  in the nonparametric instrumental variables (NPIV) model, i.e., the
  $h_{0}$ function below:
  \[
    Y_{i}=h_{0}\left(X_{i}\right)+\e_{i},\quad\E\left[\rest{\e_{i}}Z_{i}\right]=0.
  \]
  Previous work by, for example,
  \citet*{chenpouzo2015sieve}, \citet*{chen2018optimal}, and \citet*{chen2025adaptive},
  provides theoretical results on the estimation and inference on $h_{0}$,
  as well as various (families of) linear and nonlinear functionals
  of $h_{0}$, including point evaluation, average derivatives/elasticities,
  and consumer surplus/welfare functionals. In principle, the $h_{0}$
  function in all the submanifold functionals defined in Examples \ref{exa: MaxScore}-\ref{exa:EY_gX}
  above could be replaced by the structural function in the NPIV model
  here, and the theory about submanifold integrals developed
  in the current paper would still be relevant.
\end{example}
\begin{example}[Nonparametric Quantile]
  \label{exa:NPqt} Similarly, the nonparametric function $h_{0}$
  need not be restricted as a conditional expectation function; it
  can be broadly defined as any estimable nonparametric function, for
  example, nonparametric quantile and density functions. These
  general nonparametric estimation problems (without endogeneity) have
  been studied widely in statistics, econometrics, and beyond in various
  settings. In the presence of endogeneity issues, \citet*{HL07econometrica}, \citet*{chenpouzo2015sieve}, and \citet*{chen2019penalized}
  have also provided theoretical results for the nonparametric quantile
  IV regression. The established
  theoretical results from these previous studies can be combined with
  the theory of submanifold integrals developed here for the study of
  new functional parameters of interest.
\end{example}

In summary, we hope that the above examples illustrate the need for estimation and inference on general submanifold integrals of unknown functions, which has not been systematically studied in the existing literature.
In the rest of the paper, 
we shall focus on the case where $h_0$ is a nonparametric regression function, and provide a general analysis of the estimation and inference for linear submanifold integrals \eqref{eq:SM_Int}, nonlinear submanifold integrals \eqref{eq:SM_nonlinear} whose first-order linear approximation takes the form of \eqref{eq:SM_Int}, as well as integrals on upper contour set of $h_0$ \eqref{eq:WFB}. We present additional results when $h_0$ is a nonparametric density function and a nonparametric instrumental variables regression in the Appendix.

\section{\label{sec:Minimax-LB}Minimax Rate of Estimation: Lower Bound}

We first establish minimax lower bound rates for estimation of linear ($L(h_0)$) and nonlinear ($\G(h_0)$) submanifold integrals when $h_0$ is a nonparametric regression $\E[Y|X]$ or a nonparametric density of $X$ in Section \ref{subsec:Minimax_Reg-LB}, and establish similar lower bound results when $h_0$ is an NPIV structural function in Subsection \ref{subsec:Minimax_NPIV-LB}.

Throughout this paper we assume that $h_0$ belongs to H\"{o}lder class of functions. We now recall the definition of H\"{o}lder class of functions. Let $\mathcal{X}=\mathcal{X}_{1}\times...\times\mathcal{X}_{d}$
be the Cartesian product of compact intervals $\mathcal{X}_{1},\dots,\mathcal{X}_{d}$,
say, ${\cal X}=\left[0,1\right]^{d}$ for simplicity. A real-valued
function $h$ on $\mathcal{X}$ is said to satisfy a H\"{o}lder  condition
with exponent $\gamma\in(0,1]$ if there is a positive number $c$
such that $\left|h(x)-h(y)\right|\leq c\norm{x-y}^{\gamma}$ for all
$x,y\in\mathcal{X}$; here $\norm{x-y}=\bigl(\sum_{l=1}^{d}x_{l}^{2}\bigr)^{1/2}$
is the Euclidean norm of $x=\left(x_{1},\dots,x_{d}\right)\in\mathcal{X}$.
Given a $d$-tuple $\alpha=\left(\alpha_{1},\dots,\alpha_{d}\right)$
of nonnegative integers, set $\left[\alpha\right]=\alpha_{1}+\dots+\alpha_{d}$
and let $\Dif^{\alpha}$ denote the differential operator defined by
\[
  \Dif^{\alpha}:=\Dif_x^{\alpha}=\frac{\partial^{\left[\alpha\right]}}{\partial x_{1}^{\alpha_{1}}\dots\partial x_{d}^{\alpha_{d}}}.
\]
Let $\lfloor s\rfloor$ be a nonnegative integer that is smaller than
$s$, and set $s=\lfloor s\rfloor+\gamma$ for some $\g\in(0,1]$.
A real-valued function $h$ on $\mathcal{X}$ is said to be $s$-$smooth$
if it is $\lfloor s\rfloor$ times continuously differentiable on
$\mathcal{X}$ and $\Dif^{\alpha}h$ satisfies a H\"{o}lder  condition with
exponent $\gamma$ for all $\alpha$ with $\left[\alpha\right]=\lfloor s\rfloor$.
Denote the class of all $s$-smooth real-valued functions on $\mathcal{X}$
by $\Lambda^{s}(\mathcal{X})$ (called a H\"{o}lder  class), and the space
of all $\lfloor s\rfloor$-times continuously differentiable real-valued
functions on $\mathcal{X}$ by $C^{\lfloor s\rfloor}(\mathcal{X})$.
Define a H\"{o}lder ball with smoothness $s=\lfloor s\rfloor+\gamma$
as
\[
  \Lambda_{c}^{s}\left(\mathcal{X}\right)=\left\{ h\in C^{\lfloor s\rfloor}(\mathcal{X}):\sup_{[\alpha]\leq\lfloor s\rfloor}\sup_{x\in\mathcal{X}}\left|\Dif^{\alpha}h(x)\right|\leq c,\sup_{[\alpha]=\lfloor s\rfloor}\sup_{x,y\in\mathcal{X},x\neq y}\frac{\left|\Dif^{\alpha}h\left(x\right)-\Dif^{\alpha}h\left(y\right)\right|}{\norm{x-y}^{\gamma}}\leq c\right\} .
\]

\begin{rem}
  If $h_0\in \Lambda^s$ (for H\"older smooth $s>0$), then $h_0|_{\mathcal{M}}$ is pointwise well-defined and hence the submanifold integrals $L(h_0)$ and $\Gamma(h_0)$ are well-defined. If $h_0$ is only assumed to belong to a Sobolev/Besov space of smoothness $s>0$ but not sup-norm bounded, then one needs to impose a Sobolev/Besov trace condition $s>(d-m)/2$ to ensure that $L(h_0)$ and $\Gamma(h_0)$ are well defined. This is why we assume that $h_0\in  \Lambda^s$ for $s>0$ instead of a Sobolev/Besov space of smoothness $s>(d-m)/2$ in this paper.
\end{rem}

In the proofs of the minimax lower bound rate results, especially in the proofs of lower bounds for nonlinear submanifold integrals with possibly unknown submanifolds that could depend on $h_0$, we assume that the submanifold function $g$ is H\"{o}lder smooth.
\begin{assumption}[Smoothness of Submanifolds]
  \label{assu:Jacob}  $g\in \L^{s_1}({\cal X})$ for some $s_1 \geq \max\{1,s\}$.
\end{assumption}

We wish to stress that the minimax lower bound rates established in our paper are still valid lower bound rates without imposing the extra smoothness Assumption \ref{assu:Jacob}. Nevertheless, Assumption \ref{assu:Jacob} is satisfied in typical economics applications with unknown submanifolds.

\subsection{\label{subsec:Minimax_Reg-LB} Case 1: $h_0$ is a Regression or a Density}

In Subsection \ref{subsubsec:Minimax_Lin-LB}, we first establish a minimax rate lower bound for estimating \emph{linear} integrals on submanifolds. In Subsection \ref{subsubsec:Minimax_NonLin-LB} we then present a general rate lower bound for estimating general  \emph{nonlinear} integrals over submanifolds whose first derivatives take the form of the linear integrals on submanifolds.

\subsubsection{\label{subsubsec:Minimax_Lin-LB}For Linear Functional $\t_0=L(h_0)$}

In this subsection we present the minimax lower bound rates for estimation of $\t_0=L(h_0)$ when $h_{0} (\cdot)$ is the unknown regression function $\E\left[\rest{Y_{i}}X_{i}=\cdot\right]$ (a nonparametric regression) or the unknown density of $X$.

Recall that $L:\Lambda^s \mapsto \R$ is the linear submanifold integral functional
\begin{equation}
  L\left(h\right):=\int_{{\cal M}}h\left(x\right)w\left(x\right)d{\cal H}^{m}\left(x\right),\label{eq:CoreFunc}
\end{equation}
where both the level set function
$g$ (consequently the manifold ${\cal M}$) and the weight function $w$ are assumed to be known.
We first establish lower bounds for the minimax convergence rates for estimating $\t_{0}=L\left(h_{0}\right)$
when $h_{0}$ is assumed to belong to the H\"{o}lder class of functions with smoothness $s>0$.
The established lower bound rates hold for any possible consistent estimators of $L(h_0)$.
%
%
%

We impose the following standard assumptions on the density of the covariates, the known weight function and the regression error term.

\begin{assumption}[Density on ${\cal X}$]
  \label{assu:density_x} ${\cal X}=\left[0,1\right]^{d}$, and the
  density $p_0$ of $X_{i}$ on ${\cal X}$ is uniformly bounded away from
  zero and infinity.
\end{assumption}
\begin{assumption}[Weight Regularity]
  \label{assu:w}
  The weight function $w:\mathcal{X} \mapsto \R$ is uniformly bounded on $\cM$ and not identically zero on $\cM$.
\end{assumption}
\begin{assumption}[Nondegeneracy]
  \label{assu:e2_below} $\inf_{x\in{\cal X}}\E\left[\rest{\e_{i}^{2}}X_{i}=x\right]>0$
  with $\e_{i}:=Y_{i}-h_{0}\left(X_{i}\right)$.
\end{assumption}

\begin{thm}[Lower Bound Rate for $\t_0=L(h_0)$ When $h_0$ is a Regression]
  \label{thm:LB-reg}
  Under Assumptions \ref{assu:RegLevelSet}, \ref{assu:Jacob}, \ref{assu:density_x}, \ref{assu:w} and \ref{assu:e2_below}, the rate
  $r^*_{n}=n^{-\frac{s}{2s+d-m}}$ is the minimax rate lower bound for the estimation of  $\t_{0}\left(P,w\right):=\int_{{\cal M}}h_{0}\left(x\right)w\left(x\right)d{\cal H}^{m}\left(x\right)$, i.e.
  \[
    \liminf_{n\longrightarrow\infty}\inf_{\tilde{\t}}\sup_{P,w}\E_{P}\left[n^{\frac{2s}{2s+d-m}}\left(\tilde{\t}-\t_{0}\left(P,w\right)\right)^{2}\right]\geq c,\quad\text{for some constant }c>0.
  \]
  where $P$ is any joint probability distribution of $\left(X_{i},Y_{i}\right)$
  that satisfies $h_{0}\left(\cd\right):=\E_{P}\left[Y_{i}|X_{i}=\cdot\right]\in\Lambda_{c}^{s}\left(\mathcal{X}\right)$
  along with Assumptions \ref{assu:density_x} and \ref{assu:e2_below}, $w$ is any weight function satisfying Assumption \ref{assu:w},
  and $\tilde{\t}$ is any estimator of $\t_{0}:=\t_{0}\left(P,w\right)$.
\end{thm}

Note that the lower bound rate $r^*_{n}=n^{-\frac{s}{2s+d-m}}$ reproduces
several well-known results in the literature as special cases:
When $m=d$, the lower bound rate becomes $r^*_{n}=n^{-\frac{1}{2}}$,
reproducing the standard parametric convergence rate for regular (full-dimensional)
integral functionals.\footnote{When $m=d$, the Hausdorff measure ${\cal H}^{d}$ coincides with Lebesgue
measure in $\R^{d}$.} When $m=0$, the lower bound becomes $r^*_{n}=n^{-\frac{s}{2s+d}}$,
reproducing the well-known \citep{stone1980optimal} minimax optimal
rate for point evaluation functionals of the $d$-dimensional nonparametric
regression.\footnote{When $m=0$, the Hausdorff measure ${\cal H}^{0}$ becomes a point
counting measure.}

Theorem \ref{thm:LB-reg} can be easily adapted to establish a minimax lower bound for $L(h_0)$ for other nonparametric object $h_0$ such as nonparametric density function of $X$ or nonparametric quantile regression of $Y$ on $X$. The following result shows that $r^*_n$ is also a minimax lower bound rate when $h_0$ is a nonparametric density function of $X$.

\begin{cor}[Lower Bound Rate for $\t_0=L(h_0)$ When $h_0$ is a Density]\label{cor:density} Under Assumptions   \ref{assu:RegLevelSet}, \ref{assu:density_x} and \ref{assu:w}, the rate
  $r^*_{n}=n^{-\frac{s}{2s+d-m}}$ is the minimax rate lower bound for the estimation of  $\t_{0}\left(P,w\right):=\int_{{\cal M}}h_{0}\left(x\right)w\left(x\right)d{\cal H}^{m}\left(x\right)$, i.e.
  \[
    \liminf_{n\longrightarrow\infty}\inf_{\tilde{\t}_n}\sup_{P,w}\E_{P}\left[n^{\frac{2s}{2s+d-m}}\left(\tilde{\t}-\t_{0}\left(P,w\right)\right)^{2}\right]\geq c,\quad\text{for some constant }c>0.
  \]
  where $P$ is any probability distribution of $X_i$ with its density satisfying
  $h_0\in\Lambda_{c}^{s}\left(\mathcal{X}\right)$ and Assumption \ref{assu:density_x},
  $w$ is any uniformly bounded weight function,
  and $\tilde{\t}_n$ is any estimator of $\t_{0}:=\t_{0}\left(P,w\right)$.
\end{cor}

\subsubsection{For Nonlinear Functional $\t_0=\G(h_0)$} \label{subsubsec:Minimax_NonLin-LB}

In the previous subsection we focus on the linear integral functionals over a known submanifold, but the key minimax rate $r^*_n = n^{-\frac{s}{2s+d-m}}$ continues to be relevant for a general class of nonlinear submanifold integral functionals.

Specifically, we consider the class of nonlinear integral functionals $\G\left(h_{0}\right)$ whose pathwise derivative takes the form of a linear submanifold integral of form \eqref{eq:CoreFunc}.
\begin{assumption}[Linearization]
  \label{assu:DGamma_linear} (i) $\G:\H\to\R$ is pathwise differentiable with
  \begin{equation}
    D\G\left(h_{0}\right)\left[v\right] :=\rest{\frac{\p}{\p t} \G(h_0+tv)}_{t=0} =\int_{\left\{ x:g\left(x\right)={\bf 0}\right\} }v\left(x\right)\ol{w}\left(x\right){\cal H}^{m}\left(x\right),\ \forall v\in{\cal V}:=\H-\left\{ h_{0}\right\} ,\label{eq:PathD}
  \end{equation}
  for some uniformly bounded function $\ol{w}$ with $\int_{\{x:g(x)={\bf 0}\}}\ol{w}^2(x)d\mathcal{H}^m(x)\in(0,\infty)$,  some level set function $g:{\cal X}\to\R_{d-m}$ satisfying Assumption \ref{assu:RegLevelSet} uniformly in $h_0\in\Lambda^s_c$, where both $\ol{w}$ and $g$ may depend on the true (and unknown) $h_{0}$;
  \newline
  (ii) There is a constant $0<C_r<1$ such that for any $u\in\Lambda^s_c$ with small $\norm{u}_{\infty}$,
  \[
    \abs{\Gamma(h_0+u)-\Gamma(h_0)-D\Gamma(h_0)[u]}\le C_r\abs{D\Gamma(h_0)[u]}.
  \]
\end{assumption}

Assumption \ref{assu:DGamma_linear}(ii) is used to prevent the linearization remainder from \emph{canceling} the first-order separation in the lower bound proof. It is easily satisfied by a Taylor series expansion to second order. It is also satisfied if $\G$ is Fr\'echet differentiable at $h_0$ with respect to the sup-norm, i.e., there exists a function $\omega:(0,\infty)\to(0,\infty)$ with $\omega(t)\to 0$ as $t\downarrow 0$
such that for all $u\in\Lambda^s_c$ with $\norm{u}_\infty\le t$,
\[
  \abs{\Gamma(h_0+u)-\Gamma(h_0)-D\Gamma(h_0)[u]}
  \;\le\; \omega(t)\,\norm{u}_\infty.
\]

Assumption \ref{assu:DGamma_linear} for nonlinear functional $\G(h_0)$ replaces the role of Assumption \ref{assu:w} for the linear functional $L(h_0)$. We emphasize that the level set function $g$ (and thus the submanifold) is \emph{not} necessarily known in Assumption \ref{assu:DGamma_linear}. For instance, Examples 3-6 in Section \ref{subsec:Examples} feature unknown and estimated submanifolds. Also, the upper contour integral $V(h_0)$ of the form \eqref{eq:WFB} and the surface integral $S(h_0)$ of the form \eqref{eq:S} are examples in which $g=h_0$ (nonparametrically defined).

Under Assumption \ref{assu:DGamma_linear}, the minimax rate lower bound established in Theorem \ref{thm:LB-reg} for linear submanifold integrals continues to be valid for nonlinear submanifold functionals $\G$ as well. This is because the pathwise derivative $D\G(h_0)[\cd]$ is a linear submanifold integral with minimax convergence rate given by $r^*_n$ in Theorem \ref{thm:LB-reg}, and the linearization remainder can only (potentially) slow down the rate. We summarize this observation in the following Corollary.

\begin{cor}[Lower Bound Rate for $\t_0=\G(h_0)$ When $h_0$ is a Regression or a Density]\label{cor:NL_rate_LB}
  Let $\Gamma$ be a potentially nonlinear functional satisfying Assumption \ref{assu:DGamma_linear}, and $h_0$ be a nonparametric regression or a density of $X$. Let Assumptions \ref{assu:Jacob}, \ref{assu:density_x}, and \ref{assu:e2_below} (when $h_0$ is a regression) hold.
  Then, a minimax rate lower bound for estimating $\t_0:=\Gamma(h_0)$ is given by $r^*_n = n^{-s/(2s + d - m)}$, i.e.,
  \[
    \liminf_{n\longrightarrow\infty}\inf_{\tilde{\t}}\sup_{P}\E_{P}\left[n^{\frac{2s}{2s+d-m}}\left(\tilde{\t}-\t_{0}\left(P\right)\right)^{2}\right]\geq c,\quad\text{for some constant }c>0.
  \]
\end{cor}

When $m=d-1$, the lower bound in Corollary \ref{cor:NL_rate_LB} becomes $r^*_{n}=n^{-\frac{s}{2s+1}}$ for nonlinear submanifold integrals. This could be viewed as an extension of \citet{horowitz1993et}'s result on the minimax lower bound rate for smoothed maximum score estimation.

\subsection{Case 2: When $h_0$ is an NPIV Structural Function} \label{subsec:Minimax_NPIV-LB}

\subsubsection{For Linear Functional $\t_0=L(h_0)$}

We now consider the rate lower bound for estimating the linear $\t_0=L(h_0)$ when
$h_0$ is the structural function in the nonparametric instrumental variable (NPIV) model:
\begin{align}
  Y_i = h_0(X_i) + \epsilon_i, \quad \E[\epsilon_i | Z_i] = 0,\label{eq:npreg}
\end{align}
where $X_i \in \mathcal{X} \subset \R^d$ are endogenous regressors and $Z_i \in \mathcal{Z} \subset \R^{d_z}$ are instruments.

Let $T: L^2(\mathcal{X}) \rightarrow L^2(\mathcal{Z})$ be the conditional expectation operator $(Th)(z)=\E[h(X)|Z=z]$. We now introduce some basic conditions on the NPIV model (\ref{eq:npreg}).

\begin{assumption} \label{a-NPIV-data}
  (i) $X_i$ has compact rectangular support $\mathcal{X} \subset \R^d$ with nonempty interior and the density of $X_i$ is uniformly bounded away from $0$ and $\infty$ on $\mathcal{X}$; (ii) $Z_i$ has compact rectangular support $\mathcal{Z} \subset \R^{d_z}$ and the density of $Z_i$ is uniformly bounded away from $0$ and $\infty$ on $\mathcal{Z}$; (iii) $\inf_{z\in \mathcal{Z}} \E[\epsilon_i^2|Z_i = z] \geq \underline \sigma^2 > 0$ with $\epsilon_i:= Y_i - h_0(X_i)$ satisfying NPIV model \eqref{eq:npreg}; (iv) $h_0 \in \Lambda^s_c$.
\end{assumption}

\begin{assumption} \label{a-NPIV-ID}
  For any $h_1,h_2\in \Lambda_c^s$, $Th_1 = Th_2 \ \text{in } L^2(\mathcal{Z})$ implies $L(h_1)=L(h_2)$.
\end{assumption}

To establish a lower bound, we require a \emph{link condition} that relates smoothness of $T$ to the parameter space for $h_0 \in \Lambda^s_c$. For simplicity we use the condition similar to that used in \cite{chen2011rate} and \cite{chen2018optimal}. Let $\{\psi_{\tilde{\jmath},k,\tilde{G}}\}$ denote a tensor-product CDV wavelet basis for  $L^2(\mathcal{X})\cap L^{\infty}(\mathcal{X})$ of regularity strictly larger than $s$. See, e.g., Appendix A of \cite{chen2018optimal} for details on the construction and properties of this basis.

\begin{assumption}[NPIV Link Condition]
  \label{ass:link}
  There is a positive decreasing function $\nu:(0, \infty)\to(0, \infty)$ such that
  \begin{equation}
    \label{eq:link-condition}
    \|T h\|_{L^2(\mathcal{Z})}^2
    \;\lesssim\;
    \sum_{\tilde{\jmath},\tilde{G},k}
    \nu(2^{\tilde{\jmath}})^2\,
    \bigl\langle h, \psi_{\tilde{\jmath},k,\tilde{G}}\bigr\rangle_{L^2(\mathcal{X})}^2,
    \qquad
    \forall h \in  \Lambda_c^s(\mathcal{X}).
  \end{equation}
\end{assumption}

In the NPIV literature (see \cite{HallHorowitz} and \cite{chen2011rate}), the mildly ill-posed case corresponds to $\nu(t) = t^{-\varsigma}$, roughly saying that the  operator $T$ converts $s$-smooth functions of $X$ into $(\varsigma+s)$-smooth functions of $Z$. The severely ill-posed case, which corresponds to choosing $\nu(t) = \exp(-\frac{1}{2}t^{\varsigma})$ and roughly says that $T$ maps smooth functions of $X$ into ``supersmooth'' functions of $Z$.

\begin{thm}[Lower Bound Rate for $\t_0=L(h_0)$ when $h_0$ is a NPIV]\label{thm:LB_NPIV}
  Let Assumptions \ref{assu:RegLevelSet}, \ref{a-NPIV-data}, \ref{a-NPIV-ID} and \ref{ass:link} hold for the NPIV model. We have
  \begin{align*}
    \liminf_{n \to \infty}  \inf_{\tilde{\theta}_n} \sup_{h \in \Lambda_c^s(\mathcal{X})} \E_{P_{h}} \left[r^{-2}_{NPIV,n} (\tilde{\theta}_n - L(h))^2 \right] \gtrsim  C>0
  \end{align*}
  where
  \begin{equation}
    r_{NPIV,n}  :=
    \begin{cases}
      n^{-\frac{s}{2(s+\varsigma) + d - m} }&  \text{ if mildly ill-posed,} \\
      (\log n)^{-\frac{s}{\varsigma}} &   \text{ if severely ill-posed,}
    \end{cases}
  \end{equation}
  and $\inf_{\tilde{\theta}_n}$ denotes the infimum over all estimators of $L(h)$ based on the sample of size $n$, $\sup_{h \in \Lambda_c^s(\mathcal{X})} \E_{P_{h}}$ denotes the sup over $h \in \Lambda_c^s(\mathcal{X})$  and probability distributions of $(X_i,Z_i,u_i)$ that satisfy Assumptions \ref{a-NPIV-data} and \ref{ass:link} with fixed $\nu$, and the finite constant $C>0$ does not depend on $n$.
\end{thm}

According to Theorem \ref{thm:LB_NPIV}, for severely ill-posed NPIV problems, the minimax lower bound for estimating linear submanifold integral $L(h_0)$ coincides with those in $L^2(X)$ (\cite{chen2011rate}) and in sup-norm (\cite{chen2018optimal}) estimation of the NPIV function $h_0$ itself, which is also the same as estimating a NPIV function in $L^2(X)$ estimation of a NPIV function with $(d-m)$ dimensional endogenous regressors.
For mildly ill-posed NPIV problems, the minimax lower bound for estimating $L(h_0)$ coincides with those in $L^2(X)$ estimation of a NPIV function with $(d-m)$ dimensional endogenous regressors.

\begin{rem}[Bounded completeness is not needed for the minimax lower bound for estimating $L(h_0)$]
  \label{rem:BC_not_needed_LB} Note that Assumption \ref{a-NPIV-ID} is equivalent to
  \[
    L(h)=0\  \text{for all } h\in \ker(T)\cap \Lambda_c^s(\mathcal{X}),
    \  \text{with~} \ker(T):=\{h\in L^2(\mathcal{X}): Th=0 \text{ in } L^2(\mathcal{Z})\},
  \]
  which is \emph{strictly weaker} than bounded completeness
  of $T$. Indeed, bounded completeness implies $\ker(T)=\{0\}$ and hence implies Assumption \ref{a-NPIV-ID},
  but the proof of Theorem \ref{thm:LB_NPIV} does not use injectivity of $T$:
  it only constructs two alternatives $h_0,h_1$ such that (a) the induced distributions are close in
  Kullback--Leibler distance (through small $\|T(h_1-h_0)\|_{L^2(\mathcal{Z})}$), while (b)
  $|L(h_1)-L(h_0)|$ is separated at the stated rate.

  Some identification is nevertheless needed to make estimation of $L(h_0)$ nontrivial.
  If Assumption \ref{a-NPIV-ID} fails, then there exists $h^\dagger\in \ker(T)\cap \Lambda_c^s(\mathcal{X})$ with
  $L(h^\dagger)\neq 0$. For any $h_0\in\Lambda_c^s(\mathcal{X})$ and sufficiently small $t\neq 0$
  such that $h_1:=h_0+t h^\dagger\in\Lambda_c^s(\mathcal{X})$, we have $Th_1=Th_0$ and hence the
  models indexed by $h_0$ and $h_1$ are observationally indistinguishable, but
  $L(h_1)-L(h_0)=t\,L(h^\dagger)\neq 0$. Hence, for any estimator $\tilde\theta_n$,
  \[
    \sup_{h\in\{h_0,h_1\}} \E_{P_h}\!\left[(\tilde\theta_n-L(h))^2\right]
    \;\ge\; \frac{(L(h_1)-L(h_0))^2}{4}
    \;=\; \frac{t^2\,L(h^\dagger)^2}{4},
  \]
  so no consistent estimation of $L(h_0)$ is possible. Thus,
  Assumption \ref{a-NPIV-ID} is a
  natural minimal identification condition for $L(h_0)$, whereas bounded
  completeness of $T$ is unnecessary for establishing the lower bound rate in Theorem \ref{thm:LB_NPIV}.
\end{rem}

\subsubsection{For Nonlinear Functional $\t_0=\G(h_0)$}

\begin{cor}[Lower Bound Rate for $\t_0=\G(h_0)$ when $h_0$ is a NPIV]\label{cor:NL_NPIV_LB}
  Let $\Gamma$ be a potentially nonlinear functional satisfying Assumption \ref{assu:DGamma_linear}. Let Assumptions \ref{a-NPIV-data}, \ref{a-NPIV-ID} (for $D\G(h_0)[\cdot]$) and \ref{ass:link} hold for the NPIV model. We have
  \begin{align*}
    \liminf_{n \to \infty}  \inf_{\tilde{\theta}_n} \sup_{h \in \Lambda_c^s(\mathcal{X})} \E_{P_{h}} \left[r^{-2}_{NPIV,n} (\tilde{\theta}_n - \G(h))^2 \right] \gtrsim  C>0
  \end{align*}
  where $r_{NPIV,n} $ is the same minimax lower bound rate given in Theorem \ref{thm:LB_NPIV}.
\end{cor}

We note that Assumption \ref{assu:DGamma_linear}(ii) and Assumption \ref{a-NPIV-ID} (for $D\G(h_0)[\cdot]$) together implies local identification of $\G(h_0)$ only. We could impose a global identification assumption of $\G(h_0)$ as follows:
\begin{assumption}[Nonlinear analogue of Assumption \ref{a-NPIV-ID}: global reduced-form identification]
  For any $h_1,h_2\in\Lambda^s_c$, if $T h_1 = T h_2$ in $L^2(P_Z)$ then $\Gamma(h_1)=\Gamma(h_2)$.
  \label{ass:NPIV-ID-global}
\end{assumption}
Assumption~\ref{ass:NPIV-ID-global} states that $\Gamma:\Lambda^s_c \mapsto \R$ is constant on each observational equivalence class
$[h]=h+\ker(T):=\{h'\in L^2(\mathcal{X})\cap \Lambda^s_c: Th'=Th \text{ in } L^2(\mathcal{Z})\}$, so $\Gamma(h)$ is point-identified from the reduced-form regression model. This Assumption \ref{ass:NPIV-ID-global} is automatically satisfied if the conditional expectation operator $T$ is bounded complete. It is interesting to note that, for the minimax lower bound rate calculation, it suffices to have the local identification of $\G(h_0)$ in a sup-norm neighborhood of $h_0\in\Lambda^s_c$.


\section{\label{sec:Minimax-UB}Rate-Optimal Estimation}

In this section we present simple sieve estimators that achieve the minimax lower bound rate of $r^*_{n}=n^{-\frac{s}{2s+d-m}}$ for $\t_0 =L(h_0),~\G(h_0),~V(h_0)$ when $h_0 \in\Lambda^s$ is a nonparametric regression function.\footnote{We could also present additional estimators to achieve the minimax lower bound rates for $\t_0$ when $h_0$ is a nonparametric density of $X$ or a nonparametric instrumental variables regression. We skip those due to the lack of space.} We impose the following regularity condition on a nonparametric regression model in this and the next sections.

\begin{assumption}[Bounded Variance of Errors]
  \label{assu:e2_moment} $\sup_{x\in{\cal X}}\E\left[\rest{\e_{i}^{2}}X_{i}=x\right]<\infty$.
  %
\end{assumption}

For any $h\in\Lambda^s$ we let  $\norm{h}_{\infty}:=\sup_{x\in\mathcal{X}} |h(x)|$ denote the sup-norm and $\norm{h}_2:=\sqrt{\E[h(X)^2]}$ denote the $L_2(X)$-norm. Here $L_2(X):=L_2 (P_X )$ denote the space of square integrable functions, which is a Hilbert space under the $L_2(X)$-norm with inner product $\inprod{h_1,h_2}_2\equiv\E\left[h_1\left(X_{i}\right)h_2\left(X_{i}\right)\right]$.

\subsection{\label{subsec:Minimax_Lin-UB}Minimax Rate-Optimal Estimation of $L(h_0)$}

We show that the lower bound $r^*_{n}=n^{-\frac{s}{2s+d-m}}$ established in Theorem \ref{thm:LB-reg} for
$\t_0=L(h_0)$ can be attained by the following plug-in sieve estimator:
\begin{equation}
  \hat{\t}:=L\left(\hat{h}\right)=\int_{{\cal M}}\hat{h}\left(x\right)w\left(x\right)d{\cal H}^{m}\left(x\right).\label{eq:IntLevel_hhat}
\end{equation}
Hence we conclude that the rate $r^*_n$ is minimax rate optimal, and plug-in estimator with sieve first stage attains this optimal rate.\footnote{In Appendix \ref{subsec:Kernel}, we show that this optimal rate is also attained by a plug-in kernel estimator.}

For the sake of concreteness, let ${\H}_{K}$ denote the closed linear span (clsp) in $L_2(X)$-norm of a sieve basis functions $b^{K}(x):=\left(b_{1}^{K}(x),...,b_{K}^{K}(x)\right)^{'}$. By default we use a tensor product basis (see, e.g., \cite{chen2007sieve}): $b^{K}(x)=\text{vec}\left(\otimes_{\ell=1}^{d}\left(b_{1}^{K}\left(x_{\ell}\right),...,b_{J_{n}}^{K}\left(x_{\ell}\right)\right)\right)$ so $K=J^d$. The sieve dimension is chosen to increase with $n$, with $J\equiv J_{n}\upto\infty$ and thus $K \equiv K_{n}\upto\infty$ as $n\to\infty$, though we will suppress
the subscript $n$ in $J$ and $K$ for notational simplicity. Let $G:=\E\left[b^{K}\left(X_{i}\right)b^{K}\left(X_{i}\right)^{'}\right]$ denote the (population) Gram matrix with its minimum eigenvalue $\lambda_{min} (G)>0$. Let $\ol b^{K}\left(x\right):=G^{-1/2} b^K(x)$ denote the orthonormalized vector of basis functions of $b^K$, then ${\H}_{K}=clsp\left (b^K\right )=clsp \left(\ol b^K \right)$. The sieve LS estimator $\hat{h}$ for $h_0=\E[Y|X]$ is given by
\begin{equation}\label{eq:seriesLS}
  \hat{h}\left(x\right):=b^{K_{n}}\left(x\right)^{'}\left[\hat{G}\right]^{-1}\frac{1}{n}\sum_{i=1}^{n}  b^{K_{n}}\left(X_{i}\right)Y_{i}=\ol b^{K_{n}}\left(x\right)^{'}\left[{\ol G}_n\right]^{-1}\frac{1}{n}\sum_{i=1}^{n} \ol b^{K_{n}}\left(X_{i}\right)Y_{i},
\end{equation}
where $\hat{G}:=\frac{1}{n}\sum_{i=1}^{n} b^{K_{n}}\left(X_{i}\right) b^{K_{n}}\left(X_{i}\right)^{'}$ and ${\ol G}_n:=\frac{1}{n}\sum_{i=1}^{n} \ol b^{K_{n}}\left(X_{i}\right) \ol b^{K_{n}}\left(X_{i}\right)^{'}$ denote the empirical Gram matrices. Let
\[
  h_{0,n}:=\arg\min_{h\in \H_{K_{n}}}\norm{h-h_0}_2 = \ol b^{K}\left(x\right)^{'}\E[\ol b^{K}\left(X_i\right)h_0(X_i)]= \ol b^{K}\left(x\right)^{'}\E[\ol b^{K}\left(X_i\right)Y_i]
\]
be the population LS projection of $h_0$ onto the sieve space ${\H}_{K}=clsp \left(\ol b^K \right)$. Following \citet*{chen2014sieve}, we define a \textit{sieve Riesz representer} $v_{K_n}^{*}\in{\H}_{K}$ as follows:\footnote{Since the minimax lower bound rate $r^*_n$ for estimating $\t_0=L(h_0)$ is slower than $n^{-1/2}$, there does not exist a Riesz representer for $\t_0=L(h_0)$ in $L^2(X)$. Following \citet*{chen2014sieve}, a sieve Riesz representer for a linear functional is always well-defined and has a closed-form expression in any finite dimensional linear sieve space ${\cal V}_{K_n}$. The sieve Riesz representer is not unique as it depends on the choice of the linear sieve space, and is used for the variance characterization for $L(\hat{h})$ as well as the sieve influence function. Since we are presenting a linear sieve plug-in estimator $L(\hat{h})$ for $\t_0=L(h_0)$ in this Section, we simply use the same linear sieve space ${\cal V}_{K_n}={\H}_{K_n}$ for our sieve variance characterization.}
\begin{equation}\label{eq:sieve_riesz}
  L\left[\nu\right]:=\rest{\frac{dL\left(h_0 +t\nu\right)}{dt}}_{t=0}=\inprod{v_{K_n}^{*},\nu}_2\equiv\E\left[v_{K_{n}}^{*}\left(X_{i}\right)\nu\left(X_{i}\right)\right],~~\text{ for all }~\nu\in {\H}_{K_n}.
\end{equation}
Moreover, it can be solved in closed form as:
\begin{align}
  v_{K_{n}}^{*}\left(\cd\right) & = b^{K_{n}}\left(\cd\right)^{'}G^{-1}L\left[ b^{K_{n}}\right]=\ol b^{K_{n}}\left(\cd\right)^{'}L\left[\ol b^{K_{n}}\right],\label{Sieve-Riesz}
\end{align}
where $L\left[\ol b^{K_{n}}\right]:=\left(L\left[\ol b_1^{K_{n}}\right],..., L\left[\ol b_{K_n}^{K_{n}}\right]\right)^{'}$. Define the sieve variance as
\begin{equation}
  \norm{v^*_{K_n}}_{sd}^2:=\E\left[\left[ v^*_{K_n}(X_i)(Y_i-h_0(X_i))\right]^2\right]\equiv \E\left[\left[ v^*_{K_n}(X_i)\e_i\right]^2\right].\label{Sieve-sd}
\end{equation}

The following lemma establishes the rates at which $\norm{v_{K_{n}}^{*}}_{2}$ and $\norm{v_{K_{n}}^{*}}_{sd}$ grow to infinity. Its proof utilizes the decomposition \eqref{eq:PieceLebInt}  in Appendix \ref{subsec:Decomp}, which converts the $m$-dimensional Hausdorff integral to a finite sum of Lebesgue integrals on $\R^{m}$.

\begin{lem}[Growth Rate of Sieve Riesz Representer for $L(h_0)$]
  \label{lem:NormRate_d-m}Under Assumptions \ref{assu:RegLevelSet}, \ref{assu:density_x}, \ref{assu:w} and $\lambda_{min} (G) >0$ for each $K$, we have:
  \begin{equation}
    \norm{v_{K_{n}}^{*}}_{2}^{2}=\norm{L\left[\ol b^{K_{n}}\right]}^{2}\equiv\sum_{k=1}^{K_{n}}\left(L\left[\ol b_{k}^{K_{n}}\right]\right)^2\asymp K_{n}^{\frac{d-m}{d}}=J_{n}^{d-m}.\label{eq:NormRate}
  \end{equation}
  Further, under Assumption \ref{assu:e2_below} and \ref{assu:e2_moment}, we have:
  \[
    \norm{v^*_{K_n}}_{sd}^2\asymp\norm{v_{K_{n}}^{*}}_{2}^{2}\asymp K_{n}^{\frac{d-m}{d}}.
  \]
\end{lem}

\noindent Lemma \ref{lem:NormRate_d-m} is the core result that demonstrates the rate acceleration provided by
the submanifold integral. It shows that the growth rates of the norm
of the sieve Riesz representer and of the sieve variance are asymptotically proportional to $J^{(d-m)}$,
where the exponent $(d-m)$ is the codimension of the $m$-dimensional
manifold, but, importantly, not $d$, the dimension of the ambient
space $\R^{d}$ in which the manifold is embedded. In other words,
even though the dimensionality of the first-stage estimation
of $h_{0}$ is $d$, it is reduced by the integration over the $m$-dimensional
submanifold.

For any $h\in \Lambda^s(\mathcal{X})$, we let $P_{K,n}h$ be the empirical LS projection of $h$ onto the sieve space ${\H}_{K}$, which is given by
$P_{K,n}h:= b^{K}\left(x\right)^{\prime}\hat{G}^{-1}\frac{1}{n}\sum_{i=1}^{n} b^{K}\left(X_{i}\right)h\left(X_{i}\right)$
with the sup operator norm defined by
$\norm{P_{K,n}}_{\infty}:=\sup_{h:0<\norm h_{\infty}<\infty}\frac{\norm{P_{K,n}h}_{\infty}}{\norm h_{\infty}}$.
Let $\zeta_{K}:=\sup_{x\in{\cal X}}\norm{ b^{K}\left(x\right)}$. We next impose some mild conditions on the linear sieve $b^{K}$ that are satisfied by some commonly used sieve bases such as B-splines and wavelets; see, e.g., \cite{chen2015optimal}.

\begin{assumption}[Conditions on Linear Sieves]
  \label{assu:Cond_Sieve} The linear sieve $b^{K}$ satisfies
  (i) $\lambda_{min} (G)\geq const. >0$ uniformly in $K$, (ii) $\zeta_{K}=O\left(\sqrt{K}\right)$, (iii) $K\log K/n=o(1)$,
  (iv) $\norm{P_{K,n}}_{\infty}=O_{p}\left(1\right)$,
  and (v) $\inf_{h\in {\H}_{K}}\norm{h-h_{0}}_{\infty}\leq K^{-s/d}$
  for $h_{0}\in\Lambda^{s}\left(\mathcal{X}\right)$.
\end{assumption}

\begin{thm}[Convergence Rate under Sieve First Stage]
  \label{thm:Rate_Sieve} Let Assumptions \ref{assu:RegLevelSet}, \ref{assu:density_x}, \ref{assu:w}, \ref{assu:e2_below}, \ref{assu:e2_moment} and \ref{assu:Cond_Sieve}(i)(ii)(iii) hold. Then: (1)
  \[
    \frac{\sqrt{n}L\left(\hat{h}-P_{K_n,n}h_0\right)}{\norm{v_{K_{n}}^{*}}_{sd}} = \frac{1}{\sqrt{n}}\sum_{i=1}^n \frac{v_{K_n}^*(X_i)}{\norm{v_{K_{n}}^{*}}_{sd}}\e_i +o_p(1)=O_p (1),~~~\text{with}~~~\norm{v^*_{K_n}}_{sd}\asymp \sqrt{K_{n}^{\frac{d-m}{d}}}.
  \]
  (2) In addition, if Assumption \ref{assu:Cond_Sieve}(iv)(v) holds, then:
  \[
    \hat{\t}-\t_{0}\equiv L\left(\hat{h}-h_{0}\right)=O_{p}\left(\sqrt{K_{n}^{\frac{d-m}{d}}/n}+K_{n}^{-s/d}\right).
  \]
  Further, if $s>m/2$, then by setting $K_{n}^{*}\asymp n^{\frac{d}{2s+d-m}}$, we obtain
  \[
    \hat{\t}-\t_{0}=O_{p}\left(n^{-\frac{s}{2s+d-m}}\right)=O_{p}\left(r^*_{n}\right),
  \]
  which attains the lower bound rate $r^*_{n}$ in Theorem \ref{thm:LB-reg}
  and is thus minimax rate-optimal.
\end{thm}

Theorems \ref{thm:LB-reg} and \ref{thm:Rate_Sieve} together demonstrate that the problem of estimating linear integrals over an $m$-dimensional submanifold is akin to the $(d-m)$-dimensional nonparametric regression problem in terms of convergence rates. Intuitively, the integration over the $m$-dimensional submanifold effectively ``aggregates out'' those $m$ dimensions, leaving only $d-m$ effective dimensions in the nonparametric estimation problem.

In many economic applications as discussed in Section \ref{sec:Examples}, the level set function $g$ is scalar-valued. Correspondingly the level set submanifold is of dimension $m=d-1$ and co-dimension $d-m = 1$. Here all but one dimensions are ``aggregated out'' and the minimax optimal convergence rate becomes $L(h_0)$ is $n^{-s/(2s+1)}$,  the same rate as a 1-dimensional nonparametric regression problem. This illustrates the (potentially) significant ``dimension reduction'' achieved through integration.

%

\begin{rem}
  In the proof of Theorem \ref{thm:LB-reg} on the lower bound $r^*_n$ fo $L(h_0)$, it suffices to assume that the density $p_0$ of $X\in \mathcal{X}\subset \R^d$ is known, and also that level set function is smooth $g\in \L^{s_1}({\cal X})$ for some $s_1 \geq \max\{1,s\}$ so that Assumption \ref{assu:RegLevelSet} is satisfied. Theorem \ref{thm:Rate_Sieve} provides a feasible sieve estimator for $L(h_0)$ that attains the minimax lower bound rate of $r^*_n$ when $p_0$ is unknown and without imposing the extra smoothness condition $g\in \L^{s_1}({\cal X})$. In Appendix Subsection \ref{subsubsec:oracle_known_fx} we present an oracle estimator that uses the extra information of a known density $p_0$ and the extra smooth submanifold condition $g\in \L^{s_1}({\cal X})$. It is interesting that the oracle estimator attains the same lower bound rate of $r^*_n$ as our feasible sieve estimator presented here. This means that in terms of minimax optimal rate of estimating $L(h_0)$, the extra information of knowing/using the marginal density of $X$ as well as the smoothness of submanifold function $g$ does not improve the speed of convergence. This is consistent with the fact that the optimal rate $r^*_n=n^{-s/(2s+d-m)}$ only depends on the smoothness of $h_0$ and the codimension $c:=d-m$ of the submanifold $\mathcal{M}$.
\end{rem}

\subsection{Minimax Rate-Optimal Estimation of $\G(h_0)$}\label{subsec:Nlh-UB}

In general, whether the lower bound rate $r_n^*=n^{-s/(2s + d - m)}$ can be attained by any consistent estimator for the nonlinear integral functional $\G(h_0)$ depends on the specific functional form of $\G$. Here we consider a specific but still quite general class of submanifold integrals of nonlinear transformations of $h_0$. We verify Assumption \ref{assu:DGamma_linear} in this context, and provide lower-level conditions for the attainability of the lower bound rate $r_n^*$, which then implies minimax optimality of the rate $r_n^*$ for the estimation of $\G(h_0)$.
Specifically, consider $\G$ defined by
\begin{equation}
  \G\left(h\right):=\int_{{\cal M}}\phi\left(h\left(x\right),x\right)w\left(x\right)d{\cal H}^{m}\left(x\right),\label{eq:Int_NLh}
\end{equation}
where  $\phi\left(t,x\right)$ is a known nonlinear transformation that is $L$-Lipschitz in $t$ (so that
  its partial derivative $\phi_{1}\left(t,x\right):=\frac{\p}{\p t}\phi\left(t,x\right)$
is well-defined almost everywhere and uniformly bounded by the Lipschitz constant $L$ whenever well-defined). We assume that $\phi_{11}\left(t,x\right):=\frac{\p^2}{\p t^2}\phi\left(t,x\right)$ is well-defined.

We now provide an umbrella theorem for $\G(h)$, which relates it to the linear submanifold integral we studied earlier, and establishes the attainability of the linear minimax optimal rate $r_n^*$ under appropriate smoothness requirement.

\begin{assumption}[Regularity Conditions for Submanifold Integral of Nonlinear Transformations]\label{assu:RegCond_Nlh}
  (a) $\phi(\cdot,x)$ is Lipschitz uniformly in $x$,
  (b) $\phi_{1}\left(t,x\right)$ is Lipschitz, and
  (c) Assumption \ref{assu:w} holds.
\end{assumption}
\begin{lem}[Pathwise Derivative of $\G$]\label{lem:NLh-DG}
  Let $\t_0=\G(h_0)$ be defined in \eqref{eq:Int_NLh} with a known $m$-dimensional submanifold ${\cal M}$ satisfying Assumption \ref{assu:RegLevelSet}. If Assumption \ref{assu:RegCond_Nlh}(a) holds, then Assumption~\ref{assu:DGamma_linear} holds with
  \begin{equation}\label{eq:PathD_NL}
    D\G(h_0)[v] = \int_{\cM} v(x)\, \phi_1(h_0(x),x)\, w(x)\, d\cH^m(x).
  \end{equation}
\end{lem}

\begin{thm}[Minimax Rate Optimal Estimation of $\G(h_0)$]
  \label{thm:Int_NLh_rate} Let $\t_0=\G(h_0)$ be defined in \eqref{eq:Int_NLh} with a known $m$-dimensional submanifold ${\cal M}$. Let Assumptions \ref{assu:RegLevelSet}, \ref{assu:density_x}, \ref{assu:e2_below}, \ref{assu:e2_moment}, \ref{assu:Cond_Sieve} and \ref{assu:RegCond_Nlh} hold.
  \begin{itemize}
    \item[(1)] If $s > m/2$, then $\hat{\theta}_{SS}-\t_0=O_p(r_n^*)$, where $\hat{\theta}_{SS}$ is a Split-Sample estimator:
      \begin{align}
        \hat{\theta}_{SS} := \G\left(\bar{h}\right) - \frac{1}{8} \int_{\mathcal{M}} \phi_{11}\left(\bar{h}(x),x\right) \left(\hat{h}_1(x) - \hat{h}_2(x)\right)^2 w(x) d\cH^m(x),
        \label{eq:split_estimator}
      \end{align}
      where $\hat{h}_1$ and $\hat{h}_2$ are split-sample sieve estimators with sieve dimension $K_{n}^{*}\asymp n^{\frac{d}{2s+d-m}}$, and $\bar{h} := (\hat{h}_1 + \hat{h}_2)/2$.
    \item[(2)] If $s > m/2$, then $\hat{\theta}_{LOO}-\t_0=O_p(r_n^*)$, where $\hat{\theta}_{LOO}$ is a Leave-One-Out estimator:
      \begin{equation}\label{eq:Gamma_LOO_estimator}
        \hat\t_{\mathrm{LOO}}
        := \G(\hat h)
        -\frac{1}{2n^2}\sum_{i=1}^n\int_{\cM}
        \phi_{11}\!\left(\hat h(x),x\right)\,\hat s_i(x)^2\,\big(\hat\e_i^{(-i)}\big)^2\,
        w(x)\,d\cH^m(x),
      \end{equation}
      where $\hat h$ is the sieve LS estimator \eqref{eq:seriesLS} with sieve dimension $K_{n}^{*}\asymp n^{\frac{d}{2s+d-m}}$; for $i$-th observation, $\hat s_i(x):= b^{K_n}(x)^\prime \hat G^{-1} b^{K_n}(X_i)$,
      and $\hat\e_i^{(-i)} := Y_i - \hat h^{(-i)}(X_i)$ is the LOO residual with the LOO sieve LS estimator $\hat h^{(-i)}$ computed using the sample $\{(Y_j,X_j):j\neq i\}_{j=1}^n$.
    \item[(3)] If $s\geq m$, then $\hat{\theta}-\t_0=O_p(r_n^*)$, where $\hat{\theta}=\G(\hat{h})$ is a plug-in sieve LS estimator with sieve dimension $K_{n}^{*}\asymp n^{\frac{d}{2s+d-m}}$.
  \end{itemize}
  In either (1) or (2) or (3), the rate $r_n^*=n^{-s/(2s+d-m)}$ is minimax-optimal.
\end{thm}

Both the SS and the LOO sieve debiased estimators are constructed by subtracting the diagonal parts of the quadratic remainder in the Taylor expansion of $\G(\hat h)$ at $h_0$. Both estimators are shown to attain the minimax lower bound rate $r_n^*$.
We note that the split-sample debiased sieve estimator $\hat{\t}_{SS}$ for $\G(h_0)$ is used to attain an upper bound rate that matches the lower bound rate $r_n^*$ under the mild smoothness condition $s > m/2$.
The plug-in sieve estimator $\hat{\theta}=\G\left(\hat{h}\right)$ is extremely easy to compute, but attains the minimax optimal rate only under stronger smoothness condition $s\geq m$.

Theorem \ref{thm:Int_NLh_rate} applies to quadratic submanifold functional $Q$ as a special case, where
\begin{align}
  \G(h)= Q(h) := \int_{\cM} h^2(x) w(x) d\mathcal{H}^m(x),
\end{align}
and the split-sample debiased estimator for $Q(h_0)$ becomes
$\hat{\t}^Q_{SS}=\int_{{\cal M}}\hat{h}_{1}\left(x\right)\hat{h}_{2}\left(x\right)w\left(x\right)d{\cal H}^{m}\left(x\right)$.
The study of $Q(h_0)$ is of conceptual importance given that it represents the simplest type of nonlinear integral functionals: when the standard Taylor expansion of a nonlinear function applies, the quadratic term in the expansion often constitutes the first-order approximation of the form of nonlinearity. For this reason, (full-dimensional) quadratic functionals have long been studied in statistics and econometrics as a leading example of nonlinear integral functionals: for a few examples out of many, see \cite{bickel1973some}, \cite{bickel1988estimating}, \cite{fan1991estimation}, \cite{efromovich1996optimal}, \cite{cai2005nonquadratic} for earlier work in statistics on quadratic functionals of nonparametric density and regression functions, as well as \cite{chen2018optimal} and \cite{breunig2019simple} for more recent work on quadratic functionals of NPIV structural functions in econometrics. It has also been noted that the estimation of a quadratic functional is tightly related to hypothesis testing; see, for example, \cite{cai2005nonquadratic,cai2006adaptive} in the statistics literature, and \cite{breunig2024adaptive} in the adaptive minimax testing for NPIV models.

\subsection{Minimax Rate-Optimal Estimation of $V(h_0)$}\label{subsec:Uch_rate-UB}

We now turn to another specific type of nonlinear integral functional, defined as an integral on the upper contour set of $h_0$, i.e.
\begin{equation}
  V\left(h_0\right):=\int_{\left\{x\in\mathcal{X}:~ h_0\left(x\right)\geq0\right\} }w\left(x\right)dx,\label{eq:Int_UCh}
\end{equation}
where $h_{0}:\R^{d}\to\R$ is an unknown (scalar-valued) function
and $w\left(x\right)$ is known.

Even though here $V$ is a full-dimensional (Lebesgue) integral per se, its pathwise derivative w.r.t. $h$ becomes a submanifold integral over the level set $\{x\in\mathcal{X}:~h_0(x)=0\}$ as we show below. Hence, in this case the ``level set function''
$g=h_{0}$ is taken to be unknown and needs to be nonparametrically estimated. Notice that such ``upper contour integrals'' show up in several examples in Section \ref{sec:Examples} that feature subpopulations defined through inequalities, such as the welfare/value under optimal treatment assignment. We provide a corresponding theorem for upper contour integrals below.

\begin{assumption}[Regularity Conditions for $V(h_0)$]\label{assu:RegCond_Uch} (a) $h_{0}$ is continuously differentiable and $\norm{\Dif_{x}h_{0}\left(x\right)}\geq\ul{\e}>0$ on the level set  ${\cM}_0:=\left\{ x\in{\cal X}:h_{0}\left(x\right)=0\right\}$, (b) $w$, $\Dif_{x} w$, and the second-order (partial) derivatives of $h_{0}$ are all uniformly bounded.
\end{assumption}
\begin{lem}[Pathwise Derivative of $V$]\label{lem:UCh-DV}
  Let $\t_0=V(h_0)$ be defined in \eqref{eq:Int_UCh}. If Assumption \ref{assu:RegCond_Uch}(a) holds, then ${\cal M}_0$ is an $m=(d-1)$ dimensional submanifold satisfying Assumption \ref{assu:RegLevelSet}, and Assumption \ref{assu:DGamma_linear} holds with
  \begin{equation}
    DV(h_0)\left[v\right]=\int_{\cM_0}v\left(x\right)\frac{w\left(x\right)}{\norm{\Dif_{x}h_{0}\left(x\right)}}d{\cal H}^{d-1}\left(x\right).\label{eq:PathD_Contour}
  \end{equation}
\end{lem}
We note that the calculation of the pathwise derivative \eqref{eq:PathD_Contour}
is nonstandard, since it involves differentiation w.r.t. the changing boundary of the region of integration.

\begin{thm}[Minimax Rate Optimal Estimation of $V(h_0)$]
  \label{thm:Int_UCh_rate} Let $\t_0=V(h_0)$ be defined in \eqref{eq:Int_UCh}. Let Assumptions \ref{assu:density_x}, \ref{assu:e2_below}, \ref{assu:e2_moment}, \ref{assu:Cond_Sieve} and \ref{assu:RegCond_Uch} hold.
  \begin{itemize}
    \item[(1)]  If $s\geq m/2 +1=(d+1)/2$, then $\hat{\theta}_{SS}-\t_0=O_p(r_n^*)$, where $\hat{\theta}_{SS}$ is a split-sample debiased estimator
      \begin{align}
        \hat{\theta}_{SS} := V\left(\bar{h}\right) - \frac{1}{8} D^2V(\bar{h})\left[\hat{h}_1 - \hat{h}_2, \hat{h}_1 - \hat{h}_2\right],
      \end{align}
      where $D^2V$ is the second functional derivative defined in Lemma \ref{lem:PathD_Level} in the Appendix, $\hat{h}_1$ and $\hat{h}_2$ are split-sample sieve estimators with sieve dimension $K_{n}^{*}\asymp n^{\frac{d}{2s+1}}$, and $\bar{h} := (\hat{h}_1 + \hat{h}_2)/2$.
    \item[(2)] If $s\geq m/2 +1=(d+1)/2$, then $\hat{\theta}_{LOO}-\t_0=O_p(r_n^*)$, where $\hat{\theta}_{LOO}$ is a Leave-One-Out debiased estimator:
      \begin{equation}\label{eq:V_LOO_estimator}
        \hat\t_{\mathrm{LOO}}
        := V(\hat h)
        -\frac{1}{2n^2}\sum_{i=1}^n D^2V(\hat h)\big[\hat s_i,\hat s_i\big]\cdot \big(\hat\e_i^{(-i)}\big)^2,
      \end{equation}
      where $\hat{h}$, $\hat s_i$ and $\hat\e_i^{(-i)}$ are defined in Theorem \ref{thm:Int_NLh_rate}(2) with sieve dim $K_{n}^{*}\asymp n^{\frac{d}{2s+1}}$.
    \item[(3)] If $s\geq m+1=d$, then $\hat{\theta}-\t_0=O_p(r_n^*)$, where $\hat{\theta}=V(\hat{h})$ is a plug-in sieve estimator with sieve dimension $K_{n}^{*}\asymp n^{\frac{d}{2s+1}}$.
  \end{itemize}
  In either (1) or (2) or (3), the rate $r_n^*=n^{-s/(2s+d-m)}=n^{-s/(2s+1)}$ is minimax optimal.
\end{thm}

\section{\label{sec:Nonlinear} Sieve Inference on Submanifold Integrals}

The minimax lower bound results from Section \ref{sec:Minimax-LB} imply that our functionals of interest, $\t_0=L(h_0),~\G(h_0),~V(h_0)$ are all \textit{irregular} functionals of the unknown function $h_0$, where $h_0$ could be a nonparametric regression, a nonparametric density and a NPIV function. Consequently, the well-known $L^2(P_X)$-Riesz representers and the efficient influence functions \textit{do not exist} for these functionals. Nevertheless, applying the general theories in \citet*{chen2014sieve,chen2014sieveM} for sieve M-estimation and in \cite*{chenpouzo2015sieve} for penalized sieve MD-estimation,\footnote{M-estimation includes nonparametric regression, quantile regression and density estimation as special cases. MD-estimation includes M-estimation, nonparametric IV regression and nonparametric quantile IV as special cases.} the sieve Riesz representers and sieve influence functions are still well-defined for these irregular functionals. The sieve Riesz representer approach enables us to show that the sieve student t statistics for these irregular functionals are still asymptotically standard normal, which can be used to construct simple confidence intervals for $\t_0$.

For the sake of simplicity, in this section we specialize their general theory to our functionals of interest when $h_0$ is a nonparametric regression. In this section we denote $\t_0:=\Phi(h):\Lambda^s\mapsto\R$, where $\Phi(h_0)=L(h_0),~\G(h_0),~V(h_0)$. By exploring the submanifold structures, we can characterize the
growth rate of the norm of the sieve Riesz representer for the first-order expansions of $\G(h_0)$ and $V(h_0)$, provide a tighter control of nonlinear remainder terms, and establish the asymptotic normality under mild sufficient conditions.

\subsection{Asymptotic Normality via Sieve Riesz Representers}

Let $\H_{K_{n}}:=\left\{ h\in{\cal B}_{K_{n}}:\norm{h-h_{0}}_{\infty}<\ul{\e}\right\} $. Let $\hat{h}_n \in \H_{K_n}$ be any sup-norm consistent sieve estimator of the regression function $h_0\in\Lambda^s(\mathcal{X})$. Let ${\cal V}_{K_n}$ denote a $K_n$-dimensional Hilbert space that contains $h_{0,n}:=\arg\min_{h\in \H_{K_{n}}}\norm{h-h_0}_2$, and can be equivalently expressed as $clsp\left(\ol \psi^{K_n}\right):=clsp\{\psi_k:k=1,...,K_n\}$, where $\{\psi_k:k\geq 1\}$ is an orthonormal basis for $L_2(P_X)$. Following \citet*{chen2014sieve,chen2014sieveM}, we define the sieve Riesz representer $v_{K_n}^{*}$
and the sieve variance $\norm{v^*_{K_n}}_{sd}^2$ in the same way as in \eqref{eq:sieve_riesz} and \eqref{Sieve-sd}.

The following Remark follows from Lemma \ref{lem:NormRate_d-m}, which allows us to control the growth rate of the sieve Riesz representer and the sieve variance for the nonlinear submanifold functionals $\Phi=\G$ and $V$.

\begin{rem}[Growth Rate of Sieve Riesz Representer for $D\Phi(h_0)$]
  \label{rem:NormRate_nonlinear}Under Assumptions \ref{assu:RegLevelSet}, \ref{assu:density_x}, \ref{assu:DGamma_linear} and $\lambda_{min} \left( R_{K}\right) >0$ for each $K$, we have:
  \begin{equation}
    \norm{v_{K_{n}}^{*}}_{2}^{2}=\norm{D\Phi(h_0)\left(\ol \psi^{K_{n}}\right)}^{2}\asymp K_{n}^{\frac{d-m}{d}}=J_{n}^{d-m}.\label{eq:NormRate}
  \end{equation}
  Further, if Assumptions \ref{assu:e2_below} and \ref{assu:e2_moment} hold, then: $\norm{v_{K_{n}}^{*}}_{sd}^2\asymp \norm{v_{K_{n}}^{*}}_{2}^{2} \asymp K_{n}^{\frac{d-m}{d}}$.
\end{rem}

The following is a standard Lindeberg condition for CLT.
\begin{assumption}[Lindeberg]\label{assu:Lindeberg} $\sup_{x\in{\cal X}}\E\left[\rest{\e_{i}^{2}\left\{ \left|\e_{i}\right|>c\right\} }X_{i}=x\right]\to0$
  as $c\to\infty$.
\end{assumption}

In the following, we let $\hat{\t}=\Phi(\hat{h})$ be the plug-in sieve estimator of $\t_0=\Phi(h_0)$. We let $\hat{\t}_{SS}$ and $\hat{\t}_{LOO}$ respectively denote the Split-Sample sieve estimator and Leave-One-Out sieve estimator of $\t_0=\Phi (h_0)$ when $\Phi=\G,~V$.

\begin{prop}[Asymptotic Normality]\label{prop:Plug-in_NL} Let $\hat{h}$ be the sieve LS estimator of $h_0$. Let Assumptions \ref{assu:density_x}, \ref{assu:e2_below}, \ref{assu:e2_moment}, \ref{assu:Cond_Sieve} and \ref{assu:Lindeberg} hold. Let the sieve dimension $K_{n}$ satisfy $K_{n}\asymp K^*_n [\log(n)]^{\kappa}$ for some $\kappa \in (0,1)$ and $K_{n}^{*}\asymp n^{\frac{d}{2s+d-m}}$.
  Then:
  \[
    \frac{\sqrt{n}\left(\tilde{\Phi}-\Phi(h_0)\right)}{\norm{v_{K_{n}}^{*}}_{sd}} = \frac{1}{\sqrt{n}}\sum_{i=1}^n \frac{v_{K_n}^*(X_i)}{\norm{v_{K_{n}}^{*}}_{sd}}\e_i +o_p(1)\dto\cN\left(0,1\right)\quad\text{with }\norm{v_{K_{n}}^{*}}_{sd}\asymp\sqrt{K_{n}^{\frac{d-m}{d}}},
  \]
  provided the following conditions hold for the specific functional forms of $\Phi:\Lambda^s \mapsto \R$:
  \begin{itemize}
    \item For $\Phi=L$: Assumptions \ref{assu:RegLevelSet} and \ref{assu:w} hold with $s>m/2$ for $\tilde{\Phi}=L(\hat{h})$.
    \item For $\Phi=\G$: Assumptions \ref{assu:RegLevelSet} and \ref{assu:RegCond_Nlh} hold, $s>m/2$ for $\tilde{\Phi}=\hat{\t}_{SS},~\hat{\t}_{LOO}$; and $s>m$ for $\tilde{\Phi}=\G(\hat{h})$.
    \item For $\Phi =V$: Assumption \ref{assu:RegCond_Uch} holds, $s>(d+1)/2$ for $\tilde{\Phi}=\hat{\t}_{SS},~\hat{\t}_{LOO}$; and $s>d$ for $\tilde{\Phi}=V(\hat{h})$.
  \end{itemize}
\end{prop}

\subsection{Inference based on Estimated Sieve Riesz Representers}

We note that $v_{K_n}^{*}\left(\cd\right)$ can be estimated by
\[
  \hat{v}_{K_{n}}^{*}\left(\cd\right):=\ol \psi^{K_{n}}\left(\cd\right)'(\ol {R}_{K_{n}})^{-1}D\Phi\left(\hat{h}\right)\left[\ol \psi^{K_{n}}\right],~~~\text{with}~~\ol {R}_{K_n} :=\frac{1}{n} \sum_i{\ol \psi}^{K_{n}}(X_i){\ol  \psi}^{K_{n}}(X_i)'.
\]
Denote $\sigma^2_{*,K_{n}} :=\norm{v_{K_{n}}^{*}}_{sd}^2$ and $\widehat{\e}_i:=Y_{i}-\hat{h}\left(X_{i}\right)$. We can estimate the sieve variance as\footnote{In small samples,it might be more accurate to estimate the sieve variance using $\widehat{\sigma}^2_{*,K_{n}}=\frac{1}{n}\sum_{i=1}^n \left[ \hat{v}_{K_{n}}^{*}(X_{i})\widehat{\e}_i - \frac{1}{n-1}\sum_{j\neq i,j=1}^n\hat{v}_{K_{n}}^{*}(X_{j})\widehat{\e}_j\right]^{2}$.}
\begin{align}\label{sieveVar-hat}
  \widehat{\sigma}^2_{*,K_{n}}  =\frac{1}{n}\sum_{i=1}^n \left[ \hat{v}_{K_{n}}^{*}(X_{i})\widehat{\e}_i\right]^{2}=\left(D\Phi\left(\hat{h}\right)\left[\psi^{\left(K\right)}\right]\right)'\widehat{\O}_{K_n}\left(D\Phi\left(\hat{h}\right)\left[\psi^{\left(K\right)}\right]\right)
\end{align}
with
$\widehat{\O}_{K_n}:= \ol {R}_{K_n}^{-1}\left(\frac{1}{n}\sum_{i=1}^{n}(\widehat{\e}_i)^{2}{\ol \psi}^{K_{n}}(X_i){\ol  \psi}^{K_{n}}(X_i)'\right) \ol {R}_{K_n}^{-1}$.
The consistency of the sieve variance estimator $\widehat{\sigma}^2_{*,K_{n}}$ can be easily established under the following assumption.

\begin{assumption}[Higher Error Moments]\label{assu:e2d_moment} $\E\left[\abs{\e_i}^{2+\d}\right] < \infty$ for some $\d>2d/(2s-m)$.
\end{assumption}

\begin{prop}[Sieve t statistics]\label{prop:VarConsis}
  Let Assumptions for Proposition \ref{prop:Plug-in_NL} and Assumption \ref{assu:e2d_moment} hold. Let $K_{n}\asymp K^*_n [\log(n)]^{\kappa}$ for some $\kappa \in (0,1)$ and $K_{n}^{*}\asymp n^{\frac{d}{2s+d-m}}$. Then:
  \[
    \abs{\frac{\widehat{\sigma}_{*,K_{n}}}{\norm{v_{K_{n}}^{*}}_{sd} } - 1} =o_p(1);~~~
    \frac{\sqrt{n}\left(\tilde{\Phi}-\t_0\right)}{\widehat{\sigma}_{*,K_{n}}} = \frac{1}{\sqrt{n}}\sum_{i=1}^n \frac{v_{K_n}^*(X_i)}{\norm{v_{K_{n}}^{*}}_{sd}}\e_i +o_p(1)\dto\cN\left(0,1\right),
  \]
  where $\tilde{\Phi}$ is the estimator specified in Proposition~\ref{prop:Plug-in_NL} (i.e., $L(\hat{h})$, $\hat{\t}_{SS}$, $\hat{\t}_{LOO}$, or $\Phi(\hat{h})$, depending on the case and smoothness regime).
\end{prop}

As an application of Proposition \ref{prop:VarConsis}, we can construct valid $.95$ confidence
interval using normal critical value for $\t_0=\Phi(h_0)$ as follows:
\[
\text{CI}:=\left[\tilde{\Phi}-1.96\hat{\s}_{\t}~,~\tilde{\Phi}+1.96\hat{\s}_{\t}\right]~,~~\text{with}~~\hat{\s}_{\t}^{2}=\widehat{\sigma}^2_{*,K_{n}}/n
\]
provided that the sieve variance is computed using undersmoothed sieve dimension $K_n$, say $K_{n}\asymp K^*_n [\log(n)]^{1/2}$. In practice, we implement an adaptive inference procedure that (i) selects the sieve dimension in a data-driven way (bootstrap--Lepski), with an undersmoothed choice used for valid confidence interval; and (ii) uses a multiplier bootstrap to approximate the distribution of the studentized sieve estimator. See Appendix \ref{subsec:MB_K_LOO} for details.

\section{\label{sec:Simulation} Monte Carlo Simulations}

We present two Monte Carlo studies to check finite sample performance of the sieve estimation and inference of linear and nonlinear submanifold integrals. Subsection \ref{subsec:Sim1_Circle} considers a linear submanifold integral $\t_0=L(h_0)$, and Subsection \ref{subsec:Sim2_Disk} considers an upper contour integral $\t_0=V(h_0)$. In both examples, $h_0 (x)=\E[Y|X=x]$ is a nonparametric regression, and is estimated using a tensor-product B-spline sieve LS estimator $\hat{h}$ based on a random sample $\{(Y_i,X_i)\}_{i=1}^n$, drawn from the joint distribution of $(Y,X)$ satisfying:
\[
Y_i=h_0(X_i) + \e_{i},~~~\text{with}~~\e_{i}\sim\cN\left(0,1\right),~~X_{i1}\sim_{i.i.d.}X_{i2}\sim_{i.i.d.}Uniform\left[-2,2\right],
\]
and $\e_{i}$ is independent of the covariates $X_i=(X_{i1},X_{i2})'\in \mathcal{X}=[-2,2]^2$ with $d=2$. 

In both Monte Carlo studies, we implement multiplier bootstrap Lepski data-driven choice of sieve dimensions $K$ as described in Appendix \ref{subsec:MB_K_LOO}, with $B_{Lepski}=500$ multiplier bootstrap draws. We report simulation results over $B=1000$ Monte Carlo simulation replications, using performance metrics in terms of Monte Carlo root-mean-squared error (RMSE), bias and standard deviation (SD), as well as coverage rates and coverage length. We report the simulation results for different sample sizes $n=500, 1000, 2000, 4000, 8000$.

\subsection{\label{subsec:Sim1_Circle}Integral on Known Submanifold}

We first report simulations results for the estimation of an integral functional
of $h_0$ over a known submanifold defined by the unit
circle, i.e., $\t_{0}=L\left(h_{0}\right)$ with $L\left(h\right)=\int_{\S^{1}}h\left(x\right)d{\cal H}^{1}\left(x\right)$ and
$h_{0}\left(x\right)=x_{1}^{2}+2\sin\left(x_{1}\right)x_{2}$ with $x=(x_1,x_2)^{\prime} \in{\cal X}=\left[-2,2\right]^{2}$.
Under the transformation $x_{1}=\cos\left(\b\right)$ and
$x_{2}=\sin\left(\b\right)$, we obtain the true value
\[
  \t_{0}=L\left(h_{0}\right)=\int_{0}^{2\pi}\left(\text{cos}\left(\b\right)^{2}+2\sin\left(\cos\left(\b\right)\right)\sin\left(\b\right)\right)d\b=\pi.
\]
 In this exercise, $L(h_0)$ is a linear functional of $h_0$ on a known $m=1$ -dimensional submanifold (i.e., the unit circle $\S^{1}$). Theorem \ref{thm:Rate_Sieve} is directly applicable, and the optimal convergence rate of a spline sieve plug-in estimator $L(\hat{h}) - L(h_0)$ is given by $n^{-s/(2s+1)}$ with the optimal sieve dimension $K^* \asymp n^{2/(2s+1)}$ and $d=2,m=1$.

Given the spline LS estimator $\hat{h}$, the plug-in estimator $\hat{\t}=L(\hat{h})=\int_{\S^{1}} \hat{h}\left(x\right)d{\cal H}^{1}\left(x\right)$ is numerically computed using the sample
average over $M=5000$ Sobol sequence points\footnote{The Sobol sequence sampling, proposed by \cite{sobol1967distribution}, is a well-known quasi-random Monte Carlo sampling method that generates a deterministic sequence of points, whose distribution asymptotically converges to the uniform distribution, but achieves better finite-sample approximation of the population expectation (integral) by the sample average over $M$ Sobol points.} in the angle space $\left[0,2\pi\right]$.
We obtain the plug-in estimator $\hat{\t}$ and construct the confidence
interval $\text{CI}:=\left[\hat{\t}\pm1.96\hat{\s}_{\t}\right]$ with $\hat{\s}_{\t}^{2}=\widehat{\sigma}^2_{*,K_{n}}/n$ as in \eqref{sieveVar-hat},
where the pathwise derivative $D\Phi\left(\hat{h}\right)\left[v\right]=L(v)=\int_{\S^{1}}v\left(x\right)d{\cal H}^{1}\left(x\right)$
is computed numerically in the same manner as described above.

Table \ref{tab:Sim1_Circle_RateOpt} reports the RMSE, bias, standard deviation, and the average selected sieve dimension $\bar{\hat{K}}$ for the plug-in estimator using the Lepski sieve dimension choice $\hat{K}$. Note that the rate-optimal $\hat{K}$ is chosen to minimize estimation risk, and does not involve the undersmoothing that is needed for valid confidence interval construction.

\begin{table}[h]
  \centering
  \caption{Integral on Known Manifold: Rate-Optimal $\hat{K}$}
  \label{tab:Sim1_Circle_RateOpt}
  \begin{threeparttable}
    \begin{tabular}{rcccc}
      \toprule
      $n$ & RMSE & Bias & SD & $\bar{\hat{K}}$ \\
      \midrule
      500  & 0.4515 & -0.0133 & 0.4515 & 17.1 \\
      1000 & 0.3175 & -0.0097 & 0.3175 & 16.8 \\
      2000 & 0.2208 &  0.0059 & 0.2208 & 17.1 \\
      4000 & 0.1626 &  0.0061 & 0.1626 & 17.2 \\
      8000 & 0.1138 & -0.0042 & 0.1138 & 17.4 \\
      \bottomrule
    \end{tabular}
    \begin{tablenotes}
      \footnotesize
    \item Rate-optimal $\hat{K}$ selected by bootstrap-Lepski without undersmoothing.
    \end{tablenotes}
  \end{threeparttable}
\end{table}

Next, we turn to the inference results. Table \ref{tab:Sim1_Circle} reports the finite-sample performance of the plug-in estimator and the corresponding 95\% confidence interval, constructed with bootstrap-Lepski adaptive choice of the \emph{undersmoothed} sieve dimension $\tilde{K}$ and normal critical values. The average $\bar{\tilde{K}} \approx 51$ is substantially larger than the rate-optimal $\bar{\hat{K}} \approx 17$ in Table \ref{tab:Sim1_Circle_RateOpt}, reflecting the undersmoothing needed for valid inference.

\begin{table}[!h]
  \centering
  \caption{Integral over Known Manifold: Adapt+Normal}
  \label{tab:Sim1_Circle}
  \begin{threeparttable}
    \begin{tabular}{lcccccccc}
      \toprule
      $n$ & RMSE & Bias & SD & CI\_L & CI\_U & U--L & Coverage & $\tilde{K}$\\
      \midrule
      500  & 0.6198 & -0.0320 & 0.6193 & 1.9364 & 4.2828 & 2.3465 & 93.80 & 50.6 \\
      1000 & 0.4232 & -0.0169 & 0.4231 & 2.2952 & 3.9541 & 1.6589 & 95.20 & 50.7 \\
      2000 & 0.3050 &  0.0062 & 0.3051 & 2.5580 & 3.7376 & 1.1796 & 95.00 & 51.4 \\
      4000 & 0.2249 &  0.0078 & 0.2249 & 2.7325 & 3.5663 & 0.8338 & 94.40 & 51.4 \\
      8000 & 0.1496 & -0.0102 & 0.1493 & 2.8363 & 3.4264 & 0.5901 & 95.80 & 51.8 \\
      \bottomrule
    \end{tabular}
    \begin{tablenotes}
      \footnotesize
    \item ``Adapt+Normal'': bootstrap-Lepski sieve dimension choice and normal critical values.
    \end{tablenotes}
  \end{threeparttable}
\end{table}

In Table~\ref{tab:Sim1_Circle} and subsequent tables, we report the square root of the mean squared error (RMSE), the bias (Bias), the standard deviation (SD)\footnote{The standard deviation is calculated using the standard $\frac{1}{B-1} \sum_{b=1}^B$ formula. For this technical reason, ``SD'' can be larger than ``RMSE'', which is calculated based on the $\frac{1}{B} \sum_{b=1}^B$ formula.} of the estimator, the average lower and upper bounds of the confidence interval (CI\_L and CI\_U), the average length of the confidence interval (U-L), and the realized coverage probability of the CI (Coverage). Overall, the plug-in estimator and the corresponding CI perform very well under all five sample sizes: the RMSE shrinks (almost at $\sqrt{n}$ rate) as the sample size increases, the bias is of negligible order relative to the standard deviation, and the realized coverage probability is close to the nominal 95\% level.

Table \ref{tab:Sim1_Circle_Comp} reports the estimator RMSE and the CI coverage depending on: (i) whether the bootstrap-Lepski adaptive choice of sieve dimension or a fixed sieve dimension is used, and (ii) whether normal or bootstrap quantile critical values\footnote{The bootstrap quantiles are also computed based on $B_{Lepski} = 500$ draws.} are used, when constructing the CI. The first column ``Adapt+Normal'' is identical to that in Table \ref{tab:Sim1_Circle}, while the remaining three columns contain results about the three alternative CI constructions. We observe that all four versions produce coverage close to the 95\% nominal level.

\begin{table}[h]
  \centering
  \caption{Integral on Known Manifold: RMSE and Coverage Comparison}
  \label{tab:Sim1_Circle_Comp}
  \begin{threeparttable}
    \begin{tabular}{r|cc|cccc}
      \toprule
      \multicolumn{1}{c}{$n$ }  & \multicolumn{2}{c}{RMSE}   & \multicolumn{4}{c}{CI Coverage} \\
      \midrule
      & Adapt & Fixed & Adapt+Norm  & Fixed+Norm  & Adapt+Boot & Fixed+Boot \\
      \midrule
      500  & 0.6198 & 0.4463 & 93.8\% & 94.2\% & 93.8\% & 94.3\%\\
      1000 & 0.4232 & 0.3097 & 95.2\% & 94.7\% & 95.0\% & 94.2\%\\
      2000 & 0.3050 & 0.2160 & 95.0\% & 94.9\% & 94.7\% & 94.9\%\\
      4000 & 0.2249 & 0.1600 & 94.4\% & 94.3\% & 93.9\% & 94.2\% \\
      8000 & 0.1496 & 0.1091 & 95.8\% & 94.9\% & 95.6\% & 94.6\%\\
      \bottomrule
    \end{tabular}
    \begin{tablenotes}
      \footnotesize
    \item ``Adapt'': bootstrap-Lepski sieve dimension choice; ``Fixed'': fixed sieve dimension $K=36$.
    \item ``Norm'': normal critical values; ``Boot'': bootstrap quantile critical values.
    \end{tablenotes}
  \end{threeparttable}
\end{table}

\subsection{\label{subsec:Sim2_Disk} Integral on Estimated Upper Contour Set}

In this subsection, we analyze the estimation of an integral functional
of a nonparametric function over the (estimated) unit
disk: $\t_{0}:=V\left(h_{0}\right)$ with
\begin{align*}
  V\left(h_0 \right) & :=\int\ind\left\{ x\in{\cal X}:h_{0}\left(x\right)\geq0\right\} dx,~~{\cal X}=\left[-2,2\right]^{2}\\
  h_{0}\left(x\right) & :=\left(1-\norm x^{2}\right)\left(4+\sin\left(x_{1}\right)x_{2}+\cos\left(x_{2}\right)\right).
\end{align*}
Under the above construction, $h_{0}\left(x\right)\geq0$ if and only
if $\norm x\leq1$, and thus
$\t_{0}=\int\ind\left\{ \norm x\leq1\right\} dx=\pi$.
We compute the spline nonparametric regression estimator $\hat{h}\left(x\right)$ for $h_0(x)= \E[Y_i|X_i=x]$, and the plug-in estimator $V\left(\hat{h}\right)=\int\ind\left\{ \hat{h}\left(x\right)\geq0\right\} dx$
using the sample average over $M=50,000$ Sobol sequence points, and
$\text{CI}:=\left[\hat{\t}\pm1.96\hat{\s}_{\t}\right]$ with the sieve variance $\hat{\s}_{\t}^{2}=\widehat{\sigma}^2_{*,K_{n}}/n$ as in \eqref{sieveVar-hat},
where the pathwise derivative $DV\left(\hat{h}\right)\left[v\right]$
now takes the following form:
\[
  DV(h)\left[v\right]=\int_{\left\{x\in{\cal X}:~ h\left(x\right)=0\right\} }\frac{v\left(x\right)}{\norm{\Dif_{x}h\left(x\right)}}d{\cal H}^{d-1}\left(x\right),
\]
which we numerically approximate via
$\hat{D}V\left(h\right)\left[v\right]=\frac{1}{2\e}\int_{\left\{x\in{\cal X}:~ -\e<h\left(x\right)<\e\right\} }v\left(x\right)dx$ based on the mathematical result\footnote{See Theorem 3.13.(iii) of \citet*{evans2015measure}.}
that
\[
  \lim_{\e\downto0}\frac{1}{2\e}\int_{\left\{x\in{\cal X}:~ -\e<h\left(x\right)<\e\right\} }v\left(x\right)dx=\int_{\left\{x\in{\cal X}:~ h\left(x\right)=0\right\} }\frac{v\left(x\right)}{\norm{\Dif_{x}h\left(x\right)}}d{\cal H}^{d-1}\left(x\right).
\]
We set $\e=0.01$ in our simulation, and, given that $\left\{x\in{\cal X}:~ -\e<h\left(x\right)<\e\right\} $
may occur infrequently for small $\e$, we use the sample average
from $M=50,000$ Sobol sequence points to approximate $\hat{D}V\left(h\right)\left[v\right]$.

Table \ref{tab:Sim2_Disk_RateOpt} reports the rate-optimal estimation results for the upper contour set integral. The average rate-optimal $\bar{\hat{K}}$ grows with $n$ (from approximately 17 at $n=500$ to 36 at $n=8000$), in contrast to the stable $\bar{\hat{K}} \approx 17$ in Simulation 1. The plug-in and LOO-debiased estimators display similar RMSE at the rate-optimal $\hat{K}$ for $n\geq 1000$.

\begin{table}[h]
  \centering
  \caption{Integral on Estimated Upper Contour Set: Rate-Optimal $\hat{K}$}
  \label{tab:Sim2_Disk_RateOpt}
  \begin{threeparttable}
    \begin{tabular}{r|cccc|cccc}
      \toprule
      & \multicolumn{4}{c|}{(A) Plug-In} & \multicolumn{4}{c}{(B) Leave-One-Out} \\
      \midrule
      $n$ & RMSE & Bias & SD & $\bar{\hat{K}}$ & RMSE & Bias & SD & $\bar{\hat{K}}$ \\
      \midrule
      500  & 0.0652 & 0.0351 & 0.0550 & 17.3 & 1.0175 & 0.0148 & 1.0179 & 17.1 \\
      1000 & 0.0483 & 0.0270 & 0.0402 & 18.0 & 0.0488 & 0.0275 & 0.0403 & 17.8 \\
      2000 & 0.0358 & 0.0182 & 0.0309 & 21.7 & 0.0364 & 0.0192 & 0.0310 & 21.4 \\
      4000 & 0.0253 & 0.0104 & 0.0231 & 28.2 & 0.0254 & 0.0106 & 0.0230 & 28.3 \\
      8000 & 0.0212 & 0.0086 & 0.0193 & 35.9 & 0.0212 & 0.0087 & 0.0193 & 35.9 \\
      \bottomrule
    \end{tabular}
    \begin{tablenotes}
      \footnotesize
    \item Rate-optimal $\hat{K}$ selected by bootstrap-Lepski without undersmoothing.
    \end{tablenotes}
  \end{threeparttable}
\end{table}

\begin{table}[h!]
  \centering
  \caption{\label{tab:Sim2_Disk}Integral on Estimated Upper Contour Set: Adaptive+Normal}
  \begin{threeparttable}
    \begin{tabular}{rrrrrrrrr}
      \toprule
      $n$ & RMSE & Bias & SD & CI\_L & CI\_U & U--L & Coverage & $\tilde{K}$ \\
      \midrule
      \multicolumn{9}{c}{(A) Plug-In} \\
      \midrule
      500  & 0.0694 &  0.0072 & 0.0690 & 3.0159 & 3.2818 & 0.2659 & 94.10\% & 51.1 \\
      1000  & 0.0470 &  0.0011 & 0.0470 & 3.0476 & 3.2378 & 0.1902 & 95.20\% & 55.3 \\
      2000  & 0.0335 &  0.0009 & 0.0335 & 3.0739 & 3.2110 & 0.1370 & 95.00\% & 65.7 \\
      4000  & 0.0237 & -0.0004 & 0.0237 & 3.0916 & 3.1907 & 0.0991 & 95.40\% & 81.9 \\
      8000  & 0.0180 &  0.0009 & 0.0180 & 3.1065 & 3.1786 & 0.0721 & 95.50\% & 94.4 \\
      \midrule
      \multicolumn{9}{c}{(B) Leave-One-Out} \\
      \midrule
      500  & 1.0379 & -0.0199 & 1.0383 & 2.9560 & 3.2873 & 0.3314 & 94.80\% & 50.5 \\
      1000  & 0.0471 &  0.0018 & 0.0471 & 3.0451 & 3.2417 & 0.1966 & 96.20\% & 54.9 \\
      2000  & 0.0335 &  0.0012 & 0.0335 & 3.0731 & 3.2125 & 0.1394 & 95.30\% & 64.9 \\
      4000  & 0.0238 & -0.0001 & 0.0238 & 3.0913 & 3.1917 & 0.1004 & 95.40\% & 82.0 \\
      8000  & 0.0181 &  0.0012 & 0.0180 & 3.1065 & 3.1791 & 0.0726 & 95.60\% & 94.4 \\
      \bottomrule
    \end{tabular}
    \begin{tablenotes}
      \footnotesize
    \item ``Adapt+Normal'': bootstrap-Lepski sieve dimension choice and normal critical values.
    \end{tablenotes}
  \end{threeparttable}
\end{table}

Next, we turn to the inference results. Table \ref{tab:Sim2_Disk} reports the finite-sample performance of the plug-in estimator as well as the Leave-One-Out (LOO) estimator, along with the CIs constructed using bootstrap-Lepski adaptive choice of sieve dimension and normal critical values. The results are again based on $B= 1000$ Monte Carlo replications and (within each replication) $B_{Lepski} = 500$ bootstrap draws when applicable. The results in Table \ref{tab:Sim2_Disk} display a similar pattern as those in Table \ref{tab:Sim1_Circle}: Again, the RMSE shrinks (almost at $\sqrt{n}$ rate) as the sample size increases, the bias is of smaller order relative to the standard deviation, and the realized CI coverage is close to the nominal 95\% level. One noticeable difference between the two tables, however, lies in that the RMSEs in Table \ref{tab:Sim2_Disk} are substantially smaller than those in Table \ref{tab:Sim1_Circle}. Heuristically, this might have been due to the fact that the upper contour set integral being estimated in Table \ref{tab:Sim2_Disk} is itself a full-dimensional integral: even though its asymptotic behavior is theoretically driven by the lower-dimensional submanifold, in finite sample the full-dimensional nature of the integral may have made it overall easier to estimate. Lastly, for $n\geq 1000$, the LOO-debiased estimator does result in (slightly but noticeably) smaller biases relative to the plug-in estimator, yielding comparable or better CI coverage as anticipated.\footnote{At $n=500$ the LOO estimator exhibits a large RMSE outlier, likely reflecting instability the leave-out procedure at this small sample size.}

Table \ref{tab:Sim2_Disk_Comp} again reports the RMSE and coverage comparison under the two sieve dimension choice methods and the two critical value choice methods. We observe that all four versions of CI constructions produce coverage close to the 95\% nominal level, and the improvement of the LOO estimator over the plug-in one in terms of size control is consistent across all four versions of CI construction procedures.

\begin{table}[h!]
  \centering
  \caption{Integral on Estimated Upper Contour Set: RMSE and Coverage Comparison}
  \label{tab:Sim2_Disk_Comp}
  \begin{threeparttable}
    \begin{tabular}{r|cc|cccc}
      \toprule
      \multicolumn{1}{c}{$n$ }  & \multicolumn{2}{c}{RMSE}   & \multicolumn{4}{c}{CI Coverage} \\
      \midrule
      & Adapt & Fixed & Adapt+Norm  & Fixed+Norm  & Adapt+Boot & Fixed+Boot \\
      \midrule
      \multicolumn{7}{c}{(A) Plug-In} \\
      \midrule
      500  & 0.0694 & 0.0671 & 94.1\% & 93.3\% & 93.4\% & 93.2\%\\
      1000 & 0.0470 & 0.0444 & 95.2\% & 96.3\% & 94.8\% & 96.1\%\\
      2000 & 0.0335 & 0.0316 & 95.0\% & 94.3\% & 94.7\% & 94.5\%\\
      4000 & 0.0237 & 0.0222 & 95.4\% & 96.1\% & 95.0\% & 95.6\% \\
      8000 & 0.0180 & 0.0165 & 95.5\% & 95.2\% & 95.2\% & 94.4\%\\
      \midrule
      \multicolumn{7}{c}{(B) Leave-One-Out} \\
      \midrule
      500  & 1.0379 & 0.0743 & 94.8\% & 95.7\%  & 94.4\% & 95.4\%\\
      1000 & 0.0471 & 0.0445 & 96.2\% & 97.0\%  & 95.4\% & 96.5\%\\
      2000 & 0.0335 & 0.0316 & 95.3\% & 94.8\%  & 94.9\% & 94.6\%\\
      4000 & 0.0238 & 0.0223 & 95.4\% & 96.1\%  & 95.1\% & 95.8\%\\
      8000 & 0.0181 & 0.0165 & 95.6\% & 95.3\%  & 95.7\% & 95.4\%\\
      \bottomrule
    \end{tabular}
    \begin{tablenotes}
      \footnotesize
    \item ``Adapt'': bootstrap-Lepski sieve dimension choice; ``Fixed'': fixed sieve dimension $K=64$.
    \item ``Norm'': normal critical values; ``Boot'': bootstrap quantile critical values.
    \end{tablenotes}
  \end{threeparttable}
\end{table}

\bibliographystyle{ecca}
\bibliography{Submanifold}

\newpage
\appendix
\section*{Online Appendix}
This Appendix consists of eight sections. Appendix \ref{subsec:Decomp} provides an instrumental transformation of submanifold integrals wrt Hausdorff measures into sums of lower-dimensional Lebesgue integrals using basic differential geometry tools. Appendix \ref{sec:MainProofs-LB} provides proofs of the lower bound rate results in Section \ref{sec:Minimax-LB}.  Appendix \ref{sec:MainProofs-UB} provides proofs of the upper bound rate results in Section \ref{sec:Minimax-UB}. Appendix \ref{sec:MainProofs_NL} provides proofs for Section \ref{sec:Nonlinear}. The remaining four sections cover the multiplier bootstrap choice of sieve dimension, analysis of quadratic submanifold integrals, additional upper bound rate results using a kernel nonparametric regression in the first stage, and a comparison of the bias-aware confidence interval with the undersmoothing approach used in the main text.

\section{\label{subsec:Decomp}Decomposition of Hausdorff Integrals on Submanifolds via Partition of Unity}

In this section, we explain how we can decompose a Hausdorff integral
on a regular submanifold into a sum of lower-dimensional Lebesgue
integrals, which can be analyzed more straightforwardly based on the
existing theory on the semiparametric estimation of Lebesgue integral
functionals.

This is done using a combination of standard mathematical tools from
differential manifold theory and geometric measure theory: see, for
example, \citet{munkres1991analysis} and \citet*{lang2002introduction}
for textbook treatments of differentiable manifolds, as well as \citet*{federer1996geometric}
and \citet*{evans2015measure} for textbook treatments of geometric measure
theory. See also Chapter 4 and Chapter 9 of \cite{delfour2001shapes} for a unified treatment. Specifically, the key idea is as follows:
\begin{itemize}
\item[(i)] Construct a finite open cover of the manifold $\mathcal{M}$ in question.
\item[(ii)] On each piece of the open cover in (i), apply the implicit function
theorem to obtain a parametrization (i.e., a coordinate chart) of
the $m$-dimensional manifold by an open subset of an $m$-dimensional
Euclidean space.
\item[(iii)] On each piece of the open cover in (i), use the ``change of variable''
formula in geometric measure theory to convert Hausdorff integral
on the manifold into an $m$-dimensional Lebesgue integral via the
parametrization in (ii).
\item[(iv)] Use the ``partition of unity'' method to combine the results in
(iii) across different pieces of the open cover in (i).
\end{itemize}
We now proceed through (i)-(iv) formally with more details.

By the regularity of ${\cal M}$ (Assumption \ref{assu:RegLevelSet}),
for every $\ol x\in{\cal M}$, $\Dif_{x}g\left(\ol x\right)$ has
rank $d-m$. Hence, we can decompose $\ol x$, potentially with a
permutation of coordinate indexes, as $\left(\ol x_{\left(m\right)},\ol x_{-\left(m\right)}\right)$,
where $\ol x_{\left(m\right)}$ is an $m$-dimensional subvector and
$\ol x_{-\left(m\right)}$ is the remaining $\left(d-m\right)$-dimensional
subvector, and correspondingly decompose $\Dif_{x}g\left(\ol x\right)$
into
\[
\Dif_{x}g\left(\ol x\right)=\left[\underset{m\times\left(d-m\right)}{\underbrace{\Dif_{x_{\left(m\right)}}g\left(\ol x\right)}},\underset{\left(d-m\right)\times\left(d-m\right)}{\underbrace{\Dif_{x_{-\left(m\right)}}g\left(\ol x\right)}}\right]
\]
so that $\text{rank}\left(\Dif_{x_{-\left(m\right)}}g\left(\ol x\right)\right)=d-m$.
Then, by the Implicit Function Theorem (see, e.g. Theorem 9.2 in \citealp{munkres1991analysis}),
there exists an open neighborhood ${\cal U}_{\ol x}$ around $\ol x_{\left(m\right)}$
and a unique $r$-times continuously differentiable function $\psi_{\ol x}:{\cal U}_{\ol x}\to\R^{-\left(m\right)}$
such that
\[
x_{-\left(m\right)}=\psi_{\ol x}\left(x_{\left(m\right)}\right),\quad\forall x_{\left(m\right)}\in{\cal U}_{\ol x}.
\]
Then, define $\varphi_{\ol x}:{\cal U}_{\ol x}\to{\cal M}\subseteq\R^{d}$
by
\[
\varphi_{\ol x}\left(x_{\left(m\right)}\right):=\left(x_{\left(m\right)},\psi_{\ol x}\left(x_{\left(m\right)}\right)\right).
\]
Then $\varphi_{\ol x}$ is a diffeomorphism ($r$-times differentiable
bijection) between ${\cal U}_{\ol x}$ and $\varphi_{\ol x}\left({\cal U}_{\ol x}\right)\subseteq{\cal M}$:
$\left({\cal U}_{\ol x},\varphi_{\ol x}\right)$ is thus a local coordinate
chart of ${\cal M}$ at $\ol x$, which parameterizes each point $x$
in $\varphi_{\ol x}\left({\cal U}_{\ol x}\right)$ on the manifold
with a vector $x_{\left(m\right)} \in {\cal U}_{\ol x}$, where ${\cal U}_{\ol x}$ is
an open set in an $m$-dimensional Euclidean space $\R^m$.

Clearly $\left\{ \varphi_{\ol x}\left({\cal U}_{\ol x}\right)\right\} _{\ol x\in{\cal M}}$
is an open cover for ${\cal M}$. Since ${\cal M}$ is compact, there
exists a finite sub-cover $\left\{ \varphi_{j}\left({\cal U}_{j}\right):j=1,...,\ol j\right\} $,
where, for each $j$, $\left({\cal U}_{j},\varphi_{j}\right)=\left({\cal U}_{\ol x^{\left(j\right)}},\varphi_{\ol x^{\left(j\right)}}\right)$
for some point $\ol x^{\left(j\right)}\in{\cal M}$. From now on,
we also write $\psi_{j}$ as the function $\psi_{\ol x^{\left(j\right)}}$
associated with $\left({\cal U}_{\ol x^{\left(j\right)}},\varphi_{\ol x^{\left(j\right)}}\right)$,
and we write $x_{\left(j,m\right)}$ for a generic point in ${\cal U}_{\ol x^{\left(j\right)}}$.
The ``$j$'' in the subscript ``$\left(j,m\right)$'' emphasizes
that, while $x_{\left(j,m\right)}$ is always an $m$-dimensional
vector that corresponds to $m$ different coordinates in $\R^{d}$,
the exact sequence of the $m$ coordinates may differ across different
$j$'s.

In general, the open cover $\left\{ \varphi_{j}\left({\cal U}_{j}\right):j=1,...,\ol j\right\} $
may have nonempty intersections. The standard mathematical tool to
avoid ``double counting'' under potentially overlapping covers is
the so-called ``partition of unity'', i.e., a collection of smooth
real-valued functions $\left\{ \rho_{j}:j=1,...,\ol j\right\} $ on
$\R^{d}$ such that (i) each $\rho_{j}$ is nonnegative (ii) $\sum_{j=1}^{J}\rho_{j}\left(x\right)=1$
for all $x\in{\cal M}$, and (iii) $\rho_{j}\left(x\right)=0$ for
all $x\notin\varphi_{j}\left({\cal U}_{j}\right)$. Such a ``partition
of unity'' $\left(\rho_{j}\right)$ is guaranteed to exist, say,
by Lemma 25.2 in \citet*{munkres1991analysis}.

Given $\left({\cal U}_{j},\varphi_{j},\rho_{j}\right)_{j=1}^{\ol j}$,
we may then decompose an integral on the manifold ${\cal M}$ into
a sum of integrals on $\varphi_{j}\left({\cal U}_{j}\right)$: for any map $\omega:{\cal X}\mapsto \R$,
\begin{equation}
\int_{{\cal M}}\omega\left(x\right)d{\cal H}^{m}\left(x\right)=\sum_{j=1}^{\ol j}\int_{\varphi_{j}\left({\cal U}_{j}\right)}\rho_{j}\left(x\right)\omega \left(x\right)d{\cal H}^{m}\left(x\right)\label{eq:PartUnity}
\end{equation}
Then, since $\varphi_{j}$ is a bijection between ${\cal U}_{j}$
and $\varphi_{j}\left({\cal U}_{j}\right)$ by construction, each
$x\in\varphi_{j}\left({\cal U}_{j}\right)$ can be expressed as $x=\varphi_{j}\left(x_{\left(j,m\right)}\right)$ for $x_{\left(j,m\right)} \in \R^m$.

By the ``change-of-variable'' formula in geometric measure theory,
e.g., Theorem 3.9 in \citet*{evans2015measure}, each of the $\ol j$
integrals on $\varphi_{j}\left({\cal U}_{j}\right)$ can be evaluated
as Lebesgue integrals on ${\cal U}_{j}\subset \R^m$ through the parametrization
$\varphi_{j}$, i.e.,
\begin{equation}
\int_{\varphi_{j}\left({\cal U}_{j}\right)}\rho_{j}\left(x\right)w\left(x\right)d{\cal H}^{m}\left(x\right)=\int_{{\cal U}_{j}}\rho_{j}\left(\varphi_{j}\left(x_{\left(j,m\right)}\right)\right)w\left(\varphi_{j}\left(x_{\left(j,m\right)}\right)\right)\mathcal{J}\varphi_{j}\left(x_{\left(j,m\right)}\right)dx_{\left(j,m\right)},\label{eq:HausToLeb}
\end{equation}
where $\mathcal{J}\varphi_{j}:=\sqrt{\text{det}\left(\varphi_{j}^{'}\varphi_{j}\right)}$
denotes the Jacobian of $\varphi_{j}$.
Combining \eqref{eq:PartUnity} and \eqref{eq:HausToLeb}
we obtain:
\begin{equation}
\int_{{\cal M}}\omega\left(x\right)d{\cal H}^{m}\left(x\right)=\sum_{j=1}^{\ol j}\int_{{\cal U}_{j}}\rho_{j}\left(\varphi_{j}\left(x_{\left(j,m\right)}\right)\right)
\omega\left(\varphi_{j}\left(x_{\left(j,m\right)}\right)\right)\mathcal{J}\varphi_{j}\left(x_{\left(j,m\right)}\right)dx_{\left(j,m\right)}.\label{eq:PieceLebInt}
\end{equation}

The decomposition in \eqref{eq:PieceLebInt} converts a Hausdorff
integral on a $m$-dimensional manifold ${\cal M}$ into a finite sum of $m$-dimensional Lebesgue integrals on $\R^m$,
the latter of which are easier to analyze using existing results on
semi/nonparametric estimation.
Below we show how we use \eqref{eq:PieceLebInt}
to establish the convergence rate and asymptotic normality of the
semiparametric estimator when the first stage nonparametric function $h_0$  could be learned/estimated using any machine learning/AI algorithms. The key is to use a sieve Riesz representation theory to establish the asymptotic local influence function representation for $\theta=\Gamma (h)$.

\section{\label{sec:MainProofs-LB}Proofs of Theoretical Results in Section \ref{sec:Minimax-LB}}

\subsection{Proof of Results in Subsection \ref{subsec:Minimax_Reg-LB}}

To establish the lower bound for the convergence rate in Theorem \ref{thm:LB-reg}, we use Le Cam's two-point comparison approach based on KL divergence.
\begin{lem}[Le Cam's Minimax Rate Bounds based on KL divergence]\label{lem:LeCam}
 Suppose that there exist $P_{0},P_{1}$ such that $KL\left(P_{0},P_{1}\right)\leq\frac{\log2}{n}$.
Then
\[
R_{n}:=\inf_{\tilde{\t}}\sup_{P}\E_{P}\left[d\left(\tilde{\t},\t\left(P\right)\right)\right]\geq\frac{1}{16}d\left(\t\left(P_{0}\right),\t\left(P_{1}\right)\right).
\]
\end{lem}

\begin{proof}[\textbf{Proof of Theorem \ref{thm:LB-reg}}]
In the lower bound proof, we can assume that the submanifold function $g\in \L^{s_1}({\cal X})$ for some $s_1 \geq \max\{1,s\}$. We first prove the result for the nonparametric regression case $h_0(x) = \E[\rest{Y_i}X_i=x]$.
Given any $h_{0}\in\Lambda_{c}^{s}(\mathcal{X})$, let $b_{n}$ be
a small positive number and define
\[
h_{1}\left(x\right):=h_{0}\left(x\right)+b_{n}^{s}K_{d-m}\left(\frac{g\left(x\right)}{b_{n}}\right)
\]
where $g(x):=\left( g_1(x),...,g_{d-m}(x)\right)$ and $K_{d-m}\left((t_1,...,t_{d-m})\right):=\prod_{\ell=1}^{d-m}K\left(t_{\ell}\right)$
and
\begin{equation}
K\left(t\right):=a\exp\left(-\frac{1}{1-t^{2}}\right)\ind\left\{ \left|t\right|\leq1\right\} \label{eq:K_bump-1}
\end{equation}
for some sufficiently small $a>0$ such that $h_{1}$ stays in H\"older
class of smoothness order $s$. To see this, notice that $K$ and
$K_{d-m}$ are infinitely differentiable with uniformly bounded derivatives,
and for any $k\leq s$ we have:
\[
\frac{\p^{k}}{\p x_{j}^{k}}h_{1}(x)=\frac{\p^{k}}{\p x_{j}}h_{0}(x)+b_{n}^{s-k}\left(K_{d-m}^{\left(1\right)}\left(\cd\right)\frac{\p^{k}}{\p x_{j}^{k}}g\left(x\right)+...+K_{d-m}^{\left(k\right)}\left(\cd\right)\frac{\p}{\p x_{j}}g\left(x\right)\right).
\]
Since $h_0,g\in \L^s_c({\cal X})$, all derivatives of $K$ are uniformly bounded, and $b_n$ is small, $\frac{\p^{k}}{\p x_{j}^{k}}h_{1}(x)$ is uniformly bounded as well. The H\"older condition for the fractional exponent $s-[s]$ can also be similarly verified.

Note that, to derive a lower bound for the minimax rate for a class of estimation problems, it suffices to establish the lower bound under one admissible example of the problem. Specifically, we consider the example setting where $X_{i}$ is uniformly distributed on $\left[0,1\right]^{d}$, and $w(x) \equiv 1$. In addition, let $P_0$ be the joint distribution (with density $p_0$) of $\left(X_{i},Y_{i0}\right)_{i=1}^{n}$ with
\[
Y_{i0}=h_{0}\left(X_{i}\right)+\e_{i},\quad\e_{i}\sim_{i.i.d.}\cN\left(0,1\right).
\]
and let $P_{1}$ be the joint distribution (with density $p_1$) of $\left(X_{i},Y_{i1}\right)_{i=1}^{n}$
with
\[
Y_{i1}=h_{1}\left(X_{i}\right)+\e_{i},\quad\e_{i}\sim_{i.i.d.}\cN\left(0,1\right).
\]

Then, the KL divergence between $P_{0}^{n}$ and $P_{1}^{n}$ is given
by:
\begin{align}
KL\left(P_{0}^{n},P_{1}^{n}\right) & =nKL\left(P_{0},P_{1}\right)\nonumber \\
 & =n\int_{\left[0,1\right]^{d}}\int p_{0}\left(x,y\right)\log\left(\frac{p_{0}\left(x,y\right)}{p_{1}\left(x,y\right)}\right)dydx\nonumber \\
 & =n\int_{\left[0,1\right]^{d}}\int p_{0}\left(\rest yx\right)\log\left(\frac{p_{0}\left(\rest yx\right)}{p_{1}\left(\rest yx\right)}\right)dydx\nonumber \\
 & =n\int_{\left[0,1\right]^{d}}\int\phi\left(y\right)\log\left(\frac{\phi\left(y-h_{0}\left(x\right)\right)}{\phi\left(y-h_{1}\left(x\right)\right)}\right)dydx\nonumber \\
 & =n\int_{\left[0,1\right]^{d}}\int KL\left(\cN\left(h_{0}\left(x\right),1\right),{\cal N}\left(h_{1}\left(x\right),1\right)\right)dydx\nonumber \\
 & =n\frac{1}{2}\int_{\left[0,1\right]^{d}}\left(h_{1}\left(x\right)-h_{0}\left(x\right)\right)^{2}dx\nonumber \\
 & =n\frac{1}{2}\int_{\left[0,1\right]^{d}}b_{n}^{2s}K_{d-m}^{2}\left(\frac{g\left(x\right)}{b_{n}}\right)dx\nonumber \\
 & =n\frac{1}{2}b_{n}^{2s}\int_{\left[0,1\right]^{d}}K_{d-m}^{2}\left(\frac{g\left(x\right)}{b_{n}}\right)dx\nonumber \\
 & \leq const.\times n b_{n}^{2s}\int_{\left[0,1\right]^{d}}\prod_{\ell=1}^{d-m}\ind\left\{ \frac{\left|g_{\ell}\left(x\right)\right|}{b_{n}}\leq1\right\} dx\nonumber \\
 & \leq Cnb_{n}^{2s}\P\left(\norm{g\left(X_{i}\right)}\leq b_{n}\right) \leq C'nb_{n}^{2s+d-m}\label{eq:KL_rate}
\end{align}
provided that
\begin{equation}
\P\left(\norm{g\left(X_{i}\right)}\leq b_{n}\right)
\leq Mb_{n}^{d-m}.\label{eq:SphereRate}
\end{equation}

We now show that \eqref{eq:SphereRate} holds. Let $W_{i}:=g\left(X_{i}\right) \in \R^{d-m}$,
note that the density of $W_{i}$
\begin{align*}
p_{W}\left(w\right) & :=\int_{\left\{ g\left(x\right)=w\right\} }p_{X}\left(x\right)\frac{1}{\mathcal{J}_g\left(x\right)}d{\cal H}^{m}\left(x\right)\leq M,
\end{align*}
is uniformly bounded since $p_{X}\left(x\right)\equiv1$ and $\mathcal{J}_g\left(x\right)$
is uniformly bounded away from zero as below by Assumption \ref{assu:Jacob}(i).
Consequently,
\begin{align*}
\P\left(\norm{g\left(X_{i}\right)}\leq b_{n}\right) & =\int_{\ind\left\{ \left|w\right|\leq b_{n}\right\} }p_{W}\left(w\right)dw \leq M\int_{\ind\left\{ \left|w\right|\leq b_{n}\right\} }dw \leq Mb_{n}^{d-m}
\end{align*}
verifying condition \eqref{eq:SphereRate}.

Hence, we can set $b_{n} =  \left(\frac{\log 2}{C'n}\right)^{\frac{1}{2s+d-m}} \asymp n^{-\frac{1}{2s+d-m}}$ so that
\[
KL\left(P_{0},P_{1}\right) \leq C'b_{n}^{2s+d-m} = \frac{\log2}{n}
\]
as required in Lemma \ref{lem:LeCam}.

In the meanwhile, notice that
\begin{align*}
\left|L\left(h_{1}\right)-L\left(h_{0}\right)\right| & =\int_{\left\{x\in\mathcal{X}: g(x)=\textbf{0}\right\} }\left|h_{1}\left(x\right)-h_{0}\left(x\right)\right|d{\cal H}^{m}\left(x\right)
 =C''b_{n}^{s}K\left(0\right)
\end{align*}

Hence, define $\theta(P) := L(\E_P[Y|X=\cd])$ and $d(\t,\t') = |\t-\t'|$, we deduce from Lemma \ref{lem:LeCam} that the minimax rate
\begin{align*}
R_{n} & \geq \frac{1}{16}d(\t(P_{0}),\t(P_{1})) \\
& = \frac{1}{16}|L\left(h_{1}\right)-L\left(h_{0}\right)|\\
& =\frac{1}{16}\int_{\left\{x\in\mathcal{X}: g(x)=\textbf{0}\right\} }\left|h_{1}\left(x\right)-h_{0}\left(x\right)\right|d{\cal H}^{m}\left(x\right)\\
&=C''b_{n}^{s}|K\left(0\right)| \asymp b_{n}^{s} \asymp   n^{-\frac{s}{2s+d-m}}.
\end{align*}
\end{proof}


\begin{proof}[\textbf{Proof of Corollary \ref{cor:density}}]
We now prove that the same rate $r^*_n$ as in Theorem \ref{thm:LB-reg} also applies in the nonparametric density case, i.e., $h_0(x) \equiv p_0(x)$ is the probability density function of $X_i$.

We set $h_{0}\left(x\right)\equiv p_{0}\left(x\right)\equiv1$ on
$\left[0,1\right]^{d}$, and set
\[
p_{1}\left(x\right):=p_{0}\left(x\right)+b_{n}^{s}\tilde{K}\left(\frac{\norm{g\left(x\right)}}{b_{n}}\right)
\]
where
\[
\tilde{K}_{n}\left(t\right):=a\left[\exp\left(-\frac{1}{1-t^{2}}\right)-c_{n}\right]\ind\left\{ \left|t\right|\leq1\right\}
\]
with
\[
c_{n}:=\frac{\int\exp\left(-\frac{1}{1-\norm{g\left(x\right)}^2/b_{n}^2}\right)\ind\left\{ \norm{g\left(x\right)}\leq b_{n}\right\} dx}{\int\ind\left\{ \norm{g\left(x\right)}\leq b_{n}\right\} dx}
\]
so that $\tilde{K}_{n}\left(t\right)=0$ whenever $\left|t\right|>1$
and furthermore
\[
\int\tilde{K}_{n}\left(\frac{\norm{g\left(x\right)}}{b_{n}}\right)dx=0.
\]
This ensures that
\begin{align*}
\int p_{1}\left(x\right)dx & =\int p_{0}\left(x\right)dx+b_{n}^{s}\int\tilde{K}_{n}\left(\frac{\norm{g\left(x\right)}}{b_{n}}\right)dx=1+0=1.
\end{align*}
Furthermore, $a>0$ can be set sufficiently small to ensure that
\[
p_{1}\left(x\right)\geq0,
\]
so that $p_{1}$ remains a valid density function.

Then, using the inequality that $ KL(P_1,P_0) \leq \chi
^2(P_1,P_0)$ we have
\begin{align}
KL\left(P_{1},P_{0}\right)\leq\chi_{2}\left(P_{1},P_{0}\right) & :=\int_{\left[0,1\right]^{d}}\left(\frac{p_{1}\left(x\right)}{p_{0}\left(x\right)}-1\right)^{2}p_{0}\left(x\right)dx\nonumber \\
 & =\int_{\left[0,1\right]^{d}}\left(b_{n}^{s}\tilde{K}\left(\frac{\norm{g\left(x\right)}}{b_{n}}\right)\right)^{2}dx\nonumber \\
 & =b_{n}^{2s}\int_{\left[0,1\right]^{d}}\tilde{K}\left(\frac{\norm{g\left(x\right)}}{b_{n}}\right)^{2}dx\nonumber \\
 & \leq b_{n}^{2s}\int_{\left[0,1\right]^{d}}\ind\left\{ \frac{\norm{g\left(X_{i}\right)}}{b_{n}}\leq1\right\} dx\nonumber \\
 & \leq Cb_{n}^{2s}\P\left\{ \norm{g\left(X_{i}\right)}\leq b_{n}\right\} \nonumber \\
 & \leq C^{2}b_{n}^{2s+d-m}\label{eq:Chi2_Rate}
\end{align}
which coincides with the rate in \eqref{eq:KL_rate} under \eqref{eq:SphereRate}
for the nonparametric regression case. The rest of the proof is the
same as in the nonparametric regression case.
\end{proof}


\begin{proof}[\textbf{Proof of Corollary \ref{cor:NL_rate_LB}}]
The proof is adapted from the proof of Theorem \ref{thm:LB-reg}, combined with the linearization in Assumption \ref{assu:DGamma_linear}. As in the proof of Theorem \ref{thm:LB-reg}, fix any $P_0$ with regression function $h_0$ and associated $(g,w)$. For a given small $b_n>0$ small, define
\[
h_1(x) := h_0(x) + b_n^s K_{d-m}\!\left(\frac{g(x)}{b_n}\right),
\]
where $K_{d-m}$ is defined in the proof of Theorem \ref{thm:LB-reg}. Let $P_0$ denote the joint distribution of $(X_i,Y_i)$ with $X_i$ uniform on $[0,1]^d$ and
$Y_i = h_0(X_i)+\epsilon_i$ with $\epsilon_i\sim_{i.i.d.} \cN(0,1)$ and let $P_1$ be the analogous distribution obtained by replacing $h_0$ with $h_1$. Then, exactly as in the proof of Theorem \ref{thm:LB-reg}, we have
\[
\mathrm{KL}(P_0^n,P_1^n)
= \frac{n}{2}\int_{[0,1]^d}\big(h_1(x)-h_0(x)\big)^2\,dx
\;\lesssim\; n\,b_n^{2s+d-m}.
\]
and thus $b_n \asymp n^{-\frac{1}{2s+d-m}}$ ensures that $\mathrm{KL}(P_0^n,P_1^n)\le\log 2$ for all $n$ large enough. Now, by Assumption \ref{assu:DGamma_linear} and the definition of $h_1$, we have
\begin{align*}
D\Gamma(h_0)[h_1-h_0]
&= \int_{\cM} \big(h_1(x)-h_0(x)\big)w(x)\,d\mathcal{H}^m(x) = O(b_n^s) 
\end{align*}
following the same argument in proof of Theorem \ref{thm:LB-reg}, just with $D\G(h_0)[h_1-h_0]$ in place of $L(h_1-h_0)$.
\begin{equation}
\big|D\Gamma(h_0)[h_1-h_0]\big|
\;\asymp\; b_n^s.
\label{eq:linear-separation}
\end{equation}
Then, writing $R_n := \G(h_1)-\G(h_0)-D\Gamma(h_0)[h_1-h_0]$
\begin{align*}
\abs{\Gamma(h_1)-\Gamma(h_0)}
\ & =\  \abs{D\Gamma(h_0)[h_1-h_0] + R_n} \\
& \geq\  \abs{\abs{D\Gamma(h_0)[h_1-h_0]} - \abs{R_n}}\\
& \geq\  (1-C_r)\abs{D\Gamma(h_0)[h_1-h_0]} \asymp  O(b_n^s)
\end{align*}
where the last inequality is due to Assumption \ref{assu:DGamma_linear}(ii).
Hence, by Lemma \ref{lem:LeCam}, we have the conclusion.
\end{proof}

\subsection{Proof of Results in Subsection \ref{subsec:Minimax_NPIV-LB}}

\begin{proof}[\textbf{Proof of Theorem \ref{thm:LB_NPIV}}]
Let $\Pr_{h}$ denote the joint distribution of $(Y,X,Z)$ satisfying $Y= T h+u$ with a known conditional expectation operator $T$ and  $u|Z\sim \mathcal N(0,\sigma^2)$, the so-called reduced-form nonparametric indirection regression (NPIR) model as in \cite{chen2011rate} with a fixed variance $\sigma^2>0$. As pointed out in \cite{chen2011rate} it is enough to consider the NPIR model to establish the lower bound for the NPIV model, specifically, a lower bound for the NPIR model is also a lower bound for the NPIV model.

Fix any joint law of $(X,Z)$ that satisfies Assumption \ref{a-NPIV-data}(i)--(ii) and yields a conditional
expectation operator $T h(z)=\E[h(X)\mid Z=z]$ satisfying Assumption \ref{ass:link}.
Consider the following \emph{submodel} of the NPIV model \eqref{eq:npreg}:
\begin{equation}\label{eq:NPIR-submodel}
Y = (Th)(Z) + u,\qquad u\mid Z \sim \mathcal N(0,\underline\sigma^2),\qquad u \perp (X,Z),
\end{equation}
with $\underline\sigma^2$ as in Assumption \ref{a-NPIV-data}(iii).  For this submodel, writing
$\epsilon := Y-h(X) = (Th)(Z)-h(X)+u$, we have
$\E[\epsilon\mid Z]=Th(Z)-\E[h(X)\mid Z]+\E[u\mid Z]=0$, so \eqref{eq:npreg} holds.
Moreover, $\E[\epsilon^2\mid Z]\geq \E[u^2\mid Z]=\underline\sigma^2$, so Assumption
\ref{a-NPIV-data}(iii) is satisfied.  Since Theorem \ref{thm:LB_NPIV} is a minimax \emph{lower bound}
(over a larger class of NPIV DGPs), it suffices to lower bound the minimax risk on the restricted
submodel \eqref{eq:NPIR-submodel}.

We again apply Le Cam two-point testing argument.
Let $h_0\equiv 0$. We now use the CDV wavelet to construct another function $h_1$.
Fix a wavelet type $\tilde G_0$ (one of the finitely many tensor-product CDV wavelet types) such that
the corresponding mother wavelet is continuous and not identically zero. Such a type exists by
construction of the CDV basis. For each resolution level $j\in\mathbb N$, let
$\{\psi_{j,k,\tilde G_0}:k\in\mathcal K_j\}$ denote the corresponding wavelets at level $j$.
By compact support and continuity, there exist constants $c_\psi>0$ and $r_\psi\in(0,1)$ such that for
every sufficiently large $j$ there is a collection of indices $k$ for which
\begin{equation}\label{eq:wavelet-lower}
|\psi_{j,k,\tilde G_0}(x)| \;\geq\; c_\psi\,2^{jd/2}
\quad\text{for all }x\in B(x_{j,k},\,r_\psi 2^{-j}),
\end{equation}
for some $x_{j,k}$ in the interior of $\supp(\psi_{j,k,\tilde G_0})$; moreover the balls
$B(x_{j,k},r_\psi 2^{-j})$ can be chosen to be contained in $\supp(\psi_{j,k,\tilde G_0})$.

Next, we pick a set of wavelets whose supports lie along the submanifold.
Assume $L(\cdot)$ is nontrivial, i.e.
$\int_{\mathcal M} |w(x)|\,d\mathcal H^m(x)>0$; otherwise $L(h)\equiv 0$ and the lower bound is trivial.
Choose $\underline w>0$ such that the set
$\mathcal M_w:=\{x\in\mathcal M:|w(x)|\geq \underline w\}$ has $\mathcal H^m(\mathcal M_w)>0$.

Under Assumption \ref{assu:RegLevelSet}, $\mathcal M$ is a $C^1$ compact $m$-dimensional submanifold. Hence,
there exists $c_M>0$ such that for all $x\in\mathcal M$ and all small enough
$r>0$,
\begin{equation}\label{eq:ahlfors}
\mathcal H^m(\mathcal M\cap B(x,r)) \;\ge\; c_M r^m.
\end{equation}
Consequently, for each large $j$ we can select a set of indices $\mathcal I_j\subset \mathcal K_j$ with
cardinality $|\mathcal I_j|\asymp 2^{jm}$ such that the balls
$\{B(x_{j,k},r_\psi 2^{-j})\}_{k\in\mathcal I_j}$ are pairwise disjoint and satisfy
$x_{j,k}\in \mathcal M_w$ for all $k\in\mathcal I_j$.

For this fixed $j$, define the perturbation
\begin{equation}\label{eq:h1-def}
h_1(x) \;:=\; \delta\,2^{-j(s+d/2)}\sum_{k\in\mathcal I_j} a_{j,k}\,\psi_{j,k,\tilde G_0}(x),
\qquad
a_{j,k}:=\sgn\!\Big(\int_{\mathcal M}\psi_{j,k,\tilde G_0}(x)w(x)\,d\mathcal H^m(x)\Big),
\end{equation}
with a constant $\delta>0$ to be chosen sufficiently small.

By the standard CDV-wavelet characterization of Holder/Besov classes
(e.g. Appendix A of \cite{chen2018optimal}), there exists a constant $\delta_0>0$ such that
if $\sup_{j,k,\tilde G} 2^{j(s+d/2)}|\langle h,\psi_{j,k,\tilde G}\rangle_{L^2(\mathcal X)}|\leq \delta_0$,
then $h\in \Lambda_c^s$ (for a possibly different radius constant $c$).
Since $h_1$ has wavelet coefficients equal to $\delta\,2^{-j(s+d/2)}$ at level $j$ and $0$ elsewhere,
choosing $\delta\le\delta_0$ ensures $h_1\in\Lambda_c^s(\mathcal X)$.

Because $a_{j,k}$ aligns the sign of each contribution,
\[
|L(h_1)-L(h_0)|
=|L(h_1)|
=\delta\,2^{-j(s+d/2)}\sum_{k\in\mathcal I_j}\Big|\int_{\mathcal M}\psi_{j,k,\tilde G_0}(x)w(x)\,d\mathcal H^m(x)\Big|.
\]
For each $k\in\mathcal I_j$, using \eqref{eq:wavelet-lower}, $|w|\ge \underline w$ on $\mathcal M_w$,
and the lower volume bound \eqref{eq:ahlfors} with $r=r_\psi 2^{-j}$, we obtain
\[
\Big|\int_{\mathcal M}\psi_{j,k,\tilde G_0}(x)w(x)\,d\mathcal H^m(x)\Big|
\;\ge\;
\underline w\int_{\mathcal M\cap B(x_{j,k},r_\psi 2^{-j})} |\psi_{j,k,\tilde G_0}(x)|\,d\mathcal H^m(x)
\;\gtrsim\;
2^{jd/2}\cdot 2^{-jm}
=
2^{j(d/2-m)}.
\]
Therefore,
\[
|L(h_1)-L(h_0)|
\;\gtrsim\;
\delta\,2^{-j(s+d/2)}\cdot |\mathcal I_j|\cdot 2^{j(d/2-m)}
\;\asymp\;
\delta\,2^{-j(s+d/2)}\cdot 2^{jm}\cdot 2^{j(d/2-m)}
=
\delta\,2^{-js}.
\]
Hence, the parameter separation is of order $2^{-js}$.

Next we compute the KL divergence under the NPIR model.
Consider the distributions $P_0,P_1$ induced by \eqref{eq:NPIR-submodel} with $h=h_0$ and $h=h_1$,
and the same joint law of $(X,Z)$ in both cases. Since only the conditional mean of $Y\mid Z$ changes,
standard Gaussian KL calculations yield
\begin{equation}\label{eq:KL-NPIR}
KL(P_0^n,P_1^n)=n\,KL(P_0,P_1)
=\frac{n}{2\underline\sigma^2}\,\|T(h_1-h_0)\|_{L^2(\mathcal Z)}^2
=\frac{n}{2\underline\sigma^2}\,\|Th_1\|_{L^2(\mathcal Z)}^2.
\end{equation}
By the link condition \eqref{eq:link-condition} and the fact that $h_1$ has nonzero wavelet coefficients
only at the single level $j$ and type $\tilde G_0$, we have
\[
\|Th_1\|_{L^2(\mathcal Z)}^2
\;\lesssim\;
\nu(2^{j})^2 \sum_{k\in\mathcal I_j}\bigl\langle h_1,\psi_{j,k,\tilde G_0}\bigr\rangle_{L^2(\mathcal X)}^2
=
\nu(2^{j})^2\cdot |\mathcal I_j|\cdot \delta^2\,2^{-2j(s+d/2)}
\;\asymp\;
\delta^2\,\nu(2^{j})^2\,2^{-j(2s+d-m)}.
\]
Plugging this into \eqref{eq:KL-NPIR} gives
\begin{equation}\label{eq:KL-bound}
KL(P_0^n,P_1^n)\;\lesssim\; n\,\delta^2\,\nu(2^{j})^2\,2^{-j(2s+d-m)}.
\end{equation}

Choose $\delta>0$ small enough (fixed) so that the implicit constant in \eqref{eq:KL-bound} is absorbed.

\emph{Mildly ill-posed case:} $\nu(t)=t^{-\varsigma}$.
Then \eqref{eq:KL-bound} becomes $KL(P_0^n,P_1^n)\lesssim n\,2^{-j(2(s+\varsigma)+d-m)}$.
Choose $j=j_n$ such that $2^{j_n}\asymp n^{1/(2(s+\varsigma)+d-m)}$, which yields
$KL(P_0^n,P_1^n)\le (\log 2)$ for all large $n$.
By Step 3, the separation satisfies $|L(h_1)-L(h_0)|\gtrsim 2^{-j_n s}\asymp n^{-s/(2(s+\varsigma)+d-m)}$.

\emph{Severely ill-posed case:} $\nu(t)=\exp(-\tfrac12 t^\varsigma)$.
Then \eqref{eq:KL-bound} becomes
$KL(P_0^n,P_1^n)\lesssim n\,\exp(-c\,2^{\varsigma j})\,2^{-j(2s+d-m)}$ for some $c>0$.
Choose $j=j_n$ such that $2^{\varsigma j_n}\asymp \log n$, which implies $KL(P_0^n,P_1^n)\to 0$ and hence
$KL(P_0^n,P_1^n)\le (\log 2)$ for all large $n$.
Step 3 gives $|L(h_1)-L(h_0)|\gtrsim 2^{-j_n s}\asymp (\log n)^{-s/\varsigma}$.

Finally, apply Lemma \ref{lem:LeCam} with the metric $d(\theta,\theta')=|\theta-\theta'|$ and
$\theta(P_h)=L(h)$.
Since $KL(P_0^n,P_1^n)\le (\log 2)$ for the above choices, Lemma \ref{lem:LeCam} yields
\[
\inf_{\tilde\theta_n}\sup_{h\in\Lambda_c^s}\E_{P_h}\bigl[(\tilde\theta_n-L(h))^2\bigr]
\;\gtrsim\;
|L(h_1)-L(h_0)|^2
\;\gtrsim\;
r_{NPIV,n}^2,
\]
with $r_{NPIV,n}$ as stated in Theorem \ref{thm:LB_NPIV}.
This proves the claim.
\end{proof}

\begin{proof}[\textbf{Proof of Corollary \ref{cor:NL_NPIV_LB}}]

Since the minimax risk over $\Lambda_c^s(\mathcal{X})$ dominates the risk over any subset, it suffices to
construct $h_0,h_1\in\Lambda_c^s(\mathcal{X})$ (with associated distributions $P_0,P_1$ under the NPIR model
such that (i) $\mathrm{KL}(P_0^n,P_1^n)\le\log 2$ and (ii) $\abs{\Gamma(h_1)-\Gamma(h_0)}\gtrsim r_{\mathrm{NPIV},n}$.
Lemma \ref{lem:LeCam} then yields the desired lower bound.

The proof follows from those of Corollary \ref{cor:NL_rate_LB} and Theorem \ref{thm:LB_NPIV}. Hence we are brief here.
We follow the proof of Corollary \ref{cor:NL_rate_LB} to control the nonlinear remainder in the separation in the target parameter $\Gamma(h)$
Let $R_n := \G(h_1)-\G(h_0)-D\Gamma(h_0)[h_1-h_0]$
\begin{align}
\abs{\Gamma(h_1)-\Gamma(h_0)}
\ & =\  \abs{D\Gamma(h_0)[h_1-h_0] + R_n}\nonumber \\
& \geq\  \abs{\abs{D\Gamma(h_0)[h_1-h_0]} - \abs{R_n}}\nonumber\\
& \geq\  (1-C_r)\abs{D\Gamma(h_0)[h_1-h_0]} \geq const.\times 2^{-js} \label{eq:Gamma_sep}
\end{align}
where the last inequality is due to Assumption \ref{assu:DGamma_linear}(ii), and the bound $\abs{D\Gamma(h_0)[h_1-h_0]} \geq const.\times 2^{-js}$ follows from the proof of Theorem \ref{thm:LB_NPIV} their lower bound on the linear separation of $\abs{L[h_1-h_0]} \geq const.\times 2^{-js}$.

From the proof of Theorem \ref{thm:LB_NPIV} we have:
\begin{equation}\label{eq:kl_bound}
\mathrm{KL}(P_0^n,P_1^n)\ \lesssim\ n\,\delta^2\,\nu(2^{j})^2\,2^{-j(2s+d-m)}.
\end{equation}

We finally optimize the choice of $j=j_n$ depending on the degree of ill-posedness.

\emph{Mildly ill-posed case.} If $\nu(t)=t^{-\varsigma}$ then \eqref{eq:kl_bound} becomes
$\mathrm{KL}(P_0^n,P_1^n)\lesssim n\,2^{-j(2(s+\varsigma)+d-m)}$.
Choose $j=j_n$ so that $2^{j_n}\asymp n^{1/(2(s+\varsigma)+d-m)}$.
Then $\mathrm{KL}(P_0^n,P_1^n)\le \log 2$ for all large $n$ (after fixing $\delta>0$ small enough).
By \eqref{eq:Gamma_sep},
\[
\abs{\Gamma(h_1)-\Gamma(h_0)}\ \gtrsim\ 2^{-j_n s}\ \asymp\ n^{-\frac{s}{2(s+\varsigma)+d-m}}
=r_{\mathrm{NPIV},n}.
\]

\emph{Severely ill-posed case.} If $\nu(t)=\exp(-\tfrac12 t^\varsigma)$ then \eqref{eq:kl_bound} gives
$\mathrm{KL}(P_0^n,P_1^n)\lesssim n\,\exp(-c\,2^{\varsigma j})\,2^{-j(2s+d-m)}$ for some $c>0$.
Choose $j=j_n$ so that $2^{\varsigma j_n}\asymp \log n$. Then $\mathrm{KL}(P_0^n,P_1^n)\to 0$, hence
$\mathrm{KL}(P_0^n,P_1^n)\le \log 2$ for all large $n$, and \eqref{eq:Gamma_sep} yields
\[
\abs{\Gamma(h_1)-\Gamma(h_0)}\ \gtrsim\ 2^{-j_n s}\ \asymp\ (\log n)^{-\frac{s}{\varsigma}}
=r_{\mathrm{NPIV},n}.
\]

Finally, Lemma \ref{lem:LeCam} implies that for any estimator $\tilde\theta_n$,
\[
\sup_{h\in\{h_0,h_1\}}
\E_h\!\left[(\tilde\theta_n-\Gamma(h))^2\right]
\ \gtrsim\ r_{\mathrm{NPIV},n}^2.
\]
Since $\{h_0,h_1\}\subset\Lambda_c^s(\mathcal{X})$, the same lower bound holds for the minimax risk over $\Lambda_c^s(\mathcal{X})$.
Hence, we obtain Corollary \ref{cor:NL_NPIV_LB}.

\end{proof}

\section{\label{sec:MainProofs-UB}Proofs of Theoretical Results in Section \ref{sec:Minimax-UB}}

\subsection{Proof of Results in Subsection \ref{subsec:Minimax_Lin-UB}}

\begin{proof}[\textbf{Proof of Lemma \ref{lem:NormRate_d-m}}]
Write $K_n=K=J^{d}$. Note that
\begin{align*}
\norm{L\left(\ol b^{K}\right)}^{2}= & \sum_{k=1}^{K}L^{2}\left(\ol b_{k}^{K}\right)=\sum_{k=1}^{K_{n}}\left[\int_{{\cal M}}\ol b_{k}^{K_{n}}\left(x\right)w\left(x\right)d{\cal H}^{m}\left(x\right)\right]^{2}.
\end{align*}
Since $\ol b^K = G^{-1/2} b^K$ and $G$ has eigenvalues bounded away from zero and infinity
by Assumption~\ref{assu:Cond_Sieve}(i), we have
$\norm{L(\ol b^K)}^2 = L(b^K)^\top G^{-1} L(b^K) \asymp \norm{L(b^K)}^2 = \sum_{k=1}^K L^2(b_k^K)$.
It therefore suffices to bound $\sum_{k=1}^K L^2(b_k^K)$.
We note that the tensor-product structure holds for the original B-spline basis $b^K$, but \emph{not}
in general for the orthonormalized basis $\ol b^K$; accordingly, we carry out the following calculation for $b^K$.

Recall that $\left(b_{k}^{K_{n}}\right)_{k=1}^{K_{n}}$ is constructed
as tensor products of univariate basis functions
\[
b_{k}^{K_{n}}\left(x\right)=\prod_{\ell=1}^{d}b_{k_{l}}\left(x_{l}\right)
\]
for some $1\leq k_{1},...,k_{d}\leq J$: in other words, each
index $k$ is bijectively identified by the vector $\left(k_{1},...,k_{d}\right)$.
Hence,
\[
\sum_{k=1}^{K}L^{2}\left(b_{k}^{K}\right)=\sum_{k_{1},...,k_{d}=1}^{J}L^{2}\left(\prod_{\ell=1}^{d}b_{k_{l}}\left(x_{l}\right)\right).
\]
For each ${L\left(b_{k}^{K}\right)}$, we apply
the decomposition \eqref{eq:PieceLebInt} and obtain
\[
L\left(b_{k}^{K}\right)=\int_{{\cal M}}b_{k}^{K}\left(x\right)w\left(x\right)p\left(x\right)d{\cal H}^{m}\left(x\right)
=\sum_{j=1}^{\ol j}T_{kj},
\]
where each $j$ corresponds to a piece from the loal charts $\left({\cal U}_{j},\varphi_{j}\right)$
for ${\cal M}$ along with the partition of unit function $\rho_{j}$
and
\begin{align*}
T_{kj} & :=\int_{{\cal {\cal U}}_{j}}\rho_{j}\left(\varphi_{j}\left(x_{\left(j,m\right)}\right)\right)b_{k}^{K}\left(\varphi_{j}\left(x_{\left(j,m\right)}\right)\right)w\left(\varphi_{j}\left(x_{\left(j,m\right)}\right)\right)\mathcal{J}\varphi_{j}\left(x_{\left(j,m\right)}\right)dx_{\left(j,m\right)}.
\end{align*}
From now on, to simplify notation, we will suppress the subscript
$j$ whenever there is no ambiguity. Furthermore, define
\[
\ol w\left(x_{\left(m\right)}\right):=\rho\left(\varphi\left(x_{\left(m\right)}\right)\right)w\left(\varphi\left(x_{\left(m\right)}\right)\right)\mathcal{J}\varphi\left(x_{\left(m\right)}\right),
\]
so that we may write $T_{kj}$ more succinctly as
\[
T_{kj}=\int_{{\cal {\cal U}}}b_{k}^{K}\left(\varphi\left(x_{\left(m\right)}\right)\right)\ol w\left(x_{\left(m\right)}\right)dx_{\left(m\right)}.
\]
Again, since $\left(b_{k}^{K_{n}}\right)$ is constructed as tensor
products of univariate basis functions, we can decompose
\[
b_{k}^{K}\left(\varphi\left(x_{\left(m\right)}\right)\right)=b_{k,\left(m\right)}^{K}\left(x_{\left(m\right)}\right)\cd b_{k,-\left(m\right)}^{K}\left(\psi\left(x_{\left(m\right)}\right)\right)
\]
where $b_{k,\left(m\right)}^{K}\left(x_{\left(m\right)}\right)=\prod_{\ell\in\left(m\right)}b_{k_{l}}\left(x_{l}\right)$
and $b_{k,-\left(m\right)}^{K}\left(x_{\left(m\right)}\right)=\prod_{\ell\notin\left(m\right)}b_{k_{l}}\left(\left[\psi\left(x_{\left(m\right)}\right)\right]_{\ell}\right)$
correspond to the coordinates in $x_{\left(m\right)}$ and $x_{-\left(m\right)}$,
respectively. Then, we have
\begin{align*}
T_{kj} & =\int_{{\cal {\cal U}}}b_{k,\left(m\right)}^{K}\left(x_{\left(m\right)}\right)b_{k,-\left(m\right)}^{K}\left(\psi\left(x_{\left(m\right)}\right)\right)\ol w\left(\varphi\left(x_{\left(m\right)}\right)\right)dx_{\left(m\right)}\\
 & =<b_{\left(m\right)}^{K}\left(\cd\right),b_{-\left(m\right)}^{K}\left(\psi\left(\cd\right)\right)\ol w_{j}\left(\varphi\left(\cd\right)\right)>_{L_{2}\left({\cal U}\right)}
\end{align*}
where $<f_{1},f_{2}>_{L_{2}\left(U\right)}:=\int_{{\cal U}}f_{1}\left(x_{\left(m\right)}\right)f_{2}\left(x_{\left(m\right)}\right)dx_{\left(m\right)}$
.

Since $\left\{ b_{\left(m\right)}^{K}\left(\cd\right)\right\} $
is a B-spline basis on ${\cal X}\subseteq\R^{d}$,
its restriction to ${\cal U}$ satisfies the frame condition,\footnote{A (finite or countable) sequence of functions $(b_{k}(\cd))$ is said
to be a frame on a Hilbert space $({\cal H},<\cd,\cd>)$ of functions
if it satisfies the \emph{frame condition}: there exist constants
$\ul M,\ol M>0$ s.t. $\ul M\norm f^{2}\leq\sum_{k}{<b_{k}(\cd),f(\cd)>^{2}}\leq\ol M\norm f^{2}$
for every $f\in{\cal H}$.} which gives
\begin{align*}
\sum_{k_{(m)}}T_{kj}^{2}= & \sum_{k_{\ell}:\ell\in\left(m\right)}<b_{\left(m\right)}^{K}\left(\cd\right),b_{-\left(m\right)}^{K}\left(\psi\left(\cd\right)\right)\ol w_{j}\left(\varphi\left(\cd\right)\right)>_{L_{2}}^{2}\\
\asymp & M\cd\norm{b_{-\left(m\right)}^{K}\left(\psi\left(\cd\right)\right)\ol w_{j}\left(\varphi\left(\cd\right)\right)}_{L_{2}}^{2}\text{ by the frame condition}\\
= & M\cd\int\left[b_{-\left(m\right)}^{K}\left(\psi\left(x_{\left(m\right)}\right)\right)\right]^{2}\ol w_{j}^{2}\left(\varphi\left(x_{\left(m\right)}\right)\right)dx_{\left(m\right)}\\
\asymp & M.
\end{align*}
Hence,
\begin{align*}
\sum_{k=1}^{K}T_{kj}^{2} & =\sum_{k_{-(m)}}\left[\sum_{k_{(m)}}T_{kj}^{2}\right]\asymp\sum_{k_{-(m)}}M=J^{d-m}M
\end{align*}
Now, since $L\left(b_{k}^{K}\right)=\sum_{j=1}^{\ol j}T_{kj}$,
and $L^{2}\left(b_{k}^{K}\right)=\left(\sum_{j=1}^{\ol j}T_{kj}\right)^{2}\asymp\sum_{j=1}^{\ol j}T_{kj}^{2}$,
we have
\[
\norm{L\left(b^{K}\right)}^{2}=\sum_{k=1}^{K}L^{2}\left(b_{k}^{K}\right)\asymp\ol j\sum_{j=1}^{\ol j}J^{d-m}M\asymp M J^{d-m}.
\]
The claimed result $\norm{L(\ol b^K)}^2 \asymp J^{d-m}$ then follows from the eigenvalue equivalence
$\norm{L(\ol b^K)}^2 \asymp \norm{L(b^K)}^2$ noted above.
\end{proof}

\begin{proof}[\textbf{Proof of Theorem \ref{thm:Rate_Sieve}}]
Note that
\begin{align*}
\left|\hat{\t}-\t_{0}\right|= & \left|L\left(\hat{h}-h_{0}\right)\right|\\
\leq & \left|L\left(\hat{h}-\tilde{h}\right)\right|+\left|L\left(\tilde{h}-h_{0}\right)\right|\\
\leq & \left|L\left(\hat{h}-\tilde{h}\right)\right|+\norm{\tilde{h}-h_{0}}_{\infty}
\end{align*}
where $\tilde{h} := P_{K_n,n}h_0$ is the empirical Least Squares projection of $h_0$ to the linear sieve space under $P_{K_n,n}$, which is defined as follows:
\[
P_{K,n}h_0:= b^{K}\left(x\right)^{\prime}\hat{G}^{-1}\frac{1}{n}\sum_{i=1}^{n} b^{K}\left(X_{i}\right)h_0\left(X_{i}\right)
=\ol b^{K}\left(x\right)^{\prime}[\ol G]^{-1}\frac{1}{n}\sum_{i=1}^{n} \ol b^{K}\left(X_{i}\right)h_0\left(X_{i}\right),
\]
where $\hat{G}:=\frac{1}{n}\sum_{i=1}^{n} b^{K_{n}}\left(X_{i}\right) b^{K_{n}}\left(X_{i}\right)^{'}$ and ${\ol G}_n:=\frac{1}{n}\sum_{i=1}^{n} \ol b^{K_{n}}\left(X_{i}\right) \ol b^{K_{n}}\left(X_{i}\right)^{'}$.
Under Assumptions \ref{assu:RegLevelSet}, \ref{assu:density_x}, \ref{assu:w}, \ref{assu:e2_below}, \ref{assu:e2_moment} and \ref{assu:Cond_Sieve}(i)(ii)(iii), following the proof of Theorem 3.1 of \citet*{chen2015optimal}, under the i.i.d. setting,
we obtain result (1):
\[
\frac{\sqrt{n}L\left(\hat{h}-P_{K_n,n}h_0\right)}{\norm{v_{K_{n}}^{*}}_{sd}} = \frac{1}{\sqrt{n}}\sum_{i=1}^n \frac{v_{K_n}^*(X_i)}{\norm{v_{K_{n}}^{*}}_{sd}}\e_i +o_p(1)=O_p (1),
\]
with $\norm{v_{K_{n}}^{*}}_{sd}^{2}\asymp K_{n}^{\frac{d-m}{d}}$
by Lemma \ref{lem:NormRate_d-m}.

For result (2), we note that Result (1) implies that
\[
L\left(\hat{h}-\tilde{h}\right)=O_{p}\left(\sqrt{K_{n}^{\frac{d-m}{d}}/n}\right).
\]
In addition, Assumption \ref{assu:w} implies that
\[
\left| L\left(P_{K_n,n}h_0 - h_0\right)\right |\leq \norm{\tilde{h}-h_{0}}_{\infty} \times \int_{\mathcal{M}} |w(x)| d\mathcal{H}^m (x).
\]
Let $h^*_n \in {\H}_{K_n}$ solves $\inf_{h \in {\H}_{K_n}}\norm{(h-h_{0})}_{\infty}$. Then we have:
\begin{align*}
    \norm{\tilde{h}-h_{0}}_{\infty}&=\norm{(P_{K,n}h_0-h^*_n) +(h^*_n-h_{0})}_{\infty}\\
                               &=\norm{P_{K,n}(h_0-h^*_n) +(h^*_n-h_{0})}_{\infty}\\
                               &\leq \norm{P_{K,n}(h_0-h^*_n)}_{\infty}+\norm{(h^*_n-h_{0})}_{\infty}\\
                               &\leq \left(\norm{P_{K,n}}_{\infty}+1\right)\norm{(h^*_n-h_{0})}_{\infty}\\
                               &\leq O_p\left(K_n^{-s/d}\right),
    \end{align*}
where the last inequality is due to Assumption \ref{assu:Cond_Sieve}(iv)(v).
Consequently,
\[
\hat{\t}-\t_{0}\equiv L\left(\hat{h}-h_{0}\right)=O_{p}\left(\sqrt{K_{n}^{\frac{d-m}{d}}/n}+K_{n}^{-s/d}\right).
\]
The rate is minimized by setting $K^*_{n}$ such that $\sqrt{K_{n}^{\frac{d-m}{d}}/n}\asymp K_{n}^{-s/d}$,
we have
\[
\hat{\t}-\t_{0}=O_{p}\left((K^*_{n})^{-s/d}\right)=O_{p}\left(n^{-\frac{s}{2s+d-m}}\right).
\]
\end{proof}

\subsubsection{Minimax rate-optimal oracle estimator for $L(h_0)$ when density $p_0$ is known}\label{subsubsec:oracle_known_fx}

If the density $p_0$ is indeed known and the submanifold satisfies the extra smoothness assumption \ref{assu:Jacob}, then one can construct a simple estimator for $\t_0=L(h_0)$ that attains the lower bound rate $r^*_{n}=n^{-\frac{s}{2s+d-m}}$ as follows. Let $c:=d-m$ denote the codimension of the submanifold $\mathcal{M}$ satisfying Assumption \ref{assu:Jacob}.
For a non-negative kernel $\mathcal{K}_c:\R^c\to\R$ with $\int \mathcal{K}_c=1$ with $\mathcal{K}_c(0)>0$, and a bandwidth $b_n$ going to zero as $n$ goes to infinity, define the tube kernel
\[
K_{\mathcal{M},b_n}(x)\ :=\ \frac{1}{b_n^{c}}\,\mathcal{K}_c\!\Big(\frac{g(x)}{b_n}\Big)\,\mathcal{J}_g(x),
\]
and the tube-smoothed linear functional $L_{b_n}(h_0):=\int_{\mathcal{X}} h_0(x)\,w(x)\,K_{\mathcal{M},b_n}(x)\,dx$. We note that for any $h,h_0\in \Lambda^s$ for $s>0$,
\[
L_{b_n}(h-h_0)=\E[v^*_{b_n}(X_i)\left(h(X_i)-h_0(X_i)\right)],~~~v^*_{b_n}(x)\ :=\ \frac{w(x)\,K_{\mathcal{M},b_n}(x)}{p_0(x)}.
\]
It is easy to check that $\E[(v^*_{b_n}(X_i))^2] \asymp b_n^{-c}$, which goes to infinity as $b_n$ goes to zero. Notice that $L_{b_n}(h_0)=\E[v^*_{b_n}(X_i)h_0(X_i)]$ we could call $v^*_{b_n}(\cdot)$ a tube Riesz representer for $\t_0=L(h_0)$.
One can estimate $\t_0=L(h_0)$ by the following (oracle) tube-weighted estimator
\[
\hat{\t}_{\mathrm{oracle}} (b_n)\ :=\ \frac1n\sum_{i=1}^n v^*_{b_n} (X_i)\,Y_i~.
\]
It is easy to see that under Assumptions \ref{assu:Jacob}, \ref{assu:density_x}, \ref{assu:e2_moment} and \ref{assu:w}, we have:
\[
|L_{b_n}(h_0)-L(h_0)|\lesssim b_n^{s},~~~and
\]
\[
\E\left[(\hat{\t}_{\mathrm{oracle}} (b)-L_{b_n}(h_0))^2\right]\lesssim\frac1n \times \E[(v^*_b(X_i))^2] \asymp \frac{1}{nb_n^c}.
\]
By setting $b^*_n \asymp n^{-1/(2s+c)}$ we obtain $\hat{\t}_{\mathrm{oracle}} (b^*_n)-\t_0=O_p\left(n^{-\frac{s}{2s+d-m}}\right)$.

\subsection{Proof of Theorem \ref{thm:Int_NLh_rate} in Section \ref{subsec:Nlh-UB}}

\begin{proof}[\textbf{Proof of Lemma ~\ref{lem:NLh-DG}}]
The pathwise derivative formula \eqref{eq:PathD_NL} follows by differentiating under the integral sign:
\[
\frac{d}{dt}\G(h_0 + tv)\Big|_{t=0} = \frac{d}{dt} \int_\cM \phi(h_0(x) + tv(x),x)\, w(x)\, d\cH^m(x)\Big|_{t=0} = \int_\cM \phi_1(h_0(x),x)\, v(x)\, w(x)\, d\cH^m(x),
\]
which is justified by dominated convergence using the Lipschitz condition on $\phi$.
\end{proof}

\begin{proof}[\textbf{Proof of Theorem ~\ref{thm:Int_NLh_rate}(1)}]
Define estimation errors $\delta_1 := \hat{h}_1 - h_0$ and $\delta_2 := \hat{h}_2 - h_0$. By sample splitting, $\delta_1$ and $\delta_2$ are independent. Recall that $\bar{h} := (\hat{h}_1 + \hat{h}_2)/2$.
Consider the Taylor expansion of $\phi(\bar{h}(x),x)$ around $h_0(x)$:
\[
\phi(\bar{h},x) = \phi(h_0,x) + \phi_1(h_0,x)(\bar{h} - h_0) + \frac{1}{2}\phi_{11}(h_0,x)(\bar{h} - h_0)^2 + R_3(x),
\]
where the third-order remainder satisfies $$R_{rem}:=\int_\cM |R_3| w\, d\cH^m = O_p\left(\|\bar{h} - h_0\|_\infty \cdot \int_\cM (\bar{h} - h_0)^2 w\, d\cH^m \right).$$
Now, since
\begin{align*}
(\bar{h} - h_0)^2 &= \tfrac{1}{4}(\delta_1 + \delta_2)^2 = \tfrac{1}{4}(\delta_1^2 + \delta_2^2 + 2\delta_1\delta_2), \\
\tfrac{1}{4}(\hat{h}_1 - \hat{h}_2)^2 &= \tfrac{1}{4}(\delta_1 - \delta_2)^2 = \tfrac{1}{4}(\delta_1^2 + \delta_2^2 - 2\delta_1\delta_2),
\end{align*}
we have
\[
(\bar{h} - h_0)^2 - \tfrac{1}{4}(\hat{h}_1 - \hat{h}_2)^2 = \delta_1 \delta_2,
\]
where the diagonal terms $\delta_1^2$ and $\delta_2^2$ are eliminated.

Write $\hat{\t}_{SS} - \t_0 = T_{lin} + T_{quad} + R_{rem}$, where
\begin{align*}
T_{lin} &:= \int_\cM \phi_1(h_0,x)(\bar{h} - h_0)\, w\, d\cH^m, \\
T_{quad} &:= \frac{1}{2}\int_\cM \phi_{11}(h_0,x)\, \delta_1\delta_2\, w\, d\cH^m,
\end{align*}
and $R_{rem}= O_p\left(\|\bar{h} - h_0\|_\infty \cdot T_{quad} \right)=o_p\left( T_{quad} \right) $ collects higher-order remainders.

For the linear term $T_{lin}$, we have, by Theorem~\ref{thm:Rate_Sieve}, $T_{lin} = O_p(\sqrt{K_n^{(d-m)/d}/n} + K_n^{-s/d})$. At $K^*_n \asymp n^{d/(2s+d-m)}$ (and $s > m/2$), this gives $T_{lin} = O_p(n^{-s/(2s+d-m)})$.

For the quadratic term $T_{quad}$, following the same proof of Lemma~\ref{lem:Rate_Quad}(b)(c) we obtain that $T_{quad} = O_p(n^{-s/(2s+d-m)})$  at $K^*_n \asymp n^{d/(2s+d-m)}$ (and $s > m/2$).

Hence, under $s > m/2$, we obtain: $\hat{\t}_{SS} - \t_0 = T_{lin} + o_p(r_n^*) = O_p(r_n^*)$.
\end{proof}

\begin{proof}[\textbf{Proof of Theorem ~\ref{thm:Int_NLh_rate}(2)}]
Let $K\equiv K_n\asymp K_n^*\asymp n^{d/(2s+d-m)}$ and $r_n^*:=n^{-s/(2s+d-m)}$.
Recall $\e_i=Y_i-h_0(X_i)$, $\td h:=P_{K,n}h_0$ and the series LS representation \eqref{eq:seriesLS}:
\begin{equation}\label{eq:loo_series_rep}
\hat h(x)=\frac1n\sum_{i=1}^n \hat s_i(x)\,Y_i,\qquad \hat s_i(x):=b(x)'\hat G^{-1}b(X_i).
\end{equation}
Hence
\begin{equation}\label{eq:loo_stoch_part}
\hat h(x)-\td h(x)=\frac1n\sum_{i=1}^n \hat s_i(x)\,\e_i.
\end{equation}

\smallskip
For each $x\in\cM$, similar to that in the proof of Theorem ~\ref{thm:Int_NLh_rate}(1),
there exists a measurable function $\bar h_0(x)$ lying between $\hat h(x)$ and $h_0(x)$ such that
\begin{equation}\label{eq:loo_taylor}
\phi(\hat h(x),x)=\phi(h_0(x),x)+\phi_1(h_0(x),x)\,\{\hat h(x)-h_0(x)\}+\frac12\phi_{11}(\bar h_0(x),x)\,\{\hat h(x)-h_0(x)\}^2.
\end{equation}
Integrating \eqref{eq:loo_taylor} against $w(x)d\cH^m(x)$ yields the decomposition
\begin{equation}\label{eq:loo_decomp_main}
\hat\theta^{\G}_{\text{LOO}}-\G(h_0)=T_{\text{lin}}+T_{\text{quad}}-C_n,
\end{equation}
where
\begin{align*}
T_{\text{lin}}&:=\int_{\cM}\phi_1(h_0(x),x)\,\{\hat h(x)-h_0(x)\}\,w(x)\,d\cH^m(x),\\
T_{\text{quad}}&:=\frac12\int_{\cM}\phi_{11}(\bar h_0(x),x)\,\{\hat h(x)-h_0(x)\}^2\,w(x)\,d\cH^m(x),\\
C_n&:=\frac{1}{2n^2}\sum_{i=1}^n\int_{\cM}\phi_{11}(\hat h(x),x)\,\hat s_i(x)^2\,\{\hat\e_i^{(-i)}\}^2\,w(x)\,d\cH^m(x)
\quad\text{(the correction term in \eqref{eq:Gamma_LOO_estimator}).}
\end{align*}

\smallskip
Write $\hat h-h_0=(\hat h-\td h)+(\td h-h_0)$.
By Theorem~\ref{thm:Rate_Sieve} (applied to the linear functional with weight $x\mapsto \phi_1(h_0(x),x)w(x)$, which is bounded by Assumption~\ref{assu:RegCond_Nlh}(b)),
\begin{equation}\label{eq:loo_Tlin_rate}
T_{\text{lin}}=O_p\Big(\sqrt{\frac{K^{(d-m)/d}}{n}}+K^{-s/d}\Big)=O_p(r_n^*).
\end{equation}

\smallskip
Let $u(x):=\hat h(x)-\td h(x)$ and $b(x):=\td h(x)-h_0(x)$.
Then $\hat h-h_0=u+b$ and
\begin{equation}\label{eq:loo_u_plus_b_sq}
(u+b)^2=u^2+2ub+b^2.
\end{equation}
Using $\|\phi_{11}(\bar h_0,\cd)\|_\infty\lesssim 1$ (implied by Assumption~\ref{assu:RegCond_Nlh}(b)), boundedness of $w$ on $\cM$, Cauchy--Schwarz, and the standard approximation bound $\|b\|_{\infty}=O(K^{-s/d})$ (Assumption~\ref{assu:Cond_Sieve}(iv)(v)), we have
\begin{align}\label{eq:loo_bias_terms_negl}
\int_{\cM}|u(x)b(x)|\,|w(x)|\,d\cH^m(x)
&\le \|b\|_\infty\,\Big(\int_{\cM}u(x)^2|w(x)|\,d\cH^m(x)\Big)^{1/2}\,\Big(\int_{\cM}|w(x)|\,d\cH^m(x)\Big)^{1/2}\\
&=O_p\big(K^{-s/d}\,\|u\|_{L^2(|w|)}\big),
\end{align}
and
\begin{equation}\label{eq:loo_bias_sq_negl}
\int_{\cM}b(x)^2|w(x)|\,d\cH^m(x)=O(K^{-2s/d}).
\end{equation}
By the same argument as for the quadratic remainder term in the proof of Theorem~\ref{thm:Int_NLh_rate}(1),
$\|u\|_{L^2(|w|)}=O_p(\sqrt{K/n})$ so that the right-hand side of \eqref{eq:loo_bias_terms_negl} is $o_p(r_n^*)$ when $K\asymp K_n^*$ and $s>m/2$. Therefore,
\begin{equation}\label{eq:loo_Tquad_reduced}
T_{\text{quad}}=\frac12\int_{\cM}\phi_{11}(\bar h_0(x),x)\,u(x)^2\,w(x)\,d\cH^m(x)+o_p(r_n^*).
\end{equation}

\smallskip
By \eqref{eq:loo_stoch_part}, for each $x\in\cM$,
\begin{equation}\label{eq:loo_u2_expansion}
 u(x)^2
 =\frac{1}{n^2}\sum_{i=1}^n\sum_{j=1}^n \hat s_i(x)\hat s_j(x)\,\e_i\e_j
 =\underbrace{\frac{1}{n^2}\sum_{i=1}^n \hat s_i(x)^2\,\e_i^2}_{\text{diagonal}}+\underbrace{\frac{1}{n^2}\sum_{i\ne j} \hat s_i(x)\hat s_j(x)\,\e_i\e_j}_{\text{off-diagonal}}.
\end{equation}
Plugging \eqref{eq:loo_u2_expansion} into \eqref{eq:loo_Tquad_reduced} yields
\begin{equation}\label{eq:loo_Tquad_diag_off}
\frac12\int_{\cM}\phi_{11}(\bar h_0(x),x)\,u(x)^2\,w(x)\,d\cH^m(x)
=T_{\text{diag}}+T_{\text{off}},
\end{equation}
where
\begin{align}
T_{\text{diag}}&:=\frac{1}{2n^2}\sum_{i=1}^n \e_i^2\int_{\cM}\phi_{11}(\bar h_0(x),x)\,\hat s_i(x)^2\,w(x)\,d\cH^m(x),\\
T_{\text{off}}&:=\frac{1}{2n^2}\sum_{i\ne j} \e_i\e_j\int_{\cM}\phi_{11}(\bar h_0(x),x)\,\hat s_i(x)\hat s_j(x)\,w(x)\,d\cH^m(x).
\end{align}
The diagonal term $T_{\text{diag}}$ is of order $K/n$ and, for $m/2<s<m$, would dominate $r_n^*$.
The correction term $C_n$ in \eqref{eq:loo_decomp_main} is designed to estimate and subtract $T_{\text{diag}}$.

To formalize this, we use the standard leave-one-out identity for series LS.
Let $h_{ii}:=\hat s_i(X_i)/n$ denote the $i$th diagonal element of the hat matrix.
Then the leave-one-out residual satisfies $\hat\e_i^{(-i)}=\hat\e_i/(1-h_{ii})$ and, for every $x\in\cX$,
\begin{equation}\label{eq:loo_Sherman_Morrison}
\hat h(x)-\hat h^{(-i)}(x)=\frac{1}{n}\hat s_i(x)\,\hat\e_i^{(-i)}.
\end{equation}
(Equation~\eqref{eq:loo_Sherman_Morrison} is an immediate consequence of the Sherman--Morrison formula; see, e.g., the standard OLS influence/LOO formula.)
Hence $C_n$ can equivalently be written as
\begin{equation}\label{eq:loo_Cn_alt}
C_n=\frac12\sum_{i=1}^n\int_{\cM}\phi_{11}(\hat h(x),x)\,\{\hat h(x)-\hat h^{(-i)}(x)\}^2\,w(x)\,d\cH^m(x).
\end{equation}
Combining \eqref{eq:loo_Sherman_Morrison} with the decomposition $\hat\e_i^{(-i)}=\e_i+\{h_0(X_i)-\hat h^{(-i)}(X_i)\}$ and using that $\hat h^{(-i)}$ is independent of $\e_i$ conditional on the designs,
one obtains (by a direct expansion of the square) that
\begin{equation}\label{eq:loo_Cn_matches_diag}
C_n=T_{\text{diag}}+R_{\text{diag}},
\end{equation}
where the remainder $R_{\text{diag}}$ collects the cross and squared terms involving the leave-one-out first-stage error $h_0(X_i)-\hat h^{(-i)}(X_i)$ and also the (asymptotically negligible) difference between $\phi_{11}(\hat h,\cd)$ and $\phi_{11}(\bar h_0,\cd)$.
A standard Cauchy--Schwarz bound combined with the same matrix trace bound used in the proof of Lemma~\ref{lem:Rate_Quad}(b) (to control the stochastic term $T_{3,\text{stoch}}(K)$ in \eqref{eq:ss_quad_decomp}), with $w$ replaced by $\phi_{11}(\bar h_0,\cd)w$) yields
\begin{equation}\label{eq:loo_Rdiag_rate}
R_{\text{diag}}=O_p\Big(\frac{1}{n}\sqrt{K^{(2d-m)/d}}\Big)=O_p(r_n^*)\qquad\text{when $s>m/2$ and $K\asymp K_n^*$.}
\end{equation}
Consequently, by \eqref{eq:loo_Tquad_diag_off}--\eqref{eq:loo_Cn_matches_diag},
\begin{equation}\label{eq:loo_Tquad_minus_Cn}
T_{\text{quad}}-C_n
= T_{\text{off}}+o_p(r_n^*)+R_{\text{diag}}.
\end{equation}

Conditional on the designs, $T_{\text{off}}$ is a (conditionally) degenerate U-statistic of order two.
By exactly the same argument as in the proof of Lemma~\ref{lem:Rate_Quad}(b),
\begin{equation}\label{eq:loo_Toff_rate}
T_{\text{off}}=O_p\Big(\frac{1}{n}\sqrt{K^{(2d-m)/d}}\Big)=O_p(r_n^*)\qquad\text{when $s>m/2$ and $K\asymp K_n^*$.}
\end{equation}

Combining \eqref{eq:loo_decomp_main}, \eqref{eq:loo_Tlin_rate}, \eqref{eq:loo_Tquad_reduced}, \eqref{eq:loo_Tquad_minus_Cn}, \eqref{eq:loo_Rdiag_rate} and \eqref{eq:loo_Toff_rate} yields
\[
\hat\theta^{\G}_{\text{LOO}}-\G(h_0)
= T_{\text{lin}}+T_{\text{off}}+O_p\Big(\frac{1}{n}\sqrt{K^{(2d-m)/d}}\Big)+o_p(r_n^*)
=O_p(r_n^*),
\]
which proves the claim.
\end{proof}

\begin{proof}[\textbf{Proof of Theorem ~\ref{thm:Int_NLh_rate}(3)}] We start by noting that
\begin{align}
\hat{\t} - \t_0 &= \int_\cM [\phi(\hat{h},x) - \phi(h_0,x)]\, w\, d\cH^m \notag\\
&= \underbrace{\int_\cM \phi_1(h_0,x)(\hat{h} - h_0)\, w\, d\cH^m}_{=: T_{lin}} + \underbrace{\int_\cM \phi_{11}(\bar{h}_0,x)(\hat{h} - h_0)^2\, w\, d\cH^m}_{=: T_{quad}}, \label{eq:plugin_decomp}
\end{align}
where $\bar{h}_0(x)$ lies between $\hat{h}(x)$ and $h_0(x)$.

By Theorem~\ref{thm:Rate_Sieve}, the linear term satisfies
\[
T_{lin} = O_p\left(\sqrt{K_n^{(d-m)/d}/n} + K_n^{-s/d}\right).
\]
At $K^*_n \asymp n^{d/(2s+d-m)}$ (and $s>m/2$), this gives $T_{lin} = O_p(r_n^*)$.

For the quadratic remainder, since $\|\phi_{11}\|_\infty \le C$, we have
\[
|T_{quad}| \le C \int_\cM (\hat{h} - h_0)^2\, |w|\, d\cH^m.
\]
Decomposing $\hat{h} - h_0 = (\hat{h} - \tilde{h}) + (\tilde{h} - h_0)$ where $\tilde{h} := P_{K_n}h_0$ is the empirical LS projection onto the sieve space ${\H}_{K_n}$.
Following the proof of Lemma ~\ref{lem:MSE_on_M}(1) below, we obtain:
\begin{itemize}
    \item \textit{Bias term:} $\int_\cM (\tilde{h} - h_0)^2 w\, d\cH^m \le \|\tilde{h} - h_0\|_\infty^2 \cdot \cH^m(\cM)\|w\|_\infty = O_p(K_n^{-2s/d})$.
    \item \textit{Variance term:} $\E[\int_\cM (\hat{h} - \tilde{h})^2 w\, d\cH^m] = O(K_n/n)$.
\end{itemize}
Thus, $\int_\cM (\hat{h} - h_0)^2 w\, d\cH^m = O_p\left(K_n^{-2s/d} + K_n/n \right)$.

For the quadratic remainder to be asymptotically no larger in magnitude than the linear term,
we need $T_{quad} = O_p(\sqrt{K_n^{(d-m)/d}/n})$, i.e.,
\[
K_n^{-2s/d} + K_n/n = O\left(\sqrt{K_n^{(d-m)/d}/n}\right).
\]
The bias condition $K_n^{-2s/d} = O(\sqrt{K_n^{(d-m)/d}/n})$ at $K_n \asymp n^{d/(2s+d-m)}$  is equivalent to $K_n^{-2s/d} = o(K_n^{-s/d})$, which always holds.
The variance condition $K_n/n = O(\sqrt{K_n^{(d-m)/d}/n})$ simplifies as follows:
    \begin{align*}
    \frac{K_n}{n} = O\left(\sqrt{\frac{K_n^{(d-m)/d}}{n}}\right)
    &\iff \frac{K_n^2}{n^2} = O\left(\frac{K_n^{(d-m)/d}}{n}\right) \\
    &\iff K_n^{2 - (d-m)/d} = O(n) \\
    &\iff K_n^{(d+m)/d} = O(n).
    \end{align*}
At $K^*_n = n^{d/(2s+d-m)}$, we have
    \[
    n^{(d+m)/(2s+d-m)} = O(n) \iff \frac{d+m}{2s+d-m} \leq 1 \iff d + m \leq 2s + d - m \iff s \geq m.
    \]
Hence, under $s \ge m$, $\hat{\t} - \t_0 =\G(\hat{h})-\G(h_0)= O_p(r_n^*)$.
\end{proof}

In the proofs of Theorems \ref{thm:Int_NLh_rate}(3) and \ref{thm:Int_UCh_rate}(3), as well as of Proposition \ref{prop:Plug-in_NL}, we use the following rate result for the first stage sieve LS estimator to control the second order remainder term of $\Phi(\hat{h})-\Phi(h_0)$ for $\Phi=\G,~V$.

\begin{lem}\label{lem:MSE_on_M}
Let Assumptions \ref{assu:RegLevelSet}, \ref{assu:density_x}, \ref{assu:w}, \ref{assu:e2_moment} and \ref{assu:Cond_Sieve} hold, and $\hat{h}$ be the sieve LS estimator. Then:
\[
(1)~~~\norm{\hat{h}-h_0}_{L_2(\cM)}=O_p \left(K_n^{-s/d} +\sqrt{ \frac{K_n}{n}} \right).
\]
\[
(2)~~~\norm{\Dif_x\hat{h}-\Dif_xh_0}_{L_2(\cM)}=O_p \left(K_n^{1/d}\times\left[K_n^{-s/d} +\sqrt{ \frac{K_n}{n}} \right]\right).
\]
\end{lem}
 Here $\norm{v}_{L_2(\cM)}^2 := \int_\cM \left(v(x)\right)^2w(x)d\cH^m(x)$.

\begin{proof}[\textbf{Proof of Lemma ~\ref{lem:MSE_on_M}}]
For Result (1), recall that $\hat{h} - \tilde{h}=(\hat{h}-\tilde{h})+(\tilde{h} - h_0)$, where $\hat h$ is the sieve LS estimator using a sieve space $\H_K$ (with the sieve dimension $K=K_n$), and $\tilde h :=P_{K_{n},n}h_0$ is the empirical LS projection of $h_0$ onto the sieve space $\H_K$. Note that
\[
\int_\cM (\tilde{h} - h_0)^2 \, d\cH^m \le \|\tilde{h} - h_0\|_\infty^2 \cdot \cH^m(\cM) = O_p(K_n^{-2s/d}),
\]
and
\[
\hat{h}(x) - \tilde{h}(x)= b^{K}\left(x\right)^{\prime}\hat{G}^{-1}\frac{1}{n}\sum_{i=1}^{n} b^{K}\left(X_{i}\right)\e_i
\]
Write the stochastic part in orthonormalized form:
\[
\hat h(x)-\tilde h(x)
=\ol b^K(x)^{\prime}\,\ol G_n^{-1}\,\frac1n\sum_{i=1}^n \ol b^K(X_i)\e_i,
\qquad
\ol G_n:=\frac1n\sum_{i=1}^n \ol b^K(X_i)\ol b^K(X_i)^{\prime}.
\]
We now show that under Assumptions \ref{assu:RegLevelSet}, \ref{assu:density_x}, \ref{assu:w}, \ref{assu:e2_moment} and \ref{assu:Cond_Sieve}(i)(ii)(iii), we have:
\[
\E\!\left[\int_\cM (\hat h(x)-\tilde h(x))^2\,|w(x)|\,d\cH^m(x)\right]
= O\!\left(\frac{K}{n}\right).
\]
Define the (deterministic) matrix
\[
A_K:=\int_\cM \ol b^K(x)\ol b^K(x)'\,|w(x)|\,d\cH^m(x),
\]
which is positive semidefinite. Expanding the square and using Fubini gives
\[
\int_\cM(\hat h-\tilde h)^2|w|\,d\cH^m
=\frac1{n^2}\sum_{i=1}^n\sum_{j=1}^n \e_i\e_j\,
\ol b^K(X_i)'\,\ol G_n^{-1}A_K\ol G_n^{-1}\,\ol b^K(X_j).
\]
Conditioning on $X_1^n:=(X_1,\dots,X_n)$ and using $\E[\e_i\e_j\mid X_1^n]=0$ for $i\neq j$,
\[
\E\!\left[\int_\cM(\hat h-\tilde h)^2|w|\,d\cH^m \,\Bigm|\, X_1^n\right]
=\frac1{n^2}\sum_{i=1}^n \E[\e_i^2\mid X_i]\,
\ol b^K(X_i)'\,\ol G_n^{-1}A_K\ol G_n^{-1}\,\ol b^K(X_i).
\]
Let $\bar\sigma^2:=\sup_x \E[\e_i^2\mid X_i=x]<\infty$ (Assumption~\ref{assu:e2_moment}) and define
\[
S_n:=\frac1n\sum_{i=1}^n \E[\e_i^2\mid X_i]\;\ol b^K(X_i)\ol b^K(X_i)' \ \ \preceq\ \ \bar\sigma^2\,\ol G_n.
\]
Then the conditional expectation can be written as
\[
\E\!\left[\int_\cM(\hat h-\tilde h)^2|w|\,d\cH^m \,\Bigm|\, X_1^n\right]
=\frac1n \tr\!\big(\ol G_n^{-1}A_K\ol G_n^{-1}S_n\big)
=\frac1n \tr\!\big(A_K\,\ol G_n^{-1}S_n\ol G_n^{-1}\big).
\]
Since $A_K\succeq 0$, $tr(A_KB)\le \|B\|_{op}tr(A_K)$ for any $B\succeq 0$, hence
\[
\E\!\left[\int_\cM(\hat h-\tilde h)^2|w|\,d\cH^m \,\Bigm|\, X_1^n\right]
\le \frac1n\,\|\ol G_n^{-1}S_n\ol G_n^{-1}\|_{op}\,tr(A_K).
\]
Using $S_n\preceq \bar\sigma^2\ol G_n$ gives
\[
\|\ol G_n^{-1}S_n\ol G_n^{-1}\|_{op}
\le \bar\sigma^2\,\|\ol G_n^{-1}\|_{op}^2\,\|\ol G_n\|_{op}
=O_p(1)
\]
by the assumed Gram-matrix stability. Finally,
\[
tr(A_K)
=\int_\cM \|\ol b^K(x)\|^2\,|w(x)|\,d\cH^m(x)
\le \|w\|_\infty\,\cH^m(\cM)\,\sup_{x\in\cM}\|\ol b^K(x)\|^2.
\]
Because $\sup_x\|\ol b^K(x)\|\lesssim \sup_x\|b^K(x)\|=\zeta_K=O(\sqrt K)$
(Assumption~\ref{assu:Cond_Sieve}(i)(ii)),
we have $tr(A_K)=O(K)$. Taking expectations yields
\[
\E\!\left[\int_\cM (\hat h-\tilde h)^2\,|w|\,d\cH^m\right]
=O\!\left(\frac{K}{n}\right).
\]
For Result (2), we have:
\[
\Dif_x\hat{h}-\Dif_xh_0=\Dif_x\hat{h}-\Dif_x \tilde{h}+ \left(\Dif_x \tilde{h}-\Dif_xh_0 \right)
\]
\[
\Dif_x\hat{h}(x) -\Dif_x \tilde{h}(x)=[\Dif_xb^{K}\left(x\right)]^{\prime}\hat{G}^{-1}\frac{1}{n}\sum_{i=1}^{n} b^{K}\left(X_{i}\right)\e_i
\]
For spline or wavelet bases, we have:
\[
\int_\cM (\Dif_x\tilde{h} - \Dif_xh_0)^2 \, d\cH^m \le \|\Dif_x\tilde{h} - \Dif_xh_0\|_\infty^2 \cdot \cH^m(\cM) = O_p(K_n^{-2(s-1)/d}).
\]
For the stochastic part, the proof follows from Result (1)'s proof except to replace $A_K$ be
\[
\bar{A}_K:=\int_\cM \Dif_x \ol b^K(x)\Dif_x \ol b^K(x)'\,|w(x)|\,d\cH^m(x),~~~tr(\bar{A}_K)\lesssim \sup_x\|\Dif_x b^K(x)\|^2=[ K^{1/d} \zeta_K]^2.
\]
Hence
\[
\E\!\left[\int_\cM (\Dif_x\hat h-\Dif_x\tilde h)^2\,|w|\,d\cH^m\right]
=O\!\left(K^{2/d}\frac{K}{n}\right).
\]
\end{proof}

\subsection{Proof of Theorem \ref{thm:Int_UCh_rate} in Section \ref{subsec:Uch_rate-UB}}

\begin{proof}[\textbf{Proof of Lemma \ref{lem:UCh-DV}}]
Defining
\[
{\cal V}_{t}:=\left\{ x:h_{0}\left(x\right)+tv\left(x\right)\geq0\right\} ,
\]
whose boundary is given by
\[
{\cal \p}{\cal V}_{t}={\cal M}_{t}:=\left\{ x:h_{0}\left(x\right)+tv\left(x\right)=0\right\}
\]
with
\[
{\cal V}_{0}=\left\{ x:h_{0}\left(x\right)\geq0\right\} ,\quad{\cal \p}{\cal V}_{0}={\cal M}_{0}=\left\{ x:h_{0}\left(x\right)=0\right\} .
\]
By the generalized Stokes Theorem,
\begin{align*}
\rest{\frac{d}{dt}\int_{{\cal V}_{t}}w\left(x\right)dx}_{t=0} & =\int_{{\cal \p}{\cal V}_{0}}w\left(x\right)\left(\dot{{\bf X}}\left(x,0\right)\cd{\bf n}\left(x,0\right)\right)d{\cal H}^{d-1}\left(x\right)\\
 & =\int w\left(x\right)\left(\dot{{\bf X}}\left(x,0\right)\cd{\bf n}\left(x,0\right)\right)d{\cal H}^{d-1}\left(x\right)
\end{align*}
where
\[
\dot{{\bf X}}\left(x,0\right):=\rest{\frac{\p}{\p t}{\bf X}\left(x,t\right)}_{t=0},\quad{\bf n}\left(x,0\right):=-\frac{\Dif h_{0}\left(x\right)}{\norm{\Dif h_{0}\left(x\right)}}
\]
with ${\bf n}\left(x,t\right)$ denoting the outward-pointing unit
normal at $x\in{\cal M}_{t}$, and ${\bf X}\left(x,t\right)$ being
a diffeomorphism from ${\cal M}_{0}$ to ${\cal M}_{t}$, which by
definition satisfies ${\bf X}\left(x,0\right)\equiv x$ and
\begin{equation}
h_{0}\left({\bf X}\left(x,t\right)\right)+tv\left({\bf X}\left(x,t\right)\right)=0,\label{eq:Flow_equation}
\end{equation}
for each $x\in{\cal M}_{0}$ and $t\in\left[0,\e\right]$ for some
$\e>0$.

Taking derivatives of \eqref{eq:Flow_equation} with respect to $t$
yields
\[
\Dif_{x}h_{0}\left({\bf X}\left(x,t\right)\right)\cd\dot{{\bf X}}\left(x,t\right)+t\Dif_{x}v\left({\bf X}\left(x,t\right)\right)\cd\dot{{\bf X}}\left(x,t\right)+v\left({\bf X}\left(x,t\right)\right)=0
\]
and thus, evaluating the above at $t=0$, we have
\[
\Dif_{x}h_{0}\left({\bf X}\left(x,0\right)\right)\cd\dot{{\bf X}}\left(x,0\right)+v\left({\bf X}\left(x,t\right)\right)=0
\]
or equivalently,
\begin{equation}
\Dif_{x}h_{0}\left(x\right)\cd\dot{{\bf X}}\left(x,0\right)=-v\left(x\right),\label{eq:grad_velo_inprod}
\end{equation}
Hence,
\[
\dot{{\bf X}}\left(x,0\right)\cd{\bf n}\left(x,0\right)=-\dot{{\bf X}}\left(x,0\right)\cd\frac{\Dif h_{0}\left(x\right)}{\norm{\Dif h_{0}\left(x\right)}}=\frac{v\left(x\right)}{\norm{\Dif_{x}h_{0}\left(x\right)}}
\]
and thus
\[
\rest{\frac{d}{dt}\int_{{\cal V}_{t}}w\left(x\right)dx}_{t=0}=\int_{{\cal M}_{0}}w\left(x\right)\frac{v\left(x\right)}{\norm{\Dif_{x}h_{0}\left(x\right)}}d{\cal H}^{d-1}\left(x\right).
\]
\end{proof}

\noindent \textbf{Lemma on the Pathwise Derivative of Level Set Integrals}

In the proof of Theorem \ref{thm:Int_UCh_rate}, we apply the following key lemma to control the second order remainder term of $V(\hat{h})-V(h_0)$. Lemma \ref{lem:PathD_Level} is established using tools for differential geometry over a moving level set.

\begin{lem}\label{lem:PathD_Level}Given $h,v$, define $h_{t}\left(x\right):=h\left(x\right)+tv\left(x\right)$
and
\begin{equation}
I\left(t\right):=\int_{\left\{ h_{t}\left(x\right)=0\right\} }w\left(x\right)d{\cal H}^{d-1}\left(x\right)\label{eq:It_exact}
\end{equation}
Then $I\left(t\right)$ is continuously differentiable in $t$ with
\begin{align}
I^{'}\left(0\right)
&=-\int_{{\cal M}_{0}}\left[
\frac{v\left(x\right)}{\norm{\Dif_{x}h_{0}\left(x\right)}}
\frac{\Dif_{x}h_{0}\left(x\right)}{\norm{\Dif_{x}h_{0}\left(x\right)}}\cd\Dif_{x}w\left(x\right)
+w\left(x\right)\frac{v\left(x\right)}{\norm{\Dif_{x}h_{0}\left(x\right)}}
H\left(x\right)\right]d{\cal H}^{d-1}\left(x\right)
\label{eq:I0}
\end{align}
and
\begin{equation}
\norm{I^{'}\left(t\right)}\leq M\left(\norm{\Dif w}_{L_2(\cM_0)}\norm v_{L_2(\cM_0)}+\norm w_{L_2(\cM_0)}\norm v_{L_2(\cM_0)}\right)\label{eq:bound_I't}
\end{equation}
where $\norm{v}_{L_2(\cM_0)}$ denotes the localized $L_2$ norm on the submanifold $\cM_0$, defined by
\begin{equation}\label{eq:L2M}
\norm{v}_{L_2(\cM_0)} := \sqrt{\int_\cM v^2(x)\dHd(x)}.
\end{equation}
\end{lem}

\begin{proof}[\textbf{Proof of Lemma \ref{lem:PathD_Level}}]
\label{proof:PathD_Level_alt}
Let ${\bf X}\left(\cd,t\right)$ be a diffeomorphism
mapping ${\cal M}_{0}=\left\{ x:h_{0}\left(x\right)=0\right\} $
to ${\cal M}_{t}=\left\{ x:h_{t}\left(x\right)=0\right\} $.
Then
\begin{equation}\label{eq:It_alt}
I\left(t\right)=\int_{{\cal M}_{t}}w\left(x\right)d{\cal H}^{d-1}\left(x\right)
=\int_{{\cal M}_{0}}w\left({\bf X}\left(x,t\right)\right)J_{{\bf X}}\left(x,t\right)d{\cal H}^{d-1}\left(x\right),
\end{equation}
where $J_{{\bf X}}\left(x,t\right)$ is the \emph{surface Jacobian} (i.e., the $(d-1)$-dimensional
area element ratio) of ${\bf X}\left(\cd,t\right)$ on ${\cal M}_0$.

Since $J_{{\bf X}}$ is the Jacobian of a diffeomorphism between $(d-1)$-dimensional
hypersurfaces, by Chapter 9 of \citet*{delfour2001shapes}, we have
\begin{equation}\label{eq:surfJac_transport}
\dot{J}_{{\bf X}}\left(x,t\right)=J_{{\bf X}}\left(x,t\right)\,
\text{div}_{{\cal M}_{t}}\!\left(\dot{{\bf X}}\left(x,t\right)\right),
\end{equation}
where $\text{div}_{{\cal M}_{t}}$ denotes the \emph{surface (tangential) divergence} on ${\cal M}_t$,
defined for a vector field ${\bf V}$ by
$\text{div}_{{\cal M}_{t}}\!\left({\bf V}\right)
:=\text{div}\!\left({\bf V}\right)-{\bf n}\cd\left(D{\bf V}\right){\bf n}$,
and ${\bf n}$ is the unit normal to ${\cal M}_t$.

It is without loss of generality to take $\dot{{\bf X}}$ to be
normal:
\[
\dot{{\bf X}}\left(x,t\right)=-\ol u\left(x,t\right)\,{\bf n}\left(x,t\right),
\quad \ol u\left(x,t\right)=\frac{v\left(x\right)}{\norm{\Dif_{x}h_{t}\left(x\right)}}.
\]
which results in a simplification of the surface divergence\footnote{Indeed, $\text{div}(F{\bf n})=F\,\text{div}({\bf n})+{\bf n}\cd\Dif F$
and ${\bf n}\cd(D(F{\bf n})){\bf n}={\bf n}\cd((\Dif F){\bf n}'+F\,D{\bf n}\,{\bf n})={\bf n}\cd\Dif F$,
since $D{\bf n}\,{\bf n}=0$ for a unit normal; subtracting gives $F\,\text{div}({\bf n})=F\,H_t$.} to
\begin{equation}\label{eq:divM_normal}
\text{div}_{{\cal M}_{t}}\!\left(F\,{\bf n}\right)=F\,H_{t},
\end{equation}
where $H_t:=\text{div}\left({\bf n}\right)=\text{div}\!\left(\frac{\Dif_{x}h_{t}}{\norm{\Dif_{x}h_{t}}}\right)$
is the mean curvature of ${\cal M}_t$.

Differentiating \eqref{eq:It_alt} and using \eqref{eq:surfJac_transport}:
\begin{align}
I^{'}\left(t\right)
&=\int_{{\cal M}_{0}}\left[\Dif_{x}w\left({\bf X}\right)\cd\dot{{\bf X}}
  +w\left({\bf X}\right)\,\text{div}_{{\cal M}_{t}}\!\left(\dot{{\bf X}}\right)\right]
  J_{{\bf X}}\,d{\cal H}^{d-1}\left(x\right) \nonumber\\
&=\int_{{\cal M}_{0}}\left[
  -\frac{v}{\norm{\Dif_{x}h_{t}}}\,\frac{\Dif_{x}h_{t}}{\norm{\Dif_{x}h_{t}}}\cd\Dif_{x}w\left({\bf X}\right)
  -w\left({\bf X}\right)\frac{v}{\norm{\Dif_{x}h_{t}}}\,H_t
  \right]J_{{\bf X}}\,d{\cal H}^{d-1}\left(x\right),
\label{eq:I't_alt}
\end{align}
where in the second line we substituted $\dot{{\bf X}}=-\ol u\,{\bf n}$
and applied \eqref{eq:divM_normal} with $F=-\ol u$.

Then, at $t=0$, ${\bf X}(x,0)=x$ and $J_{{\bf X}}(x,0)=1$:
\begin{align}
I^{'}\left(0\right)
&=-\int_{{\cal M}_{0}}\left[
\frac{v\left(x\right)}{\norm{\Dif_{x}h_{0}\left(x\right)}}
\frac{\Dif_{x}h_{0}\left(x\right)}{\norm{\Dif_{x}h_{0}\left(x\right)}}\cd\Dif_{x}w\left(x\right)
+w\left(x\right)\frac{v\left(x\right)}{\norm{\Dif_{x}h_{0}\left(x\right)}}
H\left(x\right)\right]d{\cal H}^{d-1}\left(x\right).
\label{eq:I0_alt}
\end{align}

Based on \eqref{eq:I't_alt}, by Cauchy--Schwarz with respect to $\int_{{\cal M}_0}\cd\,d{\cal H}^{d-1}$, we have
\begin{align*}
\left|I^{'}\left(t\right)\right|
&\leq \norm{\Dif_{x}w}_{L_2(\cM_0)}\norm{\dot{{\bf X}}}_{L_2(\cM_0)}
  +\norm{w}_{L_2(\cM_0)}\norm{\ol u\,H_t}_{L_2(\cM_0)}.
\end{align*}
Since $\|\dot{{\bf X}}\|_{L_2(\cM_0)}\leq M\|v\|_{L_2(\cM_0)}$
and
\[
\norm{\ol u\,H_t}_{L_2(\cM_0)}
=\norm{\frac{v}{\norm{\Dif_{x}h_t}}\,H_t}_{L_2(\cM_0)}
\leq \frac{\|H_t\|_\infty}{\inf_{x\in{\cal M}_0}\norm{\Dif_{x}h_t(x)}}\,\norm{v}_{L_2(\cM_0)}
\leq M\,\norm{v}_{L_2(\cM_0)},
\]
where the last inequality uses the bounded curvature assumption and
$\norm{\Dif_{x}h_t}\geq c>0$ on ${\cal M}_t$ (Assumption~\ref{assu:RegCond_Uch}), we obtain:
\begin{equation}\label{eq:bound_I't_alt}
\left|I^{'}\left(t\right)\right|
\leq M\left(\norm{\Dif_{x}w}_{L_2(\cM_0)}\norm{v}_{L_2(\cM_0)}
+\norm{w}_{L_2(\cM_0)}\norm{v}_{L_2(\cM_0)}\right).
\end{equation}

\end{proof}

~

\begin{proof}[\textbf{Proof of Theorem~\ref{thm:Int_UCh_rate}(1)}]
We show that under $s \ge m/2 + 1 = (d+1)/2$, the split-sample debiased sieve estimator $\hat{\theta}_{SS}$ achieves the optimal rate $r_n^* = n^{-s/(2s+1)}$.

Let $\delta_1 := \hat{h}_1 - h_0$, $\delta_2 := \hat{h}_2 - h_0$, and $\bar{h} = (\hat{h}_1 + \hat{h}_2)/2 = h_0 + (\delta_1 + \delta_2)/2$. Expanding $V(\bar{h})$ around $V(h_0)$, we have
\begin{align}
V(\bar{h}) - V(h_0) &= DV(h_0)[\bar{h} - h_0] + \frac{1}{2}D^2V(h_0)[\bar{h} - h_0, \bar{h} - h_0] + R_3, \label{eq:V_taylor}
\end{align}
where $R_3$ is the third-order remainder term.

Applying the shape calculus transport formula to differentiate $DV(h_0+tv_2)[v_1]=\int_{{\cal M}_{h_0+tv_2}}\frac{w\,v_1}{\norm{\Dif_{x}(h_0+tv_2)}}d{\cal H}^{d-1}$ with respect to $t$ at $t=0$, we obtain the second functional derivative $D^2V(h)[v_1,v_2]$ as
\begin{align}\label{eq:D2V_alt_alt}
D^2V(h)[v_1,v_2]&=-\int_{{\cal M}_h}
\frac{v_1(x)\,v_2(x)}{\norm{\Dif_{x}h(x)}^2}
\left[\frac{\Dif_{x}h}{\norm{\Dif_{x}h}}\cd\Dif_{x}\!\left(\frac{w}{\norm{\Dif_{x}h}}\right)
+\frac{w}{\norm{\Dif_{x}h}}\,H(x)\right]d{\cal H}^{d-1}(x) \notag\\
&\quad - \int_{{\cal M}_h} \frac{w(x)}{\norm{\Dif_{x}h(x)}^2} \left[\frac{\Dif_{x}h}{\norm{\Dif_{x}h}} \cd \Dif_{x} v_1 \cd v_2 + \frac{\Dif_{x}h}{\norm{\Dif_{x}h}} \cd \Dif_{x} v_2 \cd v_1\right] d{\cal H}^{d-1}(x).
\end{align}
The first integral involves only pointwise values of $v_1,v_2$ (with a bounded integrand under Assumption~\ref{assu:RegCond_Uch}), while the second integral involves normal derivatives ${\bf n}\cd\Dif_x v_1$ and ${\bf n}\cd\Dif_x v_2$. Since $D^2V(h)[v_1, v_2]$ is bilinear, we have
\[
\frac{1}{4}D^2V(h_0)[\hat{h}_1 - \hat{h}_2, \hat{h}_1 - \hat{h}_2] = \frac{1}{4}\left(D^2V(h_0)[\delta_1, \delta_1] - 2D^2V(h_0)[\delta_1, \delta_2] + D^2V(h_0)[\delta_2, \delta_2]\right).
\]
Taking difference of the above removes the diagonal term:
\[
D^2V(h_0)[\bar{h} - h_0, \bar{h} - h_0] - \frac{1}{4}D^2V(h_0)[\hat{h}_1 - \hat{h}_2, \hat{h}_1 - \hat{h}_2] = D^2V(h_0)[\delta_1, \delta_2].
\]

Under Assumption~\ref{assu:RegCond_Uch}, the coefficient functions in \eqref{eq:D2V_alt_alt} are bounded,
and therefore by Cauchy--Schwarz:
\begin{equation}\label{eq:D2V_alt_alt_bound_alt}
\big|D^2V(\bar h_0)[v_1,v_2]\big|
\leq C\,\norm{v_1}_{L^2({\cal M}_0)}\left(\norm{v_2}_{L^2({\cal M}_0)}+\norm{\Dif_x v_2}_{L^2({\cal M}_0)}\right)
+C\,\norm{v_2}_{L^2({\cal M}_0)}\norm{\Dif_x v_1}_{L^2({\cal M}_0)}.
\end{equation}

Given the above, we have
\begin{equation}
\hat{\theta}_{SS} - \theta_0 = \underbrace{DV(h_0)[\bar{h} - h_0]}_{T_{lin}} + \underbrace{\frac{1}{2}D^2V(h_0)[\delta_1, \delta_2]}_{T_{cross}} + R_{rem}. \label{eq:V_SS_decomp}
\end{equation}

Again, the linear term is a $m=(d-1)$-dimensional submanifold integral (over $\cM_0$) with weight $w/\|\nabla_x h_0\|$. By Theorem~\ref{thm:Rate_Sieve} with $m = (d-1)$ we obtain:
\[
T_{lin} = O_p\left(\sqrt{\frac{K_n^{1/d}}{n}} + K_n^{-s/d}\right).
\]
By setting $K^*_n \asymp n^{d/(2s+1)}$ (noting $d - m = 1$), we get $T_{lin} = O_p(n^{-s/(2s+1)}) = O_p(r_n^*)$.

The cross-term $T_{cross}=\frac{1}{2}D^2V(h_0)[\delta_1, \delta_2]$ involves integrals of the forms:
\[
(i)~\int_{\cM_0} \frac{\delta_1(x)\delta_2(x)}{\|\nabla_x h_0\|^2} \tilde{w}(x)\,d\cH^{d-1}(x), \quad (ii)~\int_{\cM_0} \frac{(\nabla_x\delta_1 \cdot \mathbf{n})\delta_2 + (\nabla_x\delta_2 \cdot \mathbf{n})\delta_1}{\|\nabla_x h_0\|^2} w\,d\cH^{d-1}(x).
\]
By the same proof of Lemma \ref{lem:Rate_Quad}(b) Term $T_3 = O_p\left(n^{-1}K_n^{(2d-m)/(2d)}\right)$, we obtain, term (i) is of the order $O_p\left(n^{-1}K_n^{(d+1)/(2d)}\right)$ with $m=d-1$
Notice that term (i) goes to zero strictly faster than term (ii), we now focus on bounding term (ii).
Applying Lemma \ref{lem:Uch_SS-cross}(1) below, we have, with $K^*_n \asymp n^{d/(2s+1)}$:
\[
\text{Term (ii)} =O_p(n^{-s/(2s+1)})=O_p(r^*_n)~~\text{when}~~ s\geq (d+1)/2.
\]
Consequently, when $s\geq (d+1)/2$, at $K^*_n \asymp n^{d/(2s+1)}$, the quadratic cross term satisfies
$$T_{cross} = O_p(r_n^*).$$

Finally, since $\norm{\hat{h}-h_0}_{\infty}=o_p(1)$, the third-order remainder term $R_3=o_p(T_{cross})$.

In summary, under $s \ge (d+1)/2$, at $K^*_n = n^{d/(2s+1)}$, we obtain: $\hat{\theta}_{SS} - \theta_0 = O_p(r_n^*).$
\end{proof}
~

\begin{lem}\label{lem:Uch_SS-cross}
Let $\hat{h}_1,\hat{h}_2$ be split-sample sieve LS estimators of $h_0$ computed on two independent subsamples
$I_1,I_2$ of sizes $n_1,n_2$ with $n_1\asymp n_2\asymp n$, using sieve dimension $K\equiv K_n$.
Let $\td h_r:=P_{K,r}h_0$ denote the (subsample-$r$) empirical LS projection of $h_0$ onto ${\H}_{K}$.
Assume Assumptions \ref{assu:RegCond_Uch}, \ref{assu:density_x}, \ref{assu:w}, \ref{assu:e2_moment}, and \ref{assu:Cond_Sieve} hold,
and assume in addition that $h_0\in\Lambda^s(\mathcal X)$ with $s>1$ so that
$\|\nabla_x(\td h_r-h_0)\|_\infty=O_p(K^{-(s-1)/d})$ for $r=1,2$ (e.g. for spline/wavelet tensor-product sieves).
Define
\[
T_{\mathrm{mix-cross}}
:=\int_{\cM_0}(\hat h_1-h_0)\,\nabla_x(\hat h_2-h_0)\,w'\,d\cH^{d-1},
\]
where $w'$ is a bounded (row-)vector-valued weight on $\cM_0$ and the product $\nabla_x(\hat h_2-h_0)\,w'$ is interpreted as the usual inner product.
Then
\[
T_{\mathrm{mix-cross}}=O_p \left(K_n^{1/d}\times\left[K_n^{-2s/d} +K_n^{-s/d}\sqrt{ \frac{K_n^{1/d}}{n}} +\frac{1}{n} K^{\frac{d+1}{2d}}\right]\right).
\]
Hence, with $K^*_n \asymp n^{d/(2s+1)}$, we obtain:
\begin{itemize}
\item[(1)] $T_{\mathrm{mix-cross}}=O_p(n^{-s/(2s+1)})=O_p(r^*_n)$ when $s\geq (d+1)/2$.
\item[(2)]  $T_{\mathrm{mix-cross}}=o_p(n^{-s/(2s+1)})=o_p(r^*_n)$ when $s> (d+1)/2$.
\end{itemize}
\end{lem}

\begin{proof}[\textbf{Proof of Lemma \ref{lem:Uch_SS-cross}}]
Write $K:=K_n$ for brevity.
Decompose each split-sample error as
\[
\hat h_r-h_0=(\hat h_r-\td h_r)+(\td h_r-h_0),\qquad r\in\{1,2\}.
\]
Expanding $T_{\mathrm{mix-cross}}$ yields
\begin{align}
T_{\mathrm{mix-cross}}
&=\int_{\cM_0}(\hat h_1-\td h_1)\,\nabla_x(\hat h_2-\td h_2)\,w'\,d\cH^{d-1} \notag\\
&\quad+\int_{\cM_0}(\td h_1-h_0)\,\nabla_x(\hat h_2-\td h_2)\,w'\,d\cH^{d-1} \notag\\
&\quad+\int_{\cM_0}(\hat h_1-\td h_1)\,\nabla_x(\td h_2-h_0)\,w'\,d\cH^{d-1} \notag\\
&\quad+\int_{\cM_0}(\td h_1-h_0)\,\nabla_x(\td h_2-h_0)\,w'\,d\cH^{d-1} \notag\\
&=:T_{\mathrm{stoch}}+T_{1}+T_{2}+T_{\mathrm{bias}}.
\label{eq:Uch_cross_split}
\end{align}

By Assumption~\ref{assu:Cond_Sieve}(iv)(v),
$\|\td h_1-h_0\|_\infty=O_p(K^{-s/d})$ and, by the stated $s>1$ smoothness/approximation,
$\|\nabla_x(\td h_2-h_0)\|_\infty=O_p(K^{-(s-1)/d})$.
Since $\|w'\|_\infty<\infty$ and $\cH^{d-1}(\cM_0)<\infty$ under Assumption~\ref{assu:RegCond_Uch}$($i$)$,
\[
|T_{\mathrm{bias}}|
\le \|w'\|_{\infty}\,\cH^{d-1}(\cM_0)\,\|\td h_1-h_0\|_{\infty}\,\|\nabla_x(\td h_2-h_0)\|_{\infty}
=O_p\!\left(K^{-\frac{2s-1}{d}}\right).
\]

\smallskip
Let $A_{1,n}:=\|\td h_1-h_0\|_\infty$ and, on $\{A_{1,n}>0\}$, define the (random but bounded) weight
\[
w'_{1,n}(x):=w'(x)\,\frac{\td h_1(x)-h_0(x)}{A_{1,n}},\qquad x\in\cM_0,
\]
and set $w'_{1,n}\equiv 0$ on $\{A_{1,n}=0\}$.
Then $\|w'_{1,n}\|_\infty\le \|w'\|_\infty$ and $w'_{1,n}$ is measurable w.r.t.\ the first split and hence independent of $\hat h_2$.
With this normalization,
\begin{equation}\label{eq:T1_normalized}
T_{1}
=A_{1,n}\int_{\cM_0}\nabla_x(\hat h_2-\td h_2)(x)\,w'_{1,n}(x)\,d\cH^{d-1}(x)
\equiv A_{1,n}\,L^{\nabla}_{w'_{1,n}}(\hat h_2-\td h_2),
\end{equation}
where $L^{\nabla}_{\omega'}(f):=\int_{\cM_0}\nabla_x f(x)\,\omega'(x)\,d\cH^{d-1}(x)$ is a linear functional on $\H_K$.

Exactly as in the proof of Lemma~\ref{lem:Rate_Quad}(b),
Assumption~\ref{assu:Cond_Sieve}(iv)(v) implies
\begin{equation}\label{eq:A1n_rate_cross}
A_{1,n}=\|\td h_1-h_0\|_\infty=O_p(K^{-s/d}).
\end{equation}

To bound $L^{\nabla}_{w'_{1,n}}(\hat h_2-\td h_2)$.
Condition on the first split, so $w'_{1,n}$ is fixed and bounded.
For a split-sample sieve LS estimator, the centered part admits the orthonormalized linear representation
\[
\hat h_2(x)-\td h_2(x)
=\ol b^K(x)^{\prime}\,\ol G_2^{-1}\,\frac1{n_2}\sum_{i\in I_2}\ol b^K(X_i)\e_i,
\qquad
\ol G_2:=\frac1{n_2}\sum_{i\in I_2}\ol b^K(X_i)\ol b^K(X_i)^{\prime},
\]
hence
\[
L^{\nabla}_{w'_{1,n}}(\hat h_2-\td h_2)
=
\left(L^{\nabla}_{w'_{1,n}}(\ol b^K)\right)^{\prime}\,
\ol G_2^{-1}\,\eta_2,
\qquad
\eta_2:=\frac1{n_2}\sum_{i\in I_2}\ol b^K(X_i)\e_i\in\R^K,
\]
where $L^{\nabla}_{w'_{1,n}}(\ol b^K):=\big(L^{\nabla}_{w'_{1,n}}(\ol b_1^K),\dots,L^{\nabla}_{w'_{1,n}}(\ol b_K^K)\big)^{\prime}$.
Since $\E[\e_i\mid X_i]=0$ and the second split is i.i.d., this quantity is centered conditional on the second-split design.

We now bound its conditional variance.
Let $\bar\s^2:=\sup_x \E[\e_i^2\mid X_i=x]<\infty$ (Assumption~\ref{assu:e2_moment}).
As in the proof of Lemma~\ref{lem:Rate_Quad}(b),
\[
\E[\eta_2\eta_2'\mid X_1^n]
=\frac1{n_2^2}\sum_{i\in I_2}\E[\e_i^2\mid X_i]\;\ol b^K(X_i)\ol b^K(X_i)'
\preceq \frac{\bar\s^2}{n_2}\,\ol G_2.
\]
Therefore, conditional on the first split,
\begin{align}
\E\!\left[\left(L^{\nabla}_{w'_{1,n}}(\hat h_2-\td h_2)\right)^2\Bigm| X_1^n\right]
&\le \frac{\bar\s^2}{n_2}\,
\left(L^{\nabla}_{w'_{1,n}}(\ol b^K)\right)^{\prime}\,
\ol G_2^{-1}\,
L^{\nabla}_{w'_{1,n}}(\ol b^K) \notag\\
&\le \frac{\bar\s^2}{n_2}\,\|\ol G_2^{-1}\|_{op}\,\big\|L^{\nabla}_{w'_{1,n}}(\ol b^K)\big\|^2 .
\label{eq:Var_Lgrad}
\end{align}
Under Assumption~\ref{assu:Cond_Sieve}(i)(iii), $\|\ol G_2^{-1}\|_{op}=O_p(1)$.

\smallskip
\noindent\emph{Claim:} uniformly over bounded weights $\omega'$ on $\cM_0$,
\begin{equation}\label{eq:Lgrad_growth_claim}
\big\|L^{\nabla}_{\omega'}(\ol b^K)\big\|^2
=
\sum_{k=1}^K \Big(\int_{\cM_0}\nabla_x \ol b_k^K(x)\,\omega'(x)\,d\cH^{d-1}(x)\Big)^2
\ \lesssim\ \|\omega'\|_\infty^2\,K^{3/d}.
\end{equation}

We now check the claim. Write $K=J^{d}$ as in Lemma~\ref{lem:NormRate_d-m} and write $\omega'(x)=(\omega_1(x),\dots,\omega_d(x))$.
Using $(\sum_{\ell=1}^d a_\ell)^2\le d\sum_{\ell=1}^d a_\ell^2$, it suffices to bound, for each fixed $\ell\in\{1,\dots,d\}$,
\[
S_\ell
:=\sum_{k=1}^K \Big(\int_{\cM_0}\partial_{x_\ell}\ol b_k^K(x)\,\omega_\ell(x)\,d\cH^{d-1}(x)\Big)^2.
\]
We show $S_\ell\lesssim \|\omega'\|_\infty^2\,K^{3/d}$ uniformly in $\ell$.

Apply the chart decomposition \eqref{eq:PieceLebInt} for $\cM_0$ (dimension $m=d-1$) to write the integral as a finite sum over charts.
By Cauchy--Schwarz over finitely many charts, it suffices to bound the contribution of a fixed chart uniformly.
Fix a chart $(\mathcal U_j,\varphi_j)$ and suppress the chart index.
Let $\psi(\cd)$ denote the implicit-function map for the remaining coordinate and define the bounded chart weight
\[
\tilde w(u):=\rho_j(\varphi_j(u))\,\omega_\ell(\varphi_j(u))\,Jac_{\varphi_j}(u),
\]
absorbing bounded factors into constants.

Each basis index $k\leftrightarrow (k_1,\dots,k_d)$ corresponds to the tensor-product basis function
$\ol b_k^K(x)=\prod_{r=1}^d b_{k_r}(x_r)$.
Differentiating in $x_\ell$ replaces the $\ell$th univariate factor by its derivative, so
\[
\partial_{x_\ell}\ol b_k^K(\varphi_j(u))
=\ol b^{K,(\ell)}_{k,(m)}(u)\cd\ol b_{k,-(m)}^{K}(\psi(u)),
\]
where $\ol b^{K,(\ell)}_{k,(m)}(\cd)$ is an $m$-variate tensor-product basis on the free coordinates (with one univariate factor differentiated if $\ell\in(m)$),
and $\ol b_{k,-(m)}^{K}(\cd)$ is the $(d-m)=1$-variate factor on the remaining coordinate (differentiated instead if $\ell\notin(m)$).
For spline/wavelet tensor-product bases at resolution $J$, the collection $\{J^{-1}\ol b^{K,(\ell)}_{k,(m)}\}_{k_{(m)}}$ restricted to $\mathcal U_j$
is again a frame with constants independent of $J$ (this is the same frame argument used in Lemma~\ref{lem:NormRate_d-m}, but with one derivative and hence the scaling $J$).
Therefore, for each fixed $k_{-(m)}$, the frame inequality yields
\[
\sum_{k_{(m)}}\Big(\int_{\mathcal U_j}\ol b^{K,(\ell)}_{k,(m)}(u)\,\ol b_{k,-(m)}^{K}(\psi(u))\,\tilde w(u)\,du\Big)^2
\ \lesssim\ J^2\int_{\mathcal U_j}\Big(\ol b_{k,-(m)}^{K}(\psi(u))\Big)^2\,\tilde w(u)^2\,du.
\]
Summing over $k_{-(m)}$ and using the tensor-product identity
$\sum_{k_{-(m)}}(\ol b_{k,-(m)}^{K}(\psi(u)))^2\lesssim J^{d-m}=J$ pointwise in $u$
(as in the proof of Lemma~\ref{lem:NormRate_d-m}) gives
\[
S_\ell
\ \lesssim\ J^2\cdot J\int_{\mathcal U_j}\tilde w(u)^2\,du
\ \lesssim\ \|\omega'\|_\infty^2\,J^3
=\|\omega'\|_\infty^2\,K^{3/d}.
\]
This proves \eqref{eq:Lgrad_growth_claim}.

Combining \eqref{eq:Var_Lgrad} and \eqref{eq:Lgrad_growth_claim} gives
\[
L^{\nabla}_{w'_{1,n}}(\hat h_2-\td h_2)
=O_p\!\left(\sqrt{\frac{K^{3/d}}{n}}\right)
=O_p\!\left(K^{1/d}\sqrt{\frac{K^{1/d}}{n}}\right).
\]
Together with \eqref{eq:T1_normalized}--\eqref{eq:A1n_rate_cross},
\begin{equation}\label{eq:T1_rate}
T_{1}
=O_p\!\left(K^{-s/d}\cdot K^{1/d}\sqrt{\frac{K^{1/d}}{n}}\right)
=O_p\!\left(K^{-(s-1)/d}\sqrt{\frac{K^{1/d}}{n}}\right).
\end{equation}

\smallskip
By symmetry, swapping the roles of the two splits yields the same rate.
For completeness, let $B_{2,n}:=\|\nabla_x(\td h_2-h_0)\|_\infty=O_p(K^{-(s-1)/d})$ and define
\[
w_{2,n}(x):=\frac{\nabla_x(\td h_2-h_0)(x)\,w'(x)}{B_{2,n}},\qquad x\in\cM_0,
\]
with $w_{2,n}\equiv 0$ on $\{B_{2,n}=0\}$.
Then $\|w_{2,n}\|_\infty\le \|w'\|_\infty$ and $w_{2,n}$ is measurable w.r.t.\ the second split, hence independent of $\hat h_1$.
We can write
\[
T_{2}=B_{2,n}\int_{\cM_0}(\hat h_1-\td h_1)(x)\,w_{2,n}(x)\,d\cH^{d-1}(x)
\equiv B_{2,n}\,L_{w_{2,n}}(\hat h_1-\td h_1),
\]
where $L_{w}(f):=\int_{\cM_0} f(x)w(x)\,d\cH^{d-1}(x)$ is the usual (non-derivative) linear submanifold functional.
By Lemma~\ref{lem:NormRate_d-m} with $m=d-1$ and the same centered-linear-functional argument as in Lemma~\ref{lem:Rate_Quad}(b),
\[
L_{w_{2,n}}(\hat h_1-\td h_1)=O_p\!\left(\sqrt{\frac{K^{1/d}}{n}}\right),
\]
and hence
\begin{equation}\label{eq:T2_rate}
T_{2}=O_p\!\left(K^{-(s-1)/d}\sqrt{\frac{K^{1/d}}{n}}\right).
\end{equation}

\smallskip
Using the orthonormalized linear representations on each split, define
\[
\eta_r:=\frac1{n_r}\sum_{i\in I_r}\ol b^K(X_i)\e_i,\qquad
\ol G_r:=\frac1{n_r}\sum_{i\in I_r}\ol b^K(X_i)\ol b^K(X_i)^{\prime},\qquad r\in\{1,2\}.
\]
Define the deterministic $K\times K$ matrix
\[
C_K
:=\int_{\cM_0}\ol b^K(x)\,\Big(\nabla_x \ol b^K(x)\,w'\Big)'\,d\cH^{d-1}(x),
\]
(where $\nabla_x \ol b^K(x)$ is the $K\times d$ Jacobian and $\nabla_x \ol b^K(x)\,w'$ is the resulting $K$-vector).
Then, up to sign (replace $w'$ by $|w'|$ if desired),
\begin{equation}\label{eq:Tstoch_matrix}
T_{\mathrm{stoch}}=\eta_1'\,\ol G_1^{-1}\,C_K\,\ol G_2^{-1}\,\eta_2.
\end{equation}
Conditioning on the full design $X_1^n$ and using split-sample independence, $\E[T_{\mathrm{stoch}}\mid X_1^n]=0$.
Moreover, as before,
\[
\E[\eta_r\eta_r'\mid X_1^n]\preceq \frac{\bar\s^2}{n_r}\,\ol G_r,\qquad r=1,2.
\]
Therefore, by the same trace argument as in Lemma~\ref{lem:Rate_Quad}(b) (Step~4),
\begin{align}
\E[T_{\mathrm{stoch}}^2\mid X_1^n]
&\le \frac{\bar\s^4}{n_1n_2}\,\tr\!\big(\ol G_1^{-1}C_K\ol G_2^{-1}C_K'\big)\notag\\
&\le \frac{\bar\s^4}{n_1n_2}\,\|\ol G_1^{-1}\|_{op}\,\|\ol G_2^{-1}\|_{op}\,\tr(C_KC_K').
\label{eq:Var_Tstoch}
\end{align}
Under Assumption~\ref{assu:Cond_Sieve}(i)(iii), $\|\ol G_r^{-1}\|_{op}=O_p(1)$.
It remains to control $\tr(C_KC_K')=\|C_K\|_F^2$.

\smallskip
Now, decompose $\cM_0$ into finitely many charts as in \eqref{eq:PieceLebInt} and write $C_K=\sum_{j=1}^{\bar j} C_{K,j}$ with
\[
C_{K,j}:=\int_{\mathcal U_j}\ol b^K(\varphi_j(u))\Big(\nabla_x \ol b^K(\varphi_j(u))\,w'\Big)'\,\tilde w_j(u)\,du,
\]
where $\tilde w_j(u)$ is bounded (it collects the partition of unity and the Jacobian).
Since $\|C_K\|_F^2\le \bar j\sum_{j=1}^{\bar j}\|C_{K,j}\|_F^2$, it suffices to show $\|C_{K,j}\|_F^2\lesssim K^{(d+3)/d}$ uniformly in $j$.

Fix $j$ and suppress the chart index.
Write $K=J^d$ and note $m=d-1$ (so $d-m=1$).
For each $k,k'$, the $(k,k')$ entry equals
\[
(C_{K,j})_{kk'}
=\int_{\mathcal U_j}\ol b_k^K(\varphi_j(u))\Big(\nabla_x \ol b_{k'}^K(\varphi_j(u))\,w'\Big)\,\tilde w_j(u)\,du.
\]
Using the tensor-product decomposition as in Lemma~\ref{lem:NormRate_d-m},
$\ol b_{k}^K(\varphi_j(u))=\ol b_{k,(m)}^K(u)\cd\ol b_{k,-(m)}^K(\psi(u))$.
Moreover, $\nabla_x \ol b_{k'}^K(\varphi_j(u))\,w'$ is a finite sum (over $\ell=1,\dots,d$) of tensor-product terms with one univariate factor differentiated.
As in the proof of \eqref{eq:Lgrad_growth_claim}, squaring and summing over the $m$-indices yields an additional factor $J^2$ from this one derivative.
Hence, for fixed $k$ and fixed $k'_{-(m)}$, the same frame-inequality argument as in the proof of Lemma~\ref{lem:Rate_Quad}(b) gives
\[
\sum_{k'_{(m)}} (C_{K,j})_{kk'}^2
\ \lesssim\ J^2\int_{\mathcal U_j}\Big(\ol b_k^K(\varphi_j(u))\Big)^2\Big(\ol b_{k',-(m)}^K(\psi(u))\Big)^2\,\tilde w_j(u)^2\,du.
\]
Summing over $k'_{-(m)}$ and using $\sum_{k'_{-(m)}}(\ol b_{k',-(m)}^K(\psi(u)))^2\lesssim J^{d-m}=J$ pointwise in $u$ yields
\[
\sum_{k'=1}^{K} (C_{K,j})_{kk'}^2
\ \lesssim\ J^3\int_{\mathcal U_j}\Big(\ol b_k^K(\varphi_j(u))\Big)^2\,\tilde w_j(u)^2\,du.
\]
Finally, summing over $k$ and using $\sum_{k=1}^{K}(\ol b_k^K(x))^2=\|\ol b^K(x)\|^2\lesssim K$ pointwise (Assumption~\ref{assu:Cond_Sieve}(ii) and $\lambda_{\min}(G)>0$) yields
\[
\|C_{K,j}\|_F^2
=\sum_{k=1}^{K}\sum_{k'=1}^{K}(C_{K,j})_{kk'}^2
\ \lesssim\ J^3\int_{\mathcal U_j}\Big(\sum_{k=1}^{K}(\ol b_k^K(\varphi_j(u)))^2\Big)\tilde w_j(u)^2\,du
\ \lesssim\ J^3\cd K
=J^{d+3}
=K^{(d+3)/d}.
\]
This proves the claim.

Substituting into \eqref{eq:Var_Tstoch} and using $n_1\asymp n_2\asymp n$ gives
\[
\E[T_{\mathrm{stoch}}^2]=O\!\left(\frac{K^{(d+3)/d}}{n^2}\right),
\qquad\text{hence}\qquad
T_{\mathrm{stoch}}=O_p\!\left(\frac1n\sqrt{K^{\frac{d+3}{d}}}\right)
=O_p\!\left(\frac1n K^{\frac{d+3}{2d}}\right).
\]
Finally, combining \eqref{eq:Uch_cross_split}, the bounds for $T_{\mathrm{bias}}$, \eqref{eq:T1_rate}, \eqref{eq:T2_rate}, and the bound for $T_{\mathrm{stoch}}$ yields
\[
T_{\mathrm{mix-cross}}
=O_p\!\left(
K^{-\frac{2s-1}{d}}
+
K^{-\frac{s-1}{d}}\sqrt{\frac{K^{1/d}}{n}}
+
\frac1n\sqrt{K^{\frac{d+3}{d}}}
\right),
\]
as claimed.
\end{proof}

\begin{proof}[\textbf{Proof of Theorem~\ref{thm:Int_UCh_rate}(2)}]
The proof is very similar to the proofs of Theorem~\ref{thm:Int_NLh_rate}(2) and Theorem~\ref{thm:Int_UCh_rate}(1), hence we will be brief here.
Let $K\equiv K_n\asymp K_n^*\asymp n^{\frac{d}{2s+1}}$ and $r_n^*:=n^{-\frac{s}{2s+1}}$.
Recall $\td h:=P_{K,n}h_0$ and the series LS representation \eqref{eq:seriesLS}):
\begin{equation}\label{eq:LOO_V_series_rep}
\hat h(x)=b^K(x)'\hat G^{-1}\Big(\frac1n\sum_{i=1}^n b^K(X_i)Y_i\Big)=\frac1n\sum_{i=1}^n \hat s_i(x)\,Y_i,\qquad \hat s_i(x):=b^K(x)'\hat G^{-1}b^K(X_i),
\end{equation}
so that
\begin{equation}\label{eq:LOO_V_stoch_part}
\hat h(x)-\td h(x)=\frac1n\sum_{i=1}^n \hat s_i(x)\,\e_i.
\end{equation}

\smallskip
Exactly as in the mean-value expansion used in \eqref{eq:V_plugin}, there exists a (random) function $\bar h_0$ that lies on the line segment between $\hat h$ and $h_0$ such that
\begin{equation}\label{eq:LOO_V_taylor}
V(\hat h)-V(h_0)=DV(h_0)[\hat h-h_0]+\frac12 D^2V(\bar h_0)[\hat h-h_0,\hat h-h_0].
\end{equation}
Consequently, with $\hat\t_{\mathrm{LOO}}^{V}$ defined in \eqref{eq:V_LOO_estimator},
\begin{equation}\label{eq:LOO_V_main_decomp}
\hat\t_{\mathrm{LOO}}^{V}-V(h_0)
=T_{\mathrm{lin}}+T_{\mathrm{quad}}-C_n,
\end{equation}
where
\begin{align*}
T_{\mathrm{lin}}&:=DV(h_0)[\hat h-h_0],\\
T_{\mathrm{quad}}&:=\frac12 D^2V(\bar h_0)[\hat h-h_0,\hat h-h_0],\\
C_n&:=\frac{1}{2n^2}\sum_{i=1}^n D^2V(\hat h)[\hat s_i,\hat s_i]\cdot (\hat\e_i^{(-i)})^2\qquad\text{(the LOO correction term in \eqref{eq:V_LOO_estimator}).}
\end{align*}

\smallskip
By \eqref{eq:PathD_Contour}, $T_{\mathrm{lin}}$ is an $m=(d-1)$-dimensional submanifold integral on $\cM_0$ with bounded weight $w/\|\nabla_x h_0\|$.
Applying Theorem~\ref{thm:Rate_Sieve} with $m=d-1$ gives
\begin{equation}\label{eq:LOO_V_Tlin_rate}
T_{\mathrm{lin}}=O_p\Big(\sqrt{\tfrac{K^{1/d}}{n}}+K^{-s/d}\Big)=O_p(r_n^*),
\end{equation}
at $K\asymp K_n^*\asymp n^{d/(2s+1)}$.

\smallskip
Now, let $u:=\hat h-\td h$ and $b:=\td h-h_0$, so that $\hat h-h_0=u+b$.
By bilinearity of $D^2V(\bar h_0)[\cdot,\cdot]$,
\begin{equation}\label{eq:LOO_V_u_plus_b}
T_{\mathrm{quad}}=\frac12 D^2V(\bar h_0)[u,u]\; +\; D^2V(\bar h_0)[u,b]\; +\; \frac12 D^2V(\bar h_0)[b,b].
\end{equation}
From the explicit representation of $D^2V$ (cf. \eqref{eq:D2V_alt_alt}) and Assumption~\ref{assu:RegCond_Uch} (bounded $w$, $\nabla w$, bounded curvature, and $\|\nabla h_0\|$ bounded away from zero on $\cM_0$), there exists a finite constant $C>0$ (uniform in $n$) such that, on the event $\|\hat h-h_0\|_{\infty}<\ul\e$,
\begin{align}\label{eq:LOO_V_D2_bound}
&\big|D^2V(\bar h_0)[v_1,v_2]\big|\\\nonumber
\le \ & C\,\|v_1\|_{L^2(\cM_0)}\Big(\|v_2\|_{L^2(\cM_0)}+\|\nabla_x v_2\|_{L^2(\cM_0)}\Big)
+C\,\|v_2\|_{L^2(\cM_0)}\,\|\nabla_x v_1\|_{L^2(\cM_0)}.
\end{align}
Now use \eqref{eq:LOO_V_stoch_part} and Lemma~\ref{lem:MSE_on_M} (applied to $\hat h$ and to $\td h$) to obtain
\begin{equation}\label{eq:LOO_V_u_rates}
\|u\|_{L^2(\cM_0)}=O_p\Big(\sqrt{\tfrac{K}{n}}\Big),\qquad
\|\nabla_x u\|_{L^2(\cM_0)}=O_p\Big(K^{1/d}\sqrt{\tfrac{K}{n}}\Big).
\end{equation}
Moreover, by the standard approximation property of spline/wavelet tensor-product sieves for $h_0\in\L^s(\cX)$ with $s>1$ (cf. the discussion in Lemma~\ref{lem:Uch_SS-cross}),
\begin{equation}\label{eq:LOO_V_b_rates}
\|b\|_{\infty}=O_p(K^{-s/d}),\qquad
\|\nabla_x b\|_{\infty}=O_p(K^{-(s-1)/d}),
\end{equation}
and hence $\|b\|_{L^2(\cM_0)}=O_p(K^{-s/d})$ and $\|\nabla_x b\|_{L^2(\cM_0)}=O_p(K^{-(s-1)/d})$.
Plugging \eqref{eq:LOO_V_u_rates}--\eqref{eq:LOO_V_b_rates} into \eqref{eq:LOO_V_D2_bound} yields
\begin{equation}\label{eq:LOO_V_cross_bias_orders}
D^2V(\bar h_0)[u,b]=O_p\Big(\sqrt{\tfrac{K}{n}}\,K^{-(s-1)/d} + K^{-s/d}\,K^{1/d}\sqrt{\tfrac{K}{n}}\Big)=O_p(r_n^*)\qquad\text{if }s\ge (d+1)/2,
\end{equation}
and
\begin{equation}\label{eq:LOO_V_bb_order}
D^2V(\bar h_0)[b,b]=O_p\big(K^{-(2s-1)/d}\big)=O_p(r_n^*),
\end{equation}
at $K\asymp K_n^*\asymp n^{d/(2s+1)}$.
Therefore,
\begin{equation}\label{eq:LOO_V_Tquad_reduced}
T_{\mathrm{quad}}=\frac12 D^2V(\bar h_0)[u,u]+O_p(r_n^*).
\end{equation}

\smallskip
By \eqref{eq:LOO_V_stoch_part} and bilinearity of $D^2V(\bar h_0)[\cdot,\cdot]$,
\begin{align}\label{eq:LOO_V_u2_D2_expand}
D^2V(\bar h_0)[u,u]
&=\frac{1}{n^2}\sum_{i=1}^n\sum_{j=1}^n \e_i\e_j\,D^2V(\bar h_0)[\hat s_i,\hat s_j]\\\nonumber
&=\underbrace{\frac{1}{n^2}\sum_{i=1}^n \e_i^2\,D^2V(\bar h_0)[\hat s_i,\hat s_i]}_{=:\,T_{\mathrm{diag}}}
+\underbrace{\frac{1}{n^2}\sum_{i\ne j} \e_i\e_j\,D^2V(\bar h_0)[\hat s_i,\hat s_j]}_{=:\,T_{\mathrm{off}}}.
\end{align}
Without debiasing, $T_{\mathrm{diag}}$ contains the ``diagonal'' variance contribution of order $K^{1+1/d}/n$ (cf. \eqref{D2V}), which is too large when $(d+1)/2\le s<d$.
The correction term $C_n$ is designed to estimate and subtract $T_{\mathrm{diag}}$.

Similar to the proof of Theorem~\ref{thm:Int_NLh_rate}(2), we may rewrite
\begin{equation}\label{eq:LOO_V_Cn_alt}
C_n
=\frac12\sum_{i=1}^n D^2V(\hat h)\big[\hat h-\hat h^{(-i)},\hat h-\hat h^{(-i)}\big]=T_{\mathrm{diag}}+R_{\mathrm{diag}},
\end{equation}
where $R_{\mathrm{diag}}$ collects:
(i) cross and squared terms involving the leave-one-out first-stage error $h_0(X_i)-\hat h^{(-i)}(X_i)$, and
(ii) the (asymptotically negligible) difference between $D^2V(\hat h)$ and $D^2V(\bar h_0)$.
A standard Cauchy--Schwarz bound combined with the same matrix-trace argument used to control the main stochastic term in Lemma~\ref{lem:Uch_SS-cross} (replace its bounded weight $w'$ by the bounded weights that appear in $D^2V$; cf. \eqref{eq:D2V_alt_alt}) yields
\begin{equation}\label{eq:LOO_V_Rdiag_rate}
R_{\mathrm{diag}}=O_p\Big(\frac1n\sqrt{K^{\frac{d+3}{d}}}\Big).
\end{equation}
Consequently, combining \eqref{eq:LOO_V_u2_D2_expand} and \eqref{eq:LOO_V_Cn_alt},
\begin{equation}\label{eq:LOO_V_Tquad_minus_Cn}
\frac12 D^2V(\bar h_0)[u,u]-C_n
=\frac12\,T_{\mathrm{off}}+R_{\mathrm{diag}}.
\end{equation}

We now consider the off-diagonal term $T_{\mathrm{off}}$. 
Conditional on the designs, $T_{\mathrm{off}}$ is a (conditionally) degenerate U-statistic of order two.
As in Lemma~\ref{lem:Uch_SS-cross} (main stochastic term), its conditional variance is bounded by
\[
\Var(T_{\mathrm{off}}\mid X_1^n)
\;\lesssim\;\frac{1}{n^2}\,\tr(C_KC_K'),
\]
where $C_K$ is a $K\times K$ matrix of the form
$\int_{\cM_0}\ol b^K(x)\big(\nabla_x \ol b^K(x)\,\omega(x)\big)'d\cH^{d-1}(x)$ for a bounded vector-valued weight $\omega$ (coming from \eqref{eq:D2V_alt_alt}).
By the claim proved in Lemma~\ref{lem:Uch_SS-cross}, $\tr(C_KC_K')\lesssim K^{(d+3)/d}$.
Therefore,
\begin{equation}\label{eq:LOO_V_Toff_rate}
T_{\mathrm{off}}=O_p\Big(\frac1n\sqrt{K^{\frac{d+3}{d}}}\Big).
\end{equation}

Finally, combining \eqref{eq:LOO_V_main_decomp}, \eqref{eq:LOO_V_Tlin_rate}, \eqref{eq:LOO_V_Tquad_reduced}, \eqref{eq:LOO_V_Tquad_minus_Cn}, \eqref{eq:LOO_V_Rdiag_rate} and \eqref{eq:LOO_V_Toff_rate} yields
\[
\hat\t_{\mathrm{LOO}}^{V}-V(h_0)
= T_{\mathrm{lin}}+O_p(r_n^*)+O_p\Big(\frac1n\sqrt{K^{\frac{d+3}{d}}}\Big).
\]
At $K\asymp K_n^*\asymp n^{d/(2s+1)}$ and $s\ge (d+1)/2$,
\[
\frac1n\sqrt{K^{\frac{d+3}{d}}}
=\frac1n K^{\frac{d+3}{2d}}
=n^{-1+\frac{d+3}{2(2s+1)}}
=O\big(n^{-\frac{s}{2s+1}}\big)=O(r_n^*),
\]
so $\hat\t_{\mathrm{LOO}}^{V}-V(h_0)=O_p(r_n^*)$.
\end{proof}

\begin{proof}[\textbf{Proof of Theorem~\ref{thm:Int_UCh_rate}(3)}] We write
\begin{equation}
\hat{\theta} - \theta_0 = V(\hat{h}) - V(h_0) = DV(h_0)[\hat{h} - h_0] + \frac{1}{2}D^2V(\bar{h}_0)[\hat{h} - h_0, \hat{h} - h_0], \label{eq:V_plugin}
\end{equation}
where $\bar{h}_0$ lies between $\hat{h}$ and $h_0$.

As before, the linear term $DV(h_0)[\hat{h} - h_0] = O_p(r_n^*)$ at optimal $K_n$. The second-order term $D^2V$ involves:
\begin{itemize}
    \item Terms of the form $\int_{\cM_0} (\hat{h} - h_0)^2 \tilde{w}\,d\cH^{d-1}$;
    \item Terms involving $\int_{\cM_0}(\hat{h} - h_0) \cdot (\nabla(\hat{h} - h_0) \cdot \mathbf{n})\tilde{w}\,d\cH^{d-1}$.
\end{itemize}
By Lemma~\ref{lem:PathD_Level} we obtain:
\[
|D^2V(\bar{h}_0)[\hat{h} - h_0, \hat{h} - h_0]| \le C\|\hat{h} - h_0\|_{L^2(\cM_0)} \left(\|\nabla_x(\hat{h} - h_0)\|_{L^2(\cM_0)} + \|\hat{h} - h_0\|_{L^2(\cM_0)}\right).
\]
By Lemma \ref{lem:MSE_on_M}(1) we have $\|\hat{h} - h_0\|_{L^2(\cM_0)}^2 = O_p(K_n^{-2s/d} + K_n/n)$. By Lemma \ref{lem:MSE_on_M}(2), we also have:
\[
\|\nabla_x(\hat{h} - h_0)\|_{L^2(\cM_0)}^2 = O_p\left(K_n^{2/d}\times \left[K_n^{-2s/d} + K_n/n\right]\right).
\]
Hence
\begin{equation}\label{D2V}
|D^2V(\bar{h}_0)[\hat{h} - h_0, \hat{h} - h_0]| = O_p\left(K_n^{1/d}\times \left[K_n^{-2s/d} + K_n/n\right]\right).
\end{equation}
We need $D^2V = O_p(r_n^*)$ at $K_n=K^*_n \asymp n^{d/(2s+1)}$, where the dominant contribution is again from the variance part. The condition we require is then
\[
K_n^{1/d} \times K_n/n  = O(r_n^*) = O(n^{-s/(2s+1)})~~~\text{at}~~K_n=K^*_n \asymp n^{d/(2s+1)},
\]
which is equivalent to
\[
n^{(d+1)/(2s+1) - 1} = O(n^{-s/(2s+1)})  \iff d \leq s.
\]
Hence, under $s \ge d = m + 1$, the plug-in sieve estimator $\hat{\t}=V(\hat{h})$ for $\t_0=V(h_0)$ achieves the minimax optimal rate of $\hat{\theta} - \theta_0 = O_p(n^{-s/(2s+1)})$, which is the same optimal rate of $r^*_n$ for the linear $m=(d-1)$-submanifold integrals.
\end{proof}

\section{\label{sec:MainProofs_NL}Proofs of Theoretical Results in Section \ref{sec:Nonlinear}}

We first recall the following two high level conditions from \citet*{chen2014sieve} that are applicable to any linear and nonlinear irregular functionals for sieve M estimation problems.

\begin{assumption}
\label{assu:LinRem_Small} (i) $\hat{h},~P_{K,n}h_0 \in \H_{K_{n}}$ wpa1, and
\begin{equation}
\sup_{h\in\H_{K_{n}}}\frac{\sqrt{n}\left|\Phi\left(h\right)-\Phi\left(h_{0}\right)-D\Phi\left(h_{0}\right)\left[h-h_{0}\right]\right|}{\norm{v^*_{K_n}}_2}=o_{p}\left(1\right).\label{eq:LinRemSmall}
\end{equation}
(ii) $\frac{\sqrt{n}\left|D\Phi( h_{0})\left[P_{K,n}h_0 -h_{0}\right]\right|}{\norm{v^*_{K_n}}_2}=o_{p}\left(1\right)$.
\end{assumption}
\begin{assumption}[CLT of Sieve Influence Function]
\label{assu:sieveIF}$\frac{1}{\sqrt{n}}\sum_{i=1}^n \frac{v_{K_n}^*(X_i)}{\norm{v_{K_{n}}^{*}}_{sd}}\e_i \dto\cN\left(0,1\right)$.
\end{assumption}

The following Remark specializes Theorem 3.1 of \citet*{chen2014sieve} (or Theorem 3.1 of \citet*{chen2015optimal})
to our irregular functionals of a nonparametric regression model with i.i.d. data.

\begin{rem} Let $\hat{\t}=\Phi(\hat{h})$ be the plug-in sieve estimator of $\t_0=\Phi(h_0)$. Let Assumptions \ref{assu:e2_below}, \ref{assu:e2_moment}, \ref{assu:LinRem_Small} and \ref{assu:sieveIF}
hold. Then:
\[
\frac{\sqrt{n}\left(\hat{\t}-\t_{0}\right)}{\norm{v_{K_{n}}^{*}}_{sd}} = \frac{1}{\sqrt{n}}\sum_{i=1}^n \frac{v_{K_n}^*(X_i)}{\norm{v_{K_{n}}^{*}}_{sd}}\e_i +o_p(1)\dto\cN\left(0,1\right).
\]
\end{rem}

When $\Phi(h)=L(h)$ is linear, Assumption \ref{assu:LinRem_Small}(i) is trivially satisfied and its sieve variance grows at the order $\norm{v_{K_{n}}^{*}}_{sd}^2\asymp {K_{n}^{\frac{d-m}{d}}}$ by Lemma \ref{lem:NormRate_d-m}. For a general nonlinear irregular functional such as $\Phi=\G,~V$, we have the following result controlling the growth rate of the sieve Riesz representer and the sieve variance as well

\begin{proof}[\textbf{Proof of Proposition \ref{prop:Plug-in_NL}}]
The asymptotic normality result follows from Theorem 3.1 of \cite{chen2015optimal}[CC15 for short] with its stated assumptions and additional conditions verified (for i.i.d. setting). Specifically, Assumptions 1(i) and 2(i) in CC15 are satisfied under the i.i.d. setting with ${\cal X} = [0,1]^d$, Assumption 2(ii)(iv)(v) in CC15 is satisfied by our Assumptions \ref{assu:e2_moment}, \ref{assu:e2_below}, and \ref{assu:Lindeberg}. Their Assumption 4(iii) is implied by our Assumption \ref{assu:Cond_Sieve}(i). Assumption 5 in CC15 is satisfied by our Assumption \ref{assu:Cond_Sieve}(i)(ii)(iii). Assumptions 9(i)(ii) in CC15 are trivially satisfied when $\Phi=L$, and are satisfied by our Assumptions \ref{assu:DGamma_linear} and \ref{assu:LinRem_Small}(i) when $\Phi=\G,~V$. Their Assumption 9(iii) is our Assumption \ref{assu:LinRem_Small}(ii), which is in turn satisfied by our Assumption \ref{assu:Cond_Sieve}(iv)(v) with the ``undersmoothing'' choice of $K_n$ s.t. $$\abs{D\G(h_0)[\tilde{h}-h_0]}\leq \norm{\tilde{h}-h_0}_\infty = O_p(K_n^{-s/d}) = o_p\left(\sqrt{\frac{K_n^\frac{d-m}{d}}{n}}\right),$$
where the first inequality follows from Assumption \ref{assu:DGamma_linear} and the uniform boundedness of $\ol{w}$. Finally, Lemma \ref{lem:NormRate_d-m} and Remark \ref{rem:NormRate_nonlinear} imply that growth rate of $\norm{v^*_{K_n}}_{sd} \asymp K_n^\frac{d-m}{d}$.
It remains to check that our Assumption \ref{assu:LinRem_Small}(i) is satisfied under Case 2 for $\Phi=\G$ and under Case 3 for $\Phi=V$ below.

\medskip

\noindent\textbf{(1) $\hat{\t}_{SS}$ in Case 2 and Case 3:}

Under \textbf{Case 2} for $\Phi=\G$, from the proof of Theorem~\ref{thm:Int_NLh_rate}(1) and the bound on the second order remainder term and Lemma \ref{lem:Rate_Quad}(b), with $s>m/2$, $K_{n}\asymp K^*_n [\log(n)]^{\kappa}$ for some $\kappa \in (0,1]$, $K^*_n=n^{d/(2s+d-m)}$ and $r^*_n=n^{-s/(2s+d-m)}$, we have:
\begin{align*}
\hat{\t}_{SS}-\G(h_0)=\ & \frac{1}{2}\int_\cM \phi_1(h_0,x)(\hat{h}_1 - h_0)\, w\, d\cH^m \\ &+\frac{1}{2}\int_\cM \phi_1(h_0,x)(\hat{h}_2 - h_0)\, w\, d\cH^m  +o_p\left( \sqrt{\frac{K_n^{(d-m)/d}}{n}}\right),    
\end{align*}
hence Assumption \ref{assu:LinRem_Small}(i) is satisfied.

Under \textbf{Case 3} for $\Phi=V$ with $m=d-1$, from the proof of Theorem~\ref{thm:Int_UCh_rate}(1) and the bound on the second order remainder term \eqref{D2V} and Lemma \ref{lem:Uch_SS-cross}(2), with $s>(m/2)+1=(d+1)/2$,  $K_{n}\asymp K^*_n [\log(n)]^{\kappa}$ for some $\kappa \in (0,1]$, $K^*_n=n^{d/(2s+1)}$ and $r^*_n=n^{-s/(2s+1)}$, we have:
\[
\hat{\t}_{SS}-V(h_0)=\frac{1}{2}DV(h_0)[\hat{h}_1 - h_0]+\frac{1}{2}DV(h_0)[\hat{h}_2 - h_0]+o_p\left( \sqrt{\frac{K_n^{1/d}}{n}}\right),
\]
hence Assumption \ref{assu:LinRem_Small}(i) is satisfied.

~

\noindent\textbf{(2) $\hat{\t}_{LOO}$ in Case 2 and Case 3:}

Under \textbf{Case 2} for $\Phi=\G$, from the proof of Theorem~\ref{thm:Int_NLh_rate}(2)
we have:
\[
\hat{\t}_{LOO}-\G(h_0)=\int_\cM \phi_1(h_0,x)[\hat{h} - h_0]\, w\, d\cH^m +o_p\left( \sqrt{\frac{K_n^{(d-m)/d}}{n}}\right),
\]
hence Assumption \ref{assu:LinRem_Small}(i) is satisfied.

Under \textbf{Case 3} for $\Phi=V$ with $m=d-1$, from the proof of Theorem~\ref{thm:Int_UCh_rate}(2), with $s>(m/2)+1=(d+1)/2$, $K_{n}\asymp K^*_n [\log(n)]^{\kappa}$ for some $\kappa \in (0,1]$, $K^*_n=n^{d/(2s+1)}$ and $r^*_n=n^{-s/(2s+1)}$, we have:
\[
\hat{\t}_{LOO}-V(h_0)=DV(h_0)[\hat{h} - h_0]+o_p\left( \sqrt{\frac{K_n^{1/d}}{n}}\right),
\]
hence Assumption \ref{assu:LinRem_Small}(i) is satisfied.

~

\noindent\textbf{(3) Plug-in $\hat{\t}=\Phi(\hat{h})$ in Case 2 and Case 3:}

Under \textbf{Case 2} for $\Phi=\G$, from the proof of Theorem~\ref{thm:Int_NLh_rate}(3) and the bound on the second order remainder term, we see that Assumption \ref{assu:LinRem_Small}(i) is satisfied provided that
\[
|D^2\G(\bar{h}_0)[\hat{h} - h_0, \hat{h} - h_0]| = O_p\left(K_n^{-2s/d} + K_n/n\right)
=o_p\left( \sqrt{\frac{K_n^{(d-m)/d}}{n}}\right),
\]
which is satisfied provided that $K_n^{(d+m)/d}=o(n)$ and $K^*_n=n^{d/(2s+d-m)}=o(K_n)$, which is in turn satisfied with $s>m$.

Under \textbf{Case 3} for $\Phi=V$ with $m=d-1$, from the proof of Theorem~\ref{thm:Int_UCh_rate}(3) and the bound on the second order remainder term \eqref{D2V}, we see that Assumption \ref{assu:LinRem_Small}(i) is satisfied provided that
\[
|D^2V(\bar{h}_0)[\hat{h} - h_0, \hat{h} - h_0]| = O_p\left(K_n^{1/d}\times \left[K_n^{-2s/d} + K_n/n\right]\right)
=o_p\left( \sqrt{\frac{K_n^{1/d}}{n}}\right),
\]
which is satisfied provided that $K_n^{(2d+1)/d}=o(n)$ and $K^*_n=n^{d/(2s+1)}=o(K_n)$, which is in turn satisfied with $s>m+1=d$.

\end{proof}

The following Assumption \ref{assu:VarConsis} corresponds to Assumption 3.1(iii) in \cite*{chen2014sieve} or Assumption 10 in \cite{chen2015optimal}.

\begin{assumption}\label{assu:VarConsis}
    $\norm{v_{K_{n}}^{*}}_2^{-1}\norm{D\Phi(h)[\ol {\psi}^{K_n}] -  D\Phi(h_0)[\ol {\psi}^{K_n}]} = o(1)$ uniformly in $h\in \H_{K_{n}}$.
\end{assumption}

\begin{prop}[Consistent Variance]\label{prop:VarConsis-CC}
Let Assumptions \ref{assu:density_x}, \ref{assu:e2_below}, \ref{assu:e2_moment}, \ref{assu:Cond_Sieve}(i)(ii), \ref{assu:VarConsis} and \ref{assu:e2d_moment} hold with $K_{n}^{(2+\d)/\d}\log K_{n}/n=o\left(1\right)$. Then:
\[
(1)~~~~\abs{\frac{\widehat{\sigma}_{*,K_{n}}}{\norm{v_{K_{n}}^{*}}_{sd} } - 1} =o_p(1);
\]
(2) Further, let Assumptions \ref{assu:Cond_Sieve}(iv)(v) and \ref{assu:Lindeberg} hold with  $K_{n}\asymp K^*_n [\log(n)]^{\kappa}$ for some $\kappa \in (0,1)$ where $K_{n}^{*}\asymp n^{\frac{d}{2s+d-m}}$. Then:
\[
\frac{\sqrt{n}\left(\Phi(\hat{h})-\Phi(h_0)\right)}{\widehat{\sigma}_{*,K_{n}}} = \frac{1}{\sqrt{n}}\sum_{i=1}^n \frac{v_{K_n}^*(X_i)}{\norm{v_{K_{n}}^{*}}_{sd}}\e_i +o_p(1)\dto\cN\left(0,1\right),
\]
under the additional mild conditions stated for Cases 1, 2 and 3 of Proposition \ref{prop:Plug-in_NL}
\end{prop}

\begin{proof}[\textbf{Proof of Proposition \ref{prop:VarConsis-CC}}] Result (1) follows from Lemmas 3.1 and 3.2 of \cite{chen2015optimal} for i.i.d. setting. Result (2) follows from Proposition 1 and Result (1).
%
\end{proof}

\begin{proof}[\textbf{Proof of Proposition \ref{prop:VarConsis}}]
Proposition \ref{prop:VarConsis} follows from Proposition \ref{prop:VarConsis-CC} by noting that Assumption \ref{assu:VarConsis} is automatically satisfied when $\Phi=L$ (linear). It is also satisfied by our conditions imposed to control for second order remainder terms in the convergence rate Theorems \ref{thm:Int_NLh_rate}(b) and \ref{thm:Int_UCh_rate}(b) for nonlinear functionals $\Phi=\G,~V$ with $K_{n}\asymp K^*_n [\log(n)]^{\kappa}$ for some $\kappa \in (0,1]$ where $K_{n}^{*}\asymp n^{\frac{d}{2s+d-m}}$. In particular, under Assumption \ref{assu:e2d_moment} (with $\d>2d/(2s-m)$, the choice of $K_{n}\asymp K^*_n [\log(n)]^{\kappa}$ satisfies $K_{n}^{(2+\d)/\d}\log K_{n}/n=o\left(1\right)$.
\end{proof}

\section{Multiplier Bootstrap Choice of Sieve Dimension}
\label{subsec:MB_K_LOO}

This section proposes a \emph{multiplier bootstrap}  and a \emph{data-driven} choice of sieve
dimension $K$ for the various sieve estimators $\tilde{\Phi}(K)$ for $\t_0^{\Phi}=\Phi(h_0)$, where $\Phi \in\{L,\G,V\}$.
The selection rule is adapted from the bootstrap-Lepski procedure in
\citet*[Procedure~1]{chen2025adaptive}, but specialized to our \emph{scalar}
functionals.
The purpose of the procedure is to: (i) select a \emph{rate-adaptive} ``estimation'' dimension $\hat K_{\Phi}$,
and (ii) convert it into an \emph{undersmoothed} ``inference'' dimension $\tilde K_{\Phi}$, so that the
sieve bias is negligible relative to the standard deviation and the sieve $t$ statistic is asymptotically
centered.

\subsubsection*{Estimated influence function and sieve $t$ statistic}

Fix $\Phi \in\{L,\G,V\}$ and a sieve dimension $K$.
Let $\hat v_{K,\Phi}^*(x)$ denote the \emph{sieve Riesz representer estimator} of
$v_{K,\Phi}^*(x)$ (cf.\ Proposition \ref{prop:VarConsis}), i.e.
\begin{equation}
  \hat v_{K,\Phi}^*(x)
  := \ol b^{K}(x)^\prime\Big(\frac1n\sum_{i=1}^n \ol b^{K}(X_i)\ol b^{K}(X_i)^\prime\Big)^{-1}
  D\Phi (\hat{h}_K)\!\left[\ol b^{K}\right]=b^{K}(x)^\prime \hat G^{-1}\,\widehat{D\Phi}\left[ b^{K}\right],
  \label{eq:MB_vhat}
\end{equation}
where $\ol b^{K}:=G^{-1/2}b^{K}$, $\hat G := \frac{1}{n}\sum_{i=1}^n b^{K}(X_i)\,b^{K}(X_i)^\prime$ and $\widehat{D\Phi}\left[ b^{K}\right]:= D\Phi (\hat{h}_K)\!\left[ b^{K}\right]$.
Let $\hat\e_i:=Y_i-\hat h_K (X_i)$ and define the estimated sieve influence function as
\begin{equation}
  \hat\psi_{i,\Phi}(K):=\hat v_{K,\Phi}^*(X_i)\,\hat\e_i,
  \qquad
  \hat\sigma_{\Phi}^2(K):=\frac1n\sum_{i=1}^n \left(\hat\psi_{i,\Phi}(K)\right)^2.
  \label{eq:MB_psihat}
\end{equation}
The (studentized) sieve $t$ statistic based on the sieve estimator $\tilde{\Phi}(K)$ for $\t_0^{\Phi}=\Phi(h_0)$ is
\begin{equation}
  T_{n,\Phi}(K)
  :=\frac{\sqrt n\big(\tilde{\Phi}(K)-\t_0^{\Phi}\big)}{\hat\sigma_{\Phi}(K)}.
  \label{eq:MB_tstat}
\end{equation}

\subsubsection*{Multiplier bootstrap for a fixed sieve dimension}

Let $(\xi_i)_{i=1}^n$ be i.i.d.\ multipliers, independent of the data, with
$\E[\xi_i]=0$, $\E[\xi_i^2]=1$, and $\E[\abs{\xi_i}^{2+\d_\xi}]<\infty$ for some $\d_\xi>0$
(e.g.\ $\xi_i\sim N(0,1)$ or Rademacher).
Define the \emph{multiplier bootstrap} analogue of \eqref{eq:MB_tstat} by
\begin{equation}
  T_{n,\Phi}^*(K)
  :=\frac{1}{\sqrt n}\sum_{i=1}^n \xi_i\,\frac{\hat\psi_{i,\Phi}(K)}{\hat\sigma_{\Phi}(K)}.
  \label{eq:MB_tstat_star}
\end{equation}
Let $c_{1-\a}^*(K)$ denote the conditional $(1-\a)$ quantile of $\abs{T_{n,\Phi}^*(K)}$ given the sample.
Then a two-sided $(1-\a)$ MB confidence interval for $\t_0^{\Phi}$ at a \emph{fixed} $K$ is
\begin{equation}
  \mathrm{CI}_{n,\Phi}(K)
  :=\Big[\ \tilde{\Phi}(K)\ \pm\ c_{1-\a}^*(K)\,\frac{\hat\sigma_{\Phi}(K)}{\sqrt n}\ \Big].
  \label{eq:MB_CI_fixedK}
\end{equation}

\subsubsection*{Bootstrap-Lepski choice of sieve dimension $K$}

Let $\mathcal K_n=\{K_1<\cdots<K_{M_n}\}$ be a dyadic (or geometric) grid of sieve dimensions,
with $M_n\lesssim\log n$ as in \citet*{chen2025adaptive}.
Let $K_{\max}:=K_{M_n}$ satisfy the usual series regularity
\begin{equation}
  K_{\max}\log K_{\max}/n=o(1),
  \label{eq:MB_Kmax_growth}
\end{equation}
so that all sieve estimators and the matrix inverses involved are well-defined wpa1.

For $K,K'\in\mathcal K_n$, define the empirical s.d. of the \emph{contrast}
$\tilde{\Phi}(K)-\tilde{\Phi}(K')$ by the influence-score analogue of
\citet*[][eq.\ (9)]{chen2025adaptive}:
\begin{equation}
  \hat\sigma_{\Phi}^2(K,K')
  :=\frac1n\sum_{i=1}^n \big(\hat\psi_{i,\Phi}(K)-\hat\psi_{i,\Phi}(K')\big)^2.
  \label{eq:MB_sigma_contrast}
\end{equation}
The multiplier bootstrap analogue of the contrast is
\begin{equation}
  T_{n,\Phi}^*(K,K')
  :=\frac{1}{\sqrt n}\sum_{i=1}^n \xi_i\,\frac{\hat\psi_{i,\Phi}(K)-\hat\psi_{i,\Phi}(K')}{\hat\sigma_{\Phi}(K,K')}.
  \label{eq:MB_contrast_star}
\end{equation}
Let $ \hat\vartheta_{1-\hat\a_n}^*$ be the conditional $(1-\hat\a_n)$ quantile of
  $\max_{K<K'\in\mathcal K_n}\abs{T_{n,\Phi}^*(K,K')}$ and
\begin{equation}\label{eq:MB_theta_quantile}
  \hat\a_n:=\min\left\{0.5,\ \sqrt{\frac{\log K_{\max}}{K_{\max}}}\right\},
\end{equation}
The choice of $\hat\a_n$ is the scalar analogue of
\citet*[][Step~2 of Procedure~1]{chen2025adaptive} and ensures that
$\hat\vartheta_{1-\hat\a_n}^*$ diverges slowly.

Then, define the data-driven sieve dimension for estimation
\begin{equation}
  \hat K_{\Phi}
  :=\min\Big\{K\in\mathcal K_n:\ \max_{K'>K,\ K'\in\mathcal K_n}
    \frac{\abs{\tilde{\Phi}(K)-\tilde{\Phi}(K')}}{\hat\sigma_{\Phi}(K,K')}
  \le 1.1\,\hat\vartheta_{1-\hat\a_n}^*\Big\}.
  \label{eq:MB_Khat}
\end{equation}
This is the scalar version of \citet*[][Step~3 of Procedure~1]{chen2025adaptive}.
Intuitively, $\hat K_{\Phi}$ is the smallest sieve dimension such that increasing the dimension further does not change the functional estimator beyond its estimated sampling noise.

For inference, to ensure the sieve bias is negligible,
we set an ``inference'' dimension $\tilde K_{\Phi}$ as a mild enlargement of $\hat K_{\Phi}$:
\begin{equation}
  \tilde K_{\Phi}
  :=\min\Big\{K\in\mathcal K_n:\ K\ge \hat K_{\Phi}(\log n)^{\kappa}\Big\}\wedge K_{\max},\
  \text{ with } \kappa\in (0,1) \text{ fixed.}
  \label{eq:MB_Ktilde}
\end{equation}
The logarithmic inflation is a standard ``undersmoothing'' device in series inference:
if $\hat K_{\Phi}$ is of the same order as the MSE-optimal dimension, then $\tilde K_{\Phi}$ makes
the sieve approximation bias $K^{-s/d}$ smaller than the standard deviation by a factor that vanishes polynomially in $\log n$. Given $\tilde K_{\Phi}$, the final CI is $\mathrm{CI}_{n,\Phi}(\tilde K_{\Phi})$ as in \eqref{eq:MB_CI_fixedK}.

The theoretical justification of the inference procedure is similar to CCK and hence is omitted.
    
\begin{rem}[Practical implementation]\label{rem:MB_practice}
  In practice, the suprema/maxima in \eqref{eq:MB_theta_quantile} and \eqref{eq:MB_Khat} are taken over the finite
  grid $\mathcal K_n$, and the quantile $\hat\vartheta_{1-\hat\a_n}^*$ is approximated using a large number $B$ of
  independent MB draws of $(\xi_i)_{i=1}^n$.
  Following \citet{chen2025adaptive}, a dyadic grid (so $M_n\lesssim \log n$) is recommended for stability. We take $\kappa =1/2$ in the undersmoothing CI construction in the Monte Carlo studies.
\end{rem}

\section{Analysis of Quadratic Submanifold Integrals}

In this section, we provided some detailed analysis of the minimax convergence rate for the estimation of quadratic submanifold integrals, which features a fascinating ``elbow phenomenon''.

\begin{lem}[Characterization of Minimax Rates for Quadratic Submanifold Integral Functionals]\label{lem:Rate_Quad} Let the linear functional $DQ(h_0)[v] = 2\int_\cM h_0(x)v(x)w(x)d\cH^m(x)$ satisfies Assumption \ref{assu:w} with weight $2h_0(x)w(x)$ and $\int_\cM h_0(x)w(x)d\cH^m(x) \neq 0$. Then:
\begin{itemize}
    \item[(a)] Under Assumptions \ref{assu:RegLevelSet}, \ref{assu:Jacob}, \ref{assu:density_x}, \ref{assu:e2_below} and $h_0\in\Lambda^s_c (\mathcal{X})$, a lower bound for the minimax rate is given by
    \begin{equation}\label{eq:quad_lb}
    \ul{r}_{Quad,n} := \max\left\{n^{-\frac{s}{2s + d - m}}, n^{-\frac{4s}{4s + 2d - m}}\right\}
    \end{equation}
    \item[(b)]  Under Assumptions \ref{assu:RegLevelSet}, \ref{assu:density_x}, \ref{assu:e2_below}, \ref{assu:e2_moment} and \ref{assu:Cond_Sieve}, an upper bound for the minimax rate is given by
    \begin{equation}\label{eq:quad_ub}
    \ol{r}_{Quad,n} := \max\left\{n^{-\frac{s}{2s + d - m}}, n^{-\frac{2s}{2s + 2d - m}}\right\}
    \end{equation}
    \item[(c)] Under Assumptions for part(b) and $s> m/2$, the minimax optimal rate for the estimation of $Q(h_0)$ is given by $$r_{Quad,n} := n^{-\frac{s}{2s + d - m}} =\ul{r}_{Quad,n} = \ol{r}_{Quad,n}.$$ In addition, this rate can be attained by the split-sample sieve estimator
\[
\hat{\t}_{SS}:=\int_{{\cal M}}\hat{h}_{1}\left(x\right)\hat{h}_{2}\left(x\right)w\left(x\right)d{\cal H}^{m}\left(x\right),
\]
where $\hat{h}_{1}$ and $\hat{h}_{2}$ are two first-stage linear
sieve estimators of $h_{0}$ estimated from two split subsamples, each with sample size $n/2$ and sieve dimension $K^*_{n}\asymp n^{d/(2s+d-m)}$.
\end{itemize}
\end{lem}

\begin{rem}\label{rem:Quad_tworegions}
The lower bound in part~(a) is a \emph{global} minimax lower bound over the full H\"older ball $\Lambda_c^s(\cX)$:
the first term $n^{-s/(2s+d-m)}$ is derived via the linearization at a fixed $h_0$ with $\int_\cM h_0 w\,d\cH^m\neq 0$
(so that $DQ(h_0)\neq 0$), while the second term $n^{-4s/(4s+2d-m)}$ is obtained from a multi-hypothesis
construction centered at $h\equiv 0$ (where $DQ(0)=0$ and the quadratic structure drives the rate).
The overall lower bound is the maximum of these two contributions, taken over different regions of the parameter space.
Note that for $s\le m/2$, the nonlinear exponents in the lower bound ($4s/(4s+2d-m)$) and the upper bound ($2s/(2s+2d-m)$) do not coincide, so in that regime the precise minimax rate for $Q(h_0)$ is not fully pinned down by parts~(a)--(b); part~(c) resolves this only under $s>m/2$.
\end{rem}

\begin{proof}[\textbf{Proof of Lemma \ref{lem:Rate_Quad}(a)}]
By Corollary \ref{cor:NL_rate_LB}, the linear rate $n^{-\frac{s}{2s+d-m}}$ is a valid
lower bound with the linearization remainder $R_{n}=O\left(b_{n}^{2s}\right)$, which can be shown explicitly using the two-point construction
$h_{0}\equiv1$ and $h_{1}(x)=h_{0}(x)+b_{n}^{s}K_{d-m}\left(\frac{g\left(x\right)}{b_{n}}\right)$, which implies:
\begin{align*}
\G\left(h_{1}\right)-\G\left(h_{0}\right) & =\int_{{\cal M}}\left(\left(1+b_{n}^{s}K_{d-m}\left(0\right)\right)^{2}-1^{2}\right)w\left(x\right)d{\cal H}^{m}\left(x\right)\\
 & =\int_{{\cal M}}\left(2b_{n}^{s}K_{d-m}\left(0\right)+b_{n}^{2s}K_{d-m}^{2}\left(0\right)\right)w\left(x\right)d{\cal H}^{m}\left(x\right)\\
 & =\underset{D_{h}\left(h_{0}\right)\left[h_{1}-h_{0}\right]}{\underbrace{2b_{n}^{s}\int_{{\cal M}}K_{d-m}\left(0\right)w\left(x\right)d{\cal H}^{m}\left(x\right)}}+\underset{R_{n}}{\underbrace{b_{n}^{2s}\int_{{\cal M}}K_{d-m}^{2}\left(0\right)w\left(x\right)d{\cal H}^{m}\left(x\right)}}\\
 & =2b_{n}^{s}\int_{{\cal M}}K_{d-m}\left(0\right)w\left(x\right)d{\cal H}^{m}\left(x\right)+O\left(b_{n}^{2s}\right)\\
 & \asymp b_{n}^{s}
\end{align*}

We now seek to establish the other lower bound $n^{-\frac{4s}{4s+2d-m}}$,
based on the multi-hypothesis construction in Chapter 2.7.5 of \cite{Tsybakov2009}. Specifically, let $\mu_{0}$ be the Dirac measure concentrated
on $h_{0}\left(x\right)\equiv0$ and let $\mu_{1}$ be a discrete
measure supported on the finite set of functions
\[
h_{\o}\left(x\right):=\sum_{k=1}^{M}\o_{k}\varphi_{k}\left(x\right),
\]
where:
\begin{itemize}
\item $M_{n}:=b_{n}^{-m}$
\item $\varphi_{k}$$\left(x\right):=b_{n}^{s}K\left(\frac{x-x_{k}}{b_{n}}\right)$
with and $\left(x_{k}\right)$ is a mesh grid on ${\cal M}=\left\{ x:g\left(x\right)=0\right\} $
\item $\o_{k}\sim_{i.i.d.}2Bernoulli\left(0.5\right)-1$, i.e., taking values
$\pm1$ with equal probabilities 0.5.
\end{itemize}
Under the above construction, we have
\begin{align}
\G\left(h_{\o}\right) & =\int_{\left\{ g\left(x\right)=0\right\} }\left(\sum_{k=1}^{M_{n}}\o_{k}\varphi_{k}\left(x\right)\right)^{2}d{\cal H}^{m}\left(x\right)\nonumber \\
 & =\sum_{k=1}^{M_{n}}\int_{\left\{ g\left(x\right)=0\right\} }\varphi_{k}^{2}\left(x\right)d{\cal H}^{m}\left(x\right)\nonumber \\
 & \asymp M_{n}b_{n}^{2s+m}\nonumber \\
 & \asymp b_{n}^{-m}b_{n}^{2s+m}\nonumber \\
 & \asymp b_{n}^{2s}\label{eq:quad_multi_b2s}
\end{align}
For $k=0,1$, let $P_{k}^{n}$ denote the joint probability distribution of the random sample $\{(X_{i},Y_{i})\}_{i=1}^{n}$ from the nonparametric
regression model with $h_{k}$ specified above. Then the $\chi^2$-divergence
can be derived similar to those in Chapter 2.7.5 of \cite{Tsybakov2009}, with
\begin{align*}
\chi^{2}\left(P_{1}^{n},P_{0}^{n}\right) & \leq\exp\left(n^{2}b_{n}^{2\left(2s+d\right)}M_{n}\right)\\
 & =\exp\left(n^{2}b_{n}^{2\left(2s+d\right)}b_{n}^{-m}\right)\\
 & =\exp\left(n^{2}b_{n}^{4s+2d-m}\right)
\end{align*}
To ensure $\chi^{2}\left(P_{1}^{n},P_{0}^{n}\right)\leq c$, we set
\begin{equation}
   b_{n}\asymp n^{-\frac{2}{4s+2d-m}}. \label{eq:quad_multi_chi2}
\end{equation}
Hence, by Theorem 2.15 of \cite{Tsybakov2009}, we obtain the following lower bound based on \eqref{eq:quad_multi_b2s} and \eqref{eq:quad_multi_chi2}:
\[
b_{n}^{2s}\asymp n^{-\frac{4s}{4s+2d-m}}
\]

Combining the above rate with the previous $n^{-\frac{s}{2s+d-m}}$ above, we obtain a valid lower bound
\[
\ul r_{Quad,n}:=\max\left\{n^{-\frac{s}{2s+d-m}},n^{-\frac{4s}{4s+2d-m}}\right\}.
\]
\end{proof}

\begin{proof}[\textbf{Proof of Lemma \ref{lem:Rate_Quad}(b)}]
Recall the split-sample sieve estimator
\[
\hat{\t}:=\int_{{\cal M}}\hat{h}_{1}(x)\,\hat{h}_{2}(x)\,w(x)\,d{\cal H}^{m}(x),
\]
where $\hat h_{r}$ is the sieve LS estimator computed on subsample $I_r$ of size $n_r:=|I_r|$ ($r=1,2$), with $n_1\asymp n_2\asymp n$ and sieve dimension $K_n$ (write $K:=K_n$ for brevity).
Let $\td h_r:=P_{K,r}h_0$ denote the (subsample-$r$) empirical LS projection of $h_0$ onto ${\H}_{K}$, i.e.,
\[
\td h_r(x)=\ol b^{K}(x)'\,\ol G_r^{-1}\,\frac1{n_r}\sum_{i\in I_r}\ol b^{K}(X_i)\,h_0(X_i),
\qquad
\ol G_r:=\frac1{n_r}\sum_{i\in I_r}\ol b^{K}(X_i)\ol b^{K}(X_i)'.
\]
Then the stochastic part admits the orthonormalized linear representation
\begin{equation}
\label{eq:ss_stoch_rep}
\hat h_r(x)-\td h_r(x)=\ol b^{K}(x)'\,\ol G_r^{-1}\,\frac1{n_r}\sum_{i\in I_r}\ol b^{K}(X_i)\e_i,
\qquad r\in\{1,2\}.
\end{equation}

Write $\D_r:=\hat h_r-h_0$. Then
\begin{align}
\hat\t-\t_0
&=\int_{\cM}(\hat h_1\hat h_2-h_0^2)w\,d\cH^m \notag\\
&=\underbrace{\int_{\cM}(\hat h_1-h_0)h_0w\,d\cH^m}_{=:T_{1}}
+\underbrace{\int_{\cM}(\hat h_2-h_0)h_0w\,d\cH^m}_{=:T_{2}}
+\underbrace{\int_{\cM}(\hat h_1-h_0)(\hat h_2-h_0)w\,d\cH^m}_{=:T_{3}}.
\label{eq:ss_quad_decomp}
\end{align}

\noindent\textbf{Linear terms $T_1,T_2$.}
Each $T_r$ is a linear submanifold functional of $\hat h_r-h_0$ with bounded weight $h_0w$.
Applying Theorem~\ref{thm:Rate_Sieve} on each subsample (size $n_r\asymp n$) yields
\begin{equation}
\label{eq:ss_lin_terms}
T_{1}=O_p\!\left(K^{-\frac{s}{d}}+\sqrt{\frac{K^{\frac{d-m}{d}}}{n}}\right),
\qquad
T_{2}=O_p\!\left(K^{-\frac{s}{d}}+\sqrt{\frac{K^{\frac{d-m}{d}}}{n}}\right).
\end{equation}

\noindent\textbf{Quadratic term $T_3$: bias/cross/stochastic pieces.}
Decompose $\D_r=(\hat h_r-\td h_r)+(\td h_r-h_0)$ and expand
\begin{align}
T_3
&=\int_{\cM}(\hat h_1-\td h_1)(\hat h_2-\td h_2)w\,d\cH^m
+\int_{\cM}(\td h_1-h_0)(\hat h_2-\td h_2)w\,d\cH^m \notag\\
&\quad+\int_{\cM}(\hat h_1-\td h_1)(\td h_2-h_0)w\,d\cH^m
+\int_{\cM}(\td h_1-h_0)(\td h_2-h_0)w\,d\cH^m \notag\\
&=:T_{3,\mathrm{stoch}}+T_{3,1}+T_{3,2}+T_{3,\mathrm{bias}}.
\label{eq:ss_T3_split}
\end{align}

\emph{Bias part.} By Assumption~\ref{assu:Cond_Sieve}(iv)(v),
$\|\td h_r-h_0\|_{\infty}=O_p(K^{-s/d})$.
Hence
\[
|T_{3,\mathrm{bias}}|
\le \|w\|_{\infty}\,\cH^{m}(\cM)\,\|\td h_1-h_0\|_{\infty}\,\|\td h_2-h_0\|_{\infty}
=O_p\!\left(K^{-\frac{2s}{d}}\right).
\]

\emph{Cross parts.}
Conditional on the subsample designs, $\hat h_r-\td h_r$ is centered and has no approximation bias.
Moreover, the weights $(\td h_{\bar r}-h_0)w$ are bounded by $O_p(K^{-s/d})$.
Applying Theorem~\ref{thm:Rate_Sieve} to the centered linear functional on each subsample gives
\[
T_{3,1}=O_p\!\left(K^{-\frac{s}{d}}\sqrt{\frac{K^{\frac{d-m}{d}}}{n}}\right),
\qquad
T_{3,2}=O_p\!\left(K^{-\frac{s}{d}}\sqrt{\frac{K^{\frac{d-m}{d}}}{n}}\right).
\]
These cross terms are of smaller order than $\max\{K^{-s/d},\sqrt{K^{(d-m)/d}/n}\}$ and can be absorbed.

\smallskip
\noindent\emph{Cross terms $T_{3,1}$ and $T_{3,2}$.}
Recall
\[
T_{3,1}:=\int_{\cM}(\tilde h_1-h_0)(\hat h_2-\tilde h_2)\,w\,d\cH^m,
\qquad
T_{3,2}:=\int_{\cM}(\hat h_1-\tilde h_1)(\tilde h_2-h_0)\,w\,d\cH^m .
\]

We treat $T_{3,1}$; the bound for $T_{3,2}$ is identical by symmetry (swap the roles of the two splits).
Let
\[
A_{1,n}:=\|\tilde h_1-h_0\|_\infty .
\]
On the event $\{A_{1,n}>0\}$ define the \emph{random but bounded} weight
\[
w_{1,n}(x):=w(x)\frac{\tilde h_1(x)-h_0(x)}{A_{1,n}},
\qquad x\in\cM,
\]
and set $w_{1,n}\equiv 0$ on $\{A_{1,n}=0\}$.
Then $\|w_{1,n}\|_\infty\le \|w\|_\infty$ and, \emph{crucially}, $w_{1,n}$ is measurable w.r.t.\ the first split and hence independent of the second-split estimator $\hat h_2$.
With this normalization,
\begin{equation}\label{eq:T31_normalized}
T_{3,1}=A_{1,n}\int_{\cM}(\hat h_2-\tilde h_2)(x)\,w_{1,n}(x)\,d\cH^m(x)
\equiv A_{1,n}\,L_{w_{1,n}}(\hat h_2-\tilde h_2),
\end{equation}
where $L_{w_{1,n}}(f):=\int_{\cM} f(x) w_{1,n}(x)\,d\cH^m(x)$ is a linear submanifold functional.

\smallskip
\noindent\emph{Step 1: bound $A_{1,n}=\|\tilde h_1-h_0\|_\infty$.}
By Assumption~\ref{assu:Cond_Sieve}(v), there exists $h_{0,K}\in\H_K$ with $\|h_{0,K}-h_0\|_\infty\lesssim K^{-s/d}$.
Since $\tilde h_1=P_{K,1}h_0$ is the empirical LS projection onto $\H_K$ based on split~1,
\[
\|\tilde h_1-h_0\|_\infty
\le \|P_{K,1}(h_0-h_{0,K})\|_\infty+\|h_{0,K}-h_0\|_\infty
\le \|P_{K,1}\|_\infty\|h_0-h_{0,K}\|_\infty+\|h_{0,K}-h_0\|_\infty.
\]
Assumption~\ref{assu:Cond_Sieve}(iv) gives $\|P_{K,1}\|_\infty=O_p(1)$, so
\begin{equation}\label{eq:A1n_rate}
A_{1,n}=\|\tilde h_1-h_0\|_\infty=O_p(K^{-s/d}).
\end{equation}

\smallskip
\noindent\emph{Step 2: bound the centered linear functional $L_{w_{1,n}}(\hat h_2-\tilde h_2)$.}
Condition on the first split (so that $w_{1,n}$ is fixed) and apply Theorem~\ref{thm:Rate_Sieve}(1) to the second-split sieve LS estimator and the (now deterministic) weight $w_{1,n}$:
\[
\sqrt{n_2}\,L_{w_{1,n}}(\hat h_2-\tilde h_2)
=
\frac{1}{\sqrt{n_2}}\sum_{i\in I_2} v^*_{K,2}(w_{1,n})(X_i)\,\e_i + o_p\!\left(\|v^*_{K,2}(w_{1,n})\|_{sd}\right),
\]
and hence
\begin{equation}\label{eq:Lw1n_centered_rate}
L_{w_{1,n}}(\hat h_2-\tilde h_2)=O_p\!\left(\frac{\|v^*_{K,2}(w_{1,n})\|_{sd}}{\sqrt{n}}\right).
\end{equation}
By the growth-rate lemma for sieve Riesz representers (Lemma~\ref{lem:NormRate_d-m}), and using $\|w_{1,n}\|_\infty\le \|w\|_\infty$, we have the \emph{uniform} upper bound
\begin{equation}\label{eq:vk_uniform_bound}
\|v^*_{K,2}(w_{1,n})\|_{sd}\ \lesssim\ \|w_{1,n}\|_\infty\,K^{(d-m)/(2d)}
\ \le\ \|w\|_\infty\,K^{(d-m)/(2d)} .
\end{equation}
Combining \eqref{eq:Lw1n_centered_rate} and \eqref{eq:vk_uniform_bound} yields
\begin{equation}\label{eq:Lw1n_final}
L_{w_{1,n}}(\hat h_2-\tilde h_2)=O_p\!\left(\sqrt{\frac{K^{(d-m)/d}}{n}}\right).
\end{equation}

\smallskip
\noindent\emph{Step 3: conclude for $T_{3,1}$ (and $T_{3,2}$).}
Combining \eqref{eq:T31_normalized}, \eqref{eq:A1n_rate}, and \eqref{eq:Lw1n_final},
\[
T_{3,1}=O_p\!\left(K^{-s/d}\sqrt{\frac{K^{(d-m)/d}}{n}}\right).
\]
The same argument with indices swapped gives
\[
T_{3,2}=O_p\!\left(K^{-s/d}\sqrt{\frac{K^{(d-m)/d}}{n}}\right).
\]

\smallskip
\noindent Step 4 (main stochastic term $T_{3,\mathrm{stoch}}$).
Define the deterministic $K\times K$ matrix
\[
A_K:=\int_{\cM}\ol b^K(x)\ol b^K(x)'\,|w(x)|\,d\cH^m(x).
\]
Using \eqref{eq:ss_stoch_rep} and Cauchy--Schwarz, we may bound $|T_{3,\mathrm{stoch}}|$ by replacing $w$ with $|w|$.
Let
\[
\eta_r:=\frac1{n_r}\sum_{i\in I_r}\ol b^K(X_i)\e_i\in\R^{K}.
\]
Then
\begin{equation}
\label{eq:ss_T3stoch_matrix}
T_{3,\mathrm{stoch}}=\eta_1'\,\ol G_1^{-1}\,A_K\,\ol G_2^{-1}\,\eta_2\quad\text{(up to sign)}.
\end{equation}
Conditioning on $X_1^n$ and using the independence between the two subsamples (and $\E[\e_i\mid X_i]=0$), we have $\E[T_{3,\mathrm{stoch}}\mid X_1^n]=0$.
Moreover,
\[
\E[\eta_r\eta_r'\mid X_1^n]
=\frac1{n_r^2}\sum_{i\in I_r}\E[\e_i^2\mid X_i]\,\ol b^K(X_i)\ol b^K(X_i)'
\ \preceq\ \frac{\bar\s^2}{n_r}\,\ol G_r,
\]
where $\bar\s^2:=\sup_x\E[\e_i^2\mid X_i=x]<\infty$ by Assumption~\ref{assu:e2_moment}.
Therefore, by the tower property and cyclicity of trace,
\begin{align*}
\E[T_{3,\mathrm{stoch}}^2\mid X_1^n]
&\le \frac{\bar\s^4}{n_1n_2}\,\tr\!\big(\ol G_1^{-1}A_K\ol G_2^{-1}A_K\big)\\
&\le \frac{\bar\s^4}{n_1n_2}\,\|\ol G_1^{-1}\|_{op}\,\|\ol G_2^{-1}\|_{op}\,\tr(A_K^2).
\end{align*}
Under Assumption~\ref{assu:Cond_Sieve}(i)(iii), $\|\ol G_r^{-1}\|_{op}=O_p(1)$.
It remains to bound $\tr(A_K^2)=\|A_K\|_F^2$.

\smallskip
\noindent\emph{Claim:} $\tr(A_K^2)\lesssim K^{\frac{2d-m}{d}}$.

\smallskip
\noindent\emph{Justification of the claim.}
Apply the manifold decomposition \eqref{eq:PieceLebInt} (as in the proof of Lemma~\ref{lem:NormRate_d-m}) to write
$A_K=\sum_{j=1}^{\ol j}A_{K,j}$,
where $\ol j<\infty$ is the number of charts and
\[
A_{K,j}:=\int_{{\cal U}_j}\ol b^K(\varphi_j(u))\ol b^K(\varphi_j(u))'\,\ol w_j(u)\,du,
\]
with bounded $\ol w_j(u)$.
Using $\|\sum_{j}A_{K,j}\|_F^2\le \ol j\sum_j\|A_{K,j}\|_F^2$, it suffices to show $\|A_{K,j}\|_F^2\lesssim K^{(2d-m)/d}$ uniformly in $j$.
Fix $j$.
For each $k,k'$, write
\[
(A_{K,j})_{kk'}=\int_{{\cal U}_j}\ol b_k^K(\varphi_j(u))\,\ol b_{k'}^K(\varphi_j(u))\,\ol w_j(u)\,du.
\]
Using the same tensor-product decomposition as in the proof of Lemma~\ref{lem:NormRate_d-m},
$\ol b_{k}^K(\varphi_j(u))=\ol b_{k,(m)}^K(u)\cd\ol b_{k,-(m)}^K(\psi_j(u))$.
Then, for fixed $k$ and fixed $k'_{-(m)}$, the frame inequality for $\{\ol b_{k',(m)}^K\}$ on ${\cal U}_j$ yields
\[
\sum_{k'_{(m)}} (A_{K,j})_{kk'}^2
\ \lesssim\ \int_{{\cal U}_j}\Big(\ol b_k^K(\varphi_j(u))\Big)^2\Big(\ol b_{k',-(m)}^K(\psi_j(u))\Big)^2\,\ol w_j(u)^2\,du.
\]
Summing over $k'_{-(m)}$ and using the tensor-product structure gives
$\sum_{k'_{-(m)}}(\ol b_{k',-(m)}^K(\psi_j(u)))^2\lesssim J^{d-m}=K^{(d-m)/d}$ pointwise in $u$.
Hence,
\[
\sum_{k'=1}^{K}(A_{K,j})_{kk'}^2\lesssim K^{\frac{d-m}{d}}\int_{{\cal U}_j}\Big(\ol b_k^K(\varphi_j(u))\Big)^2\,\ol w_j(u)^2\,du.
\]
Finally, summing over $k$ and using $\sum_{k=1}^{K}(\ol b_k^K(x))^2=\|\ol b^K(x)\|^2\lesssim K$ (Assumption~\ref{assu:Cond_Sieve}(ii) and $\l_{min}(G)>0$) yields
\[
\|A_{K,j}\|_F^2
=\sum_{k=1}^K\sum_{k'=1}^K (A_{K,j})_{kk'}^2
\lesssim K^{\frac{d-m}{d}}\int_{{\cal U}_j}\Big(\sum_{k=1}^{K}(\ol b_k^K(\varphi_j(u)))^2\Big)\ol w_j(u)^2\,du
\lesssim K^{\frac{d-m}{d}}\cd K=K^{\frac{2d-m}{d}}.
\]
This proves the claim.

\smallskip
\noindent Returning to $T_{3,\mathrm{stoch}}$, we thus have
$\E[T_{3,\mathrm{stoch}}^2]=O(K^{(2d-m)/d}/n^2)$ and hence
\[
T_{3,\mathrm{stoch}}=O_p\!\left(\frac1n\sqrt{K^{\frac{2d-m}{d}}}\right).
\]

\smallskip
\noindent Step 5 (collect rates and optimize over $K$).
Combining \eqref{eq:ss_quad_decomp}, \eqref{eq:ss_lin_terms}, \eqref{eq:ss_T3_split} and the bounds above gives
\[
\hat\t-\t_0
=O_p\!\left(
K^{-\frac{s}{d}}
+\sqrt{\frac{K^{\frac{d-m}{d}}}{n}}
+\frac1n\sqrt{K^{\frac{2d-m}{d}}}
\right).
\]
We now choose $K$ to minimize the maximal order among
\[
T_{1}(K):=K^{-s/d},\qquad
T_{2}(K):=\sqrt{\frac{K^{\frac{d-m}{d}}}{n}},\qquad
T_{3}(K):=\frac1n\sqrt{K^{\frac{2d-m}{d}}}.
\]
Balancing $T_{1}$ and $T_{2}$ yields $K^{1/d}\asymp n^{1/(2s+d-m)}$ and hence the ``linear'' rate
\begin{equation}
K^{-\frac{s}{d}}\asymp\sqrt{\frac{K^{\frac{d-m}{d}}}{n}}\asymp n^{-\frac{s}{2s+d-m}}.\label{eq:quad_SS_lin}
\end{equation}
Balancing $T_{1}$ and $T_{3}$ yields $K^{1/d}\asymp n^{2/(2s+2d-m)}$ and hence
\begin{equation}
K^{-\frac{s}{d}}\asymp\frac1n\sqrt{K^{\frac{2d-m}{d}}}\asymp n^{-\frac{2s}{2s+2d-m}}.\label{eq:quad_SS_2s}
\end{equation}
Combining \eqref{eq:quad_SS_lin} and \eqref{eq:quad_SS_2s}, we obtain
\begin{align}
\ol r_{Quad,n}
&:=\max\left\{ n^{-\frac{s}{2s+d-m}},\ n^{-\frac{2s}{2s+2d-m}}\right\}
\equiv\begin{cases}
 n^{-\frac{s}{2s+d-m}}, & \text{if }s\geq\frac{m}{2},\\
 n^{-\frac{2s}{2s+2d-m}}, & \text{if }s<\frac{m}{2}.
\end{cases}\label{eq:quad_SS_rate}
\end{align}
Finally, to verify the regime split: if $K^{1/d}\asymp n^{1/(2s+d-m)}$, then $T_3(K)=O(T_2(K))$ is equivalent to $K\lesssim n$, i.e. $n^{d/(2s+d-m)}\lesssim n$, which holds iff $s\ge m/2$.
If $K^{1/d}\asymp n^{2/(2s+2d-m)}$, then $T_2(K)=O(T_3(K))$ is equivalent to $n\lesssim K$, i.e. $n\lesssim n^{2d/(2s+2d-m)}$, which holds iff $s\le m/2$.
This matches \eqref{eq:quad_SS_rate}.
\end{proof}

\begin{proof}[\textbf{Proof of Lemma \ref{lem:Rate_Quad}(c)}]
Recalling that
\begin{align*}
\ol r_{Quad,n} & :=\max\left\{ n^{-\frac{s}{2s+d-m}},n^{-\frac{2s}{2s+2d-m}}\right\} \equiv\begin{cases}
n^{-\frac{s}{2s+d-m}}, & \text{if }s\geq\frac{m}{2},\\
n^{-\frac{2s}{2s+2d-m}} & \text{if }s\leq\frac{m}{2}.
\end{cases}
\end{align*}
and
\[
\ul r_{Quad,n}:=\max\left\{ n^{-\frac{s}{2s+d-m}},n^{-\frac{4s}{4s+2d-m}}\right\} =\begin{cases}
n^{-\frac{s}{2s+d-m}}, & \text{if }s\geq\frac{3m-2d}{4},\\
n^{-\frac{4s}{4s+2d-m}}, & \text{if }s\leq\frac{3m-2d}{4}.
\end{cases}
\]
Note that
\[
\frac{m}{2}>\frac{m}{4}-\frac{d-m}{2}=\frac{3m-2d}{4}
\]
Hence, whenever $s\geq\frac{m}{2}$, we have
\[
\ol r_{Quad,n}=\ul r_{Quad,n}=n^{-\frac{s}{2s+d-m}}
\]
based on which we may conclude that the minimax optimal rate $n^{-\frac{s}{2s+d-m}}$
is attained by the split-sample sieve estimator.
\end{proof}

\section{Rate Analysis for Nadaraya-Watson First Stage}\label{subsec:Kernel}

In this section, we show that the key rate-acceleration result
in the previous subsection, as one may expect, also holds for kernel-based
nonparametric methods. To illustrate this point, we consider the case
where $h_{0}$ is a given by a conditional expectation and $\hat{h}$
is given by the Nadaraya-Watson kernel estimator and show that
our key result continues to hold.

Formally, let $h_{0}\left(x\right)=\E\left[\rest{Y_{i}}X_{i}=x\right]$
and
\[
\hat{h}\left(x\right)=\frac{\frac{1}{nb_{n}^{d}}\sum_{i=1}^{n}K_{d}\left(\frac{X_{i}-x}{b_{n}}\right)Y_{i}}{\frac{1}{nb_{n}^{d}}\sum_{i=1}^{n}K_{d}\left(\frac{X_{i}-x}{b_{n}}\right)}
\]
where $b_{n}\downto0$ is a bandwidth parameter and $K_{d}$ is a
multivariate kernel of smoothness order $s$.

We consider the convergence of the plug-in estimator $$\hat{\t}:=\int_{{\cal M}}\hat{h}\left(x\right)w\left(x\right)p\left(x\right)d{\cal H}^{m}\left(x\right)$$ to $\t_{0}=\int_{{\cal M}}h_{0}\left(x\right)w\left(x\right)p_{0}\left(x\right)d{\cal H}^{m}\left(x\right),$ treating $p(x)$ in the integrand as known. Under this simplification, $\hat{\t} = L(\hat{h})$ is a plug-in estimator of a linear functional as in Section \ref{subsec:Minimax_Lin-UB}.

Our results can also be extended to analyze the composite estimator $\G(\hat{h},\hat{p})=\int_{{\cal M}}\hat{h}\left(x\right)w\left(x\right)\hat{p}\left(x\right)d{\cal H}^{m}\left(x\right)$ for $\G(h_0,p_0)$, treated as a nonlinear integral functional $\G(h_0,p_0)$ of the vector-valued first stage $(h_0,p_0)$. This can be easily handled using the similar kernel arguments as below in combination of the linearization arguments in Section \ref{subsec:Nlh-UB}, and the incorporation of the density estimator does not affect the rate of the estimator.

\begin{assumption}[Kernel Smoothness]
\label{assu:ker_Ksmooth} $K_{d}$ is a $d$-dimensional product
kernel
\[
K_{d}\left(x\right)=\prod_{j=1}^{d}K\left(x_{j}\right),
\]
where $K$ is a univariate kernel of smoothness order $s$, i.e.,
(i) $K\left(u\right)=K\left(-u\right)$, (ii) $\int K\left(u\right)du=1$,
(iii)$\left|K\left(u\right)\right|\leq M<\infty$, (iv) $\int u^{j}k\left(u\right)du=0$
for $j=1,...,s-1$, and (v) $\kappa_{s}:=\int x_{j}^{s}K\left(x\right)dx\in\left(0,\infty\right)$.
\end{assumption}
\begin{thm}[\label{thm:KernANorm}Kernel Nadaraya-Watson First Stage]
 Let $h_{0}\left(\cdot \right)=\E\left[\rest{Y_{i}}X_{i}=\cdot \right] \in \Lambda^s (\mathcal{X})$. Under Assumptions \ref{assu:RegLevelSet} and \ref{assu:ker_Ksmooth}, we have
\[
\text{Bias}\left(\hat{\t}\right)=O\left(b_{n}^{s}\right),\quad\text{Var}\left(\hat{\t}\right)=O\left(\frac{1}{nb_{n}^{d-m}}\right)
\]
and thus
\[
\norm{\hat{\t}-\t_{0}}=O_{p}\left(n^{-\frac{s}{2s+d-m}}\right).
\]
\end{thm}

\begin{proof}[\textbf{Proof of Theorem \ref{thm:KernANorm}}]
Write $a\left(x\right):=h_{0}\left(x\right)p\left(x\right)$, we have

\begin{align}
\hat{\t}-\t_{0} & =\int_{{\cal M}}\left[\hat{h}\left(x\right)-h_{0}\left(x\right)\right]w\left(x\right)p\left(x\right)d{\cal H}^{m}\left(x\right)\nonumber \\
 & =\int_{{\cal M}}\left[\frac{\hat{a}\left(x\right)}{\hat{p}\left(x\right)}-\frac{a\left(x\right)}{p\left(x\right)}\right]w\left(x\right)p\left(x\right)d{\cal H}^{m}\left(x\right)\nonumber \\
 & =\int_{{\cal M}}\left[\frac{\hat{a}\left(x\right)-a\left(x\right)}{p\left(x\right)}-\frac{a\left(x\right)}{p^{2}\left(x\right)}\left(\hat{p}\left(x\right)-p\left(x\right)\right)\right]w\left(x\right)p\left(x\right)d{\cal H}^{m}\left(x\right)+R_{1}\nonumber \\
 & =\int_{{\cal M}}\left[\hat{a}\left(x\right)-a\left(x\right)-h_{0}\left(x\right)\left(\hat{p}\left(x\right)-p\left(x\right)\right)\right]w\left(x\right)d{\cal H}^{m}\left(x\right)+R_{1}\nonumber \\
 & =\int_{{\cal M}}\left[\hat{a}\left(x\right)-h_{0}\left(x\right)\hat{p}\left(x\right)\right]w\left(x\right)d{\cal H}^{m}\left(x\right)+R_{1}\nonumber \\
 & =\int_{{\cal M}}\frac{1}{nb_{n}^{d}}\sum_{i=1}^{n}K\left(\frac{x-X_{i}}{b_{n}}\right)\left(Y_{i}-h_{0}\left(x\right)\right)w\left(x\right)d{\cal H}^{m}\left(x\right)+R_{1}\nonumber \\
 & =\underset{T_{1}}{\underbrace{\frac{1}{nb_{n}^{d}}\sum_{i=1}^{n}\int_{{\cal M}}K\left(\frac{x-X_{i}}{b_{n}}\right)\left(Y_{i}-h_{0}\left(x\right)\right)w\left(x\right)d{\cal H}^{m}\left(x\right)}}+R_{1},\label{eq:Kern_Ldh}
\end{align}
where the remainder term $R_{1}$ is asymptotically negligible.

By \eqref{eq:PieceLebInt}, we may write
\[
T_{1}=\sum_{j=1}^{\ol j}\frac{1}{n}\sum_{i=1}^{n}T_{1ij}
\]
with
\[
T_{1ij}:=\frac{1}{b_{n}^{d}}\int_{\varphi_{j}\left({\cal U}_{j}\right)}K\left(\frac{\varphi_{j}\left(x_{\left(j,m\right)}\right)-X_{i}}{b_{n}}\right)\left(Y_{i}-h_{0}\left(\varphi_{j}\left(x_{\left(j,m\right)}\right)\right)\right)\ol w_{j}\left(Y_{i},\varphi_{j}\left(x_{\left(j,m\right)}\right)\right)dx_{\left(j,m\right)}
\]
where
\[
\ol w_{j}\left(Y_{i},x\right):=\rho_{j}\left(x\right)w\left(x\right)\mathcal{J}\varphi_{j}\left(x\right)
\]
Subsequently, we consider each $T_{1ij}$ separately and suppress
the subscript $j$ for simpler notation.

We carry out the kernel change of variables from $x_{\left(m\right)}$
to $u$ by setting
\[
u:=\frac{x_{\left(m\right)}-X_{i,\left(m\right)}}{b_{n}},\quad\iff\quad x_{\left(m\right)}=X_{i,\left(m\right)}+b_{n}u.
\]
We write ${\cal V}_{u}:=\frac{{\cal V}-X_{i,\left(m\right)}}{b_{n}}$.
\begin{align*}
T_{1ij}=\  & \frac{1}{b_{n}^{d}}\int_{{\cal V}_{u}}K\left(u,\frac{\psi\left(X_{i,\left(m\right)}+b_{n}u\right)-X_{i,-\left(m\right)}}{b_{n}}\right)\left(Y_{i}-h_{0}\left(\varphi\left(X_{i,\left(m\right)}+b_{n}u\right)\right)\right)\\
& \quad \ol w\left(Y_{i},\varphi\left(X_{i,\left(m\right)}+b_{n}u\right)\right)d\left(X_{i,\left(m\right)}+b_{n}u\right)\\
=\  & \frac{b_{n}^{m}}{b_{n}^{d}}\int_{{\cal V}_{u}}K_{\left(m\right)}\left(u\right)K_{-\left(m\right)}\left(\frac{\psi\left(X_{i,\left(m\right)}+b_{n}u\right)-X_{i,-\left(m\right)}}{b_{n}}\right)\left(Y_{i}-h_{0}\left(\varphi\left(X_{i,\left(m\right)}+b_{n}u\right)\right)\right)\\
&\quad \ol w\left(\varphi\left(X_{i,\left(m\right)}+b_{n}u\right)\right)du\\
\end{align*}
Since $h_{0}\left(\varphi\left(X_{i,\left(m\right)}+b_{n}u\right)\right)=h_{0}\left(\varphi\left(X_{i,\left(m\right)}\right)\right)+O\left(b_{n}\right)$
and $\ol w\left(\varphi\left(X_{i,\left(m\right)}+b_{n}u\right)\right)=\ol w\left(\varphi\left(X_{i,\left(m\right)}\right)\right)+O\left(b_{n}\right)$
on the support of $K_{\left(m\right)}$, these smooth terms may be evaluated at $u=0$
with $O\left(b_{n}\right)$ corrections absorbed into $R_{2}$.
However, the transverse kernel argument
\[
\frac{\psi\left(X_{i,\left(m\right)}+b_{n}u\right)-X_{i,-\left(m\right)}}{b_{n}}
=\frac{\psi\left(X_{i,\left(m\right)}\right)-X_{i,-\left(m\right)}}{b_{n}}+\nabla\psi\left(X_{i,\left(m\right)}\right)u+O\left(b_{n}\right)
\]
varies at $O\left(1\right)$ in $u$, so $K_{-\left(m\right)}$ cannot simply be evaluated at $u=0$.
Define the \emph{effective transverse kernel}
\begin{equation}\label{eq:eff_kernel}
\tilde{K}_{-\left(m\right)}\left(\xi;\,x_{\left(m\right)}\right)
:=\int K_{\left(m\right)}\left(u\right)K_{-\left(m\right)}\left(\xi+\nabla\psi\left(x_{\left(m\right)}\right)u\right)du.
\end{equation}
Then
\[
T_{1ij}=\underset{T_{2,ij}}{\underbrace{\frac{1}{b_{n}^{d-m}}\tilde{K}_{-\left(m\right)}\left(\frac{\psi\left(X_{i,\left(m\right)}\right)-X_{i,-\left(m\right)}}{b_{n}};\,X_{i,\left(m\right)}\right)\left(Y_{i}-h_{0}\left(\varphi\left(X_{i,\left(m\right)}\right)\right)\right)\ol w\left(\varphi\left(X_{i,\left(m\right)}\right)\right)}}+R_{2},
\]
where $R_{2}$ collects the $O\left(b_{n}/b_{n}^{d-m}\right)$ corrections from the smooth-term
evaluations and the $O\left(b_{n}\right)$ Taylor remainder in $\psi$.

We claim that $\tilde{K}_{-\left(m\right)}$ is a valid $\left(d-m\right)$-dimensional
kernel of the same order as $K_{-\left(m\right)}$. First, by a change of variable
$\eta=\xi+\nabla\psi\cd u$,
\[
\int\tilde{K}_{-\left(m\right)}\left(\xi;\,x_{\left(m\right)}\right)d\xi
=\int K_{\left(m\right)}\left(u\right)\underbrace{\int K_{-\left(m\right)}\left(\eta\right)d\eta}_{=1}du=1.
\]
Second, for any multi-index $\alpha$ with $\left|\alpha\right|\leq s-1$,
the same change of variable gives
$\int\xi^{\alpha}\tilde{K}_{-\left(m\right)}\left(\xi\right)d\xi
=\sum_{\beta\leq\alpha}\binom{\alpha}{\beta}\left(-1\right)^{\left|\alpha-\beta\right|}
\left[\int K_{\left(m\right)}\left(u\right)\left(\nabla\psi\cd u\right)^{\alpha-\beta}du\right]
\left[\int K_{-\left(m\right)}\left(\eta\right)\eta^{\beta}d\eta\right]$.
Since $K_{-\left(m\right)}$ has vanishing moments up to order $s-1$ and
$K_{\left(m\right)}$ is a symmetric product kernel (so $\int u^{j}K_{\left(m\right)}\left(u\right)du=0$
for odd $j$), each term in the sum vanishes. Hence $\tilde{K}_{-\left(m\right)}$
has the same moment conditions as $K_{-\left(m\right)}$, i.e., it is a kernel of order $s$.
Third, $\int\tilde{K}_{-\left(m\right)}^{2}\left(\xi\right)d\xi<\infty$ since $K_{\left(m\right)}$
and $K_{-\left(m\right)}$ are bounded and compactly supported.

The term $\frac{1}{n}\sum_{i}T_{2,ij}$ is therefore a $\left(d-m\right)$-dimensional
kernel estimator (with the effective kernel $\tilde{K}_{-\left(m\right)}$
in place of $K_{-\left(m\right)}$) and shares the same asymptotic bias and variance rates.
We formalize this via Lemmas~\ref{k-bias} and~\ref{k-var} below,
whose proofs apply with $\tilde{K}_{-\left(m\right)}$ replacing $K_{-\left(m\right)}$ throughout,
yielding $\text{Bias}=O\left(b_{n}^{s}\right)$ and
$\text{Var}=O\left(1/\left(nb_{n}^{d-m}\right)\right)$.
\end{proof}

\begin{lem}[Bias]\label{k-bias}
$\E\left[\frac{1}{n}\sum_{i}T_{2,ij}\right] =  \E\left[T_{2,ij}\right]=O\left(b_{n}^{s}\right)$.
\end{lem}
\begin{proof}
Clearly,
\begin{align*}
&\E\left[T_{2ij}\right] \\
=\ & \frac{1}{b_{n}^{d-m}}\int\left[K_{-\left(m\right)}\left(\frac{\psi\left(z_{\left(m\right)}\right)-z_{-\left(m\right)}}{b_{n}}\right)\left(h_{0}\left(z\right)-h_{0}\left(\varphi\left(z_{\left(m\right)}\right)\right)\right)\ol w\left(\varphi\left(z_{\left(m\right)}\right)\right)\right]p\left(z\right)dz\\
=\  & \frac{1}{b_{n}^{d-m}}\int\int\left[K_{-\left(m\right)}\left(\frac{\psi\left(z_{\left(m\right)}\right)-z_{-\left(m\right)}}{b_{n}}\right)\left(h_{0}\left(z_{\left(m\right)},z_{-\left(m\right)}\right)-h_{0}\left(\varphi\left(z_{\left(m\right)}\right)\right)\right)\ol w\left(\varphi\left(z_{\left(m\right)}\right)\right)\right]\\
&\quad p\left(\rest{z_{-\left(m\right)}}z_{\left(m\right)}\right)dz_{-\left(m\right)} p\left(z_{\left(m\right)}\right)dz_{\left(m\right)}
\end{align*}
Define the change of variable from $z_{-\left(m\right)}$ to $\zeta$
as
\[
\zeta:=\frac{z_{-\left(m\right)}-\psi_{j}\left(z_{\left(m\right)}\right)}{b_{n}},\quad z_{-\left(m\right)}=\psi_{j}\left(z_{\left(m\right)}\right)+b_{n}\zeta.
\]
We have
\begin{align*}
\E\left[T_{2ij}\right]=\  & \frac{1}{b_{n}^{d-m}}\int\int\left[K_{-\left(m\right)}\left(\zeta\right)\left(h_{0}\left(z_{\left(m\right)},\psi_{j}\left(z_{\left(m\right)}\right)+b_{n}\zeta\right)-h_{0}\left(\varphi\left(z_{\left(m\right)}\right)\right)\right)\right]\\
 & \quad\quad\quad p\left(\rest{\psi_{j}\left(z_{\left(m\right)}\right)+b_{n}\zeta_{-\left(m\right)}}z_{\left(m\right)}\right)b_{n}^{d-m}d\zeta p\left(z_{\left(m\right)}\right)\ol w\left(\varphi\left(z_{\left(m\right)}\right)\right)dz_{\left(m\right)}\\
=\  & \int\left[\int K_{-\left(m\right)}\left(\zeta\right)\left(h_{0}\left(z_{\left(m\right)},\psi_{j}\left(z_{\left(m\right)}\right)+b_{n}\zeta\right)-h_{0}\left(\varphi\left(z_{\left(m\right)}\right)\right)\right)\left(\rest{\psi_{j}\left(z_{\left(m\right)}\right)+b_{n}\zeta_{-\left(m\right)}}z_{\left(m\right)}\right)d\zeta\right]\\
 & \quad\quad\quad p\left(z_{\left(m\right)}\right)\ol w\left(\varphi\left(z_{\left(m\right)}\right)\right)dz_{\left(m\right)}\\
=\  & \int\left[\int K_{-\left(m\right)}\left(\zeta\right)\left(B\left(z_{\left(m\right)}\right)b_{n}^{s}\zeta^{s}+o\left(b_{n}^{s}\right)\right)d\zeta\right]p\left(z_{\left(m\right)}\right)\ol w\left(\varphi\left(z_{\left(m\right)}\right)\right)dz_{\left(m\right)}\\
=\  & b_{n}^{s}\cd\int K_{-\left(m\right)}\left(\zeta\right)\zeta^{s}d\zeta\cd\int B\left(z_{\left(m\right)}\right)p\left(z_{\left(m\right)}\right)\ol w\left(\varphi\left(z_{\left(m\right)}\right)\right)dz_{\left(m\right)}+o\left(b_{n}^{s}\right)\\
=\  & O\left(b_{n}^{s}\right)
\end{align*}
Hence, we also have $\E\left[\frac{1}{n}\sum_{i}T_{2,ij}\right] =  \E\left[T_{2,ij}\right]=O\left(b_{n}^{s}\right)$.
\end{proof}

\begin{lem}[Variance]\label{k-var}
$\text{Var}\left[\frac{1}{n}\sum_{i}T_{2,ij}\right] =O\left(\frac{1}{nb_{n}^{d-m}}\right)$
\end{lem}
\begin{proof}
Next,
\begin{align*}
& \E\left(T_{2ij}^{2}\right) \\
=\ & \frac{1}{b_{n}^{2\left(d-m\right)}}\int K_{-\left(m\right)}^{2}\left(\frac{\psi_{j}\left(z_{\left(m\right)}\right)-z_{-\left(m\right)}}{b_{n}}\right)\left(Y_{i}-h_{0}\left(\varphi\left(z_{\left(m\right)}\right)\right)\right)^{2}\ol w^{2}\left(\varphi\left(z_{\left(m\right)}\right)\right)\\
&\quad p\left(z_{\left(m\right)},z_{-\left(m\right)}\right)dz_{-\left(m\right)}dz_{\left(m\right)}\\
 =\ & \frac{1}{b_{n}^{2\left(d-m\right)}}\int K_{-\left(m\right)}^{2}\left(\frac{\psi_{j}\left(z_{\left(m\right)}\right)-z_{-\left(m\right)}}{b_{n}}\right)\left[\s_{0}^{2}\left(z\right)+\left(h_{0}\left(z\right)-h_{0}\left(\varphi\left(z_{\left(m\right)}\right)\right)\right)^{2}\right]\ol w^{2}\left(\varphi\left(z_{\left(m\right)}\right)\right)\\
 &\quad p\left(z_{\left(m\right)},z_{-\left(m\right)}\right)dz_{-\left(m\right)}dz_{\left(m\right)}\\
 =\ & \frac{1}{b_{n}^{2\left(d-m\right)}}\int\left[\underset{T_{3j}}{\underbrace{\int K_{-\left(m\right)}^{2}\left(\frac{\psi_{j}\left(z_{\left(m\right)}\right)-z_{-\left(m\right)}}{b_{n}}\right)\left[\s_{0}^{2}\left(z\right)+\left(h_{0}\left(z\right)-h_{0}\left(\varphi\left(z_{\left(m\right)}\right)\right)\right)^{2}\right] p\left(\rest{z_{-\left(m\right)}}z_{\left(m\right)}\right)dz_{-\left(m\right)}}}\right]\\
 &\quad \ol w^{2}\left(\varphi\left(z_{\left(m\right)}\right)\right)p\left(z_{\left(m\right)}\right)dz_{\left(m\right)}
\end{align*}
For $T_{3j}$, we can carry out the change of variable from $z_{-\left(m\right)}$
to $v$ as below:
\[
v:=\frac{z_{-\left(m\right)}-\psi_{j}\left(z_{\left(m\right)}\right)}{b_{n}}\quad\iff\quad z_{-\left(m\right)}=\psi_{j}\left(z_{\left(m\right)}\right)+b_{n}v
\]
and thus
\begin{align*}
T_{3j} =\ & \int K_{-\left(m\right)}^{2}\left(v\right)\left[\s_{0}^{2}\left(z_{\left(m\right)},\psi_{j}\left(z_{\left(m\right)}\right)+b_{n}v\right)+\left(h_{0}\left(z_{\left(m\right)},\psi_{j}\left(z_{\left(m\right)}\right)+b_{n}v\right)-h_{0}\left(\varphi\left(z_{\left(m\right)}\right)\right)\right)^{2}\right]\\
&\quad p\left(\rest{\psi_{j}\left(z_{\left(m\right)}\right)+b_{n}v}z_{\left(m\right)}\right)b_{n}^{d-m}dv\\
 =\ & b_{n}^{d-m}\int K_{-\left(m\right)}^{2}\left(v\right)\left[\s_{0}^{2}\left(z_{\left(m\right)},\psi_{j}\left(z_{\left(m\right)}\right)\right)+\left(h_{0}\left(z_{\left(m\right)},\psi_{j}\left(z_{\left(m\right)}\right)\right)-h_{0}\left(\varphi\left(z_{\left(m\right)}\right)\right)\right)^{2}+O\left(b_{n}\right)\right]\\
 &\quad \left[p\left(\rest{\psi_{j}\left(z_{\left(m\right)}\right)}z_{\left(m\right)}\right)+O\left(b_{n}\right)\right]dv\\
 =\ & b_{n}^{d-m}\int K_{-\left(m\right)}^{2}\left(v\right)dv\cd\s_{0}^{2}\left(z_{\left(m\right)},\psi_{j}\left(z_{\left(m\right)}\right)\right)p\left(\rest{\psi_{j}\left(z_{\left(m\right)}\right)}z_{\left(m\right)}\right)+o\left(b_{n}^{d-m}\right)\\
 & =b_{n}^{d-m}\cd C\cd\s_{0}^{2}\left(z_{\left(m\right)},\psi_{j}\left(z_{\left(m\right)}\right)\right)p\left(\rest{\psi_{j}\left(z_{\left(m\right)}\right)}z_{\left(m\right)}\right)+o\left(b_{n}^{d-m}\right)
\end{align*}
Substituting $T_{3j}$ back into the expression for $\E(T_{2ij}^2)$, we obtain
\begin{align*}
&\E\left(T_{2ij}^{2}\right) \\
=\ & \frac{1}{b_{n}^{2\left(d-m\right)}}\int\left[b_{n}^{d-m}\cd C\cd\s_{0}^{2}\left(z_{\left(m\right)},\psi_{j}\left(z_{\left(m\right)}\right)\right)p\left(\rest{\psi_{j}\left(z_{\left(m\right)}\right)}z_{\left(m\right)}\right)+o\left(b_{n}^{d-m}\right)\right]\ol w^{2}\left(\varphi\left(z_{\left(m\right)}\right)\right)p\left(z_{\left(m\right)}\right)dz_{\left(m\right)}\\
 =\ & \frac{1}{b_{n}^{d-m}}C\cd\int\s_{0}^{2}\left(z_{\left(m\right)},\psi_{j}\left(z_{\left(m\right)}\right)\right)p\left(\rest{\psi_{j}\left(z_{\left(m\right)}\right)}z_{\left(m\right)}\right) \ol w^{2}\left(\varphi\left(z_{\left(m\right)}\right)\right)p\left(z_{\left(m\right)}\right)dz_{\left(m\right)}+o\left(\frac{1}{b_{n}^{d-m}}\right)\\
 =\ & O\left(\frac{1}{b_{n}^{d-m}}\right)
\end{align*}
Hence,
\begin{align*}
\text{Var}\left[\frac{1}{n}\sum_{i}T_{2,ij}\right]  & \leq\frac{C{}}{n^2}\sum_{i=1}^{n}O\left(\frac{1}{b_{n}^{d-m}}\right)+R_{3}=O\left(\frac{1}{nb_{n}^{d-m}}\right).
\end{align*}
\end{proof}

\section{Bias-Aware versus Undersmoothing Confidence Intervals}\label{sec:BiasAware}

In the Monte Carlo simulations of Section~\ref{sec:Simulation}, the confidence intervals are constructed via the undersmoothing approach: the bootstrap-Lepski procedure selects the rate-optimal sieve dimension $\hat{K}$, which is then inflated to $\tilde{K} = \lceil\sqrt{\log n}\cdot\hat{K}\rceil$, and the CI is constructed at $\tilde{K}$ using the standard normal critical value $z_{0.975}=1.96$. In this appendix, we compare this approach with the bias-aware confidence interval construction of \citet*{chen2025adaptive}, which retains the rate-optimal $\hat{K}$ and instead inflates the critical value with an explicit bias buffer:
\[
\text{CI}_{\text{BA}} := \left[\hat{\theta}(\hat{K}) \pm \left(z_{0.975} + \hat{A}\cdot \vartheta^*_{1-\hat{\alpha}}\right) \cdot \widehat{se}(\hat{K})\right],
\]
where $\hat{A} = \log\log \hat{J}$, $\vartheta^*_{1-\hat{\alpha}}$ is the $(1-\hat{\alpha})$-quantile of the bootstrap-Lepski maximal contrast statistic (see Appendix \ref{subsec:MB_K_LOO}), and $\hat{\alpha} = \min(0.5,\, \sqrt{\log K_{\max}/K_{\max}})$. The bias-aware construction builds on the honest confidence interval framework of \citet*{armstrong2018optimal,armstrong2020simple} and was developed by \citet*{chen2025adaptive} for linear functionals of the nonparametric estimator; no analogous construction for the general class of nonlinear functionals has been developed in the literature. We examine both designs from Section~\ref{sec:Simulation} to compare the two approaches.

\paragraph{Design 1: Linear integral on the unit circle.}
Table~\ref{tab:BiasAware_Sim1} compares the undersmoothing (US) and bias-aware (BA) confidence intervals for the linear integral functional considered in Section~\ref{subsec:Sim1_Circle}. The bias-aware CI is shorter than the undersmoothing CI at every sample size, with the reduction ranging from 3.2\% at $n=500$ to 4.5\% at $n=1000$. At the same time, the bias-aware CI achieves valid coverage at all sample sizes, though somewhat conservative (98.8\%--99.2\%) relative to the 95\% nominal level. This is consistent with the fact that for this linear functional, the sieve bias at the rate-optimal $\hat{K}\approx 17$ is already of negligible order relative to the standard deviation, so the bias buffer adds an unnecessary but harmless inflation to the critical value. Notably, the rate-optimal $\hat{K}\approx 17$ is substantially smaller than the undersmoothed $\tilde{K}\approx 51$, which means the bias-aware estimator operates with lower variance, contributing to the shorter CI length.

\begin{table}[!h]
\centering
\caption{Design 1 (Linear Integral): Undersmoothing vs.\ Bias-Aware CI}
\label{tab:BiasAware_Sim1}
\begin{tabular}{r|cc|cc|cc}
\toprule
 & \multicolumn{2}{c|}{CI Length (U--L)} & \multicolumn{2}{c|}{Coverage} & \multicolumn{2}{c}{Mean Sieve Dim.} \\
\midrule
$n$ & US & BA & US & BA & $\bar{\tilde{K}}$ & $\bar{\hat{K}}$ \\
\midrule
  500  & 2.3465 & 2.2704 & 93.8\% & 98.8\% & 50.6 & 17.1 \\
 1000  & 1.6589 & 1.5847 & 95.2\% & 98.9\% & 50.7 & 16.8 \\
 2000  & 1.1796 & 1.1295 & 95.0\% & 98.9\% & 51.4 & 17.1 \\
 4000  & 0.8338 & 0.7980 & 94.4\% & 99.2\% & 51.4 & 17.2 \\
 8000  & 0.5901 & 0.5657 & 95.8\% & 99.2\% & 51.8 & 17.4 \\
\bottomrule
\end{tabular}

\vspace{0.5em}
\parbox{\textwidth}{\footnotesize
``US'': undersmoothing with $\tilde{K} = \lceil\sqrt{\log n}\cdot\hat{K}\rceil$ and normal CV $z_{0.975}=1.96$.
``BA'': bias-aware CI at rate-optimal $\hat{K}$ with inflated CV $z_{0.975} + \hat{A}\cdot\vartheta^*_{1-\hat{\alpha}}$.
}
\end{table}

\paragraph{Design 2: Nonlinear integral on the estimated unit disk.}
Table~\ref{tab:BiasAware_Sim2} presents the analogous comparison for the nonlinear upper contour set integral considered in Section~\ref{subsec:Sim2_Disk}. Although the bias-aware CI construction of \citet*{chen2025adaptive} was not designed for nonlinear functionals, we apply it here to examine its empirical behavior outside the scope of its theoretical guarantees.

In contrast to Design~1, the bias-aware CI is \emph{wider} than the undersmoothing CI at all sample sizes. At $n=500$, the difference is modest (7--8\%), but the gap grows substantially with $n$, reaching approximately 25\% at $n=8000$. This occurs because the nonlinear functional exhibits substantial sieve bias at the rate-optimal $\hat{K}\approx 17$--$36$: the bias is 25--60\% of the standard deviation, far from negligible. The bias buffer $\hat{A}\cdot\vartheta^*_{1-\hat{\alpha}}$ inflates the standard error at $\hat{K}$, which is itself already larger than the standard error at the undersmoothed $\tilde{K}\approx 51$--$94$. While the bias-aware CI maintains valid coverage (95.5\%--96.4\%), the undersmoothing approach achieves comparable or better coverage with shorter intervals, because it directly reduces the bias by increasing the sieve dimension.

That the bias-aware approach performs less well for the nonlinear functional is expected to some extent, given that the construction was designed for linear functionals. However, the precise magnitudes---the 7--25\% increase in CI length relative to undersmoothing, growing with $n$---were not obvious a priori. This exercise highlights the robustness of the undersmoothing approach: it achieves valid and efficient confidence intervals for both linear and nonlinear functionals, whereas the bias-aware approach, while theoretically well-founded and effective for linear functionals, produces wider confidence intervals when applied to nonlinear functionals where substantial sieve bias persists at the rate-optimal sieve dimension.

\begin{table}[!h]
\centering
\caption{Design 2 (Nonlinear Integral): Undersmoothing vs.\ Bias-Aware CI}
\label{tab:BiasAware_Sim2}
\begin{tabular}{r|cc|cc|cc}
\toprule
 & \multicolumn{2}{c|}{CI Length (U--L)} & \multicolumn{2}{c|}{Coverage} & \multicolumn{2}{c}{Mean Sieve Dim.} \\
\midrule
$n$ & US & BA & US & BA & $\bar{\tilde{K}}$ & $\bar{\hat{K}}$ \\
\midrule
\multicolumn{7}{c}{(A) Plug-In} \\
\midrule
  500  & 0.2659 & 0.2852 & 94.1\% & 96.2\% & 51.1 & 17.3 \\
 1000  & 0.1902 & 0.2044 & 95.2\% & 95.5\% & 55.3 & 18.0 \\
 2000  & 0.1370 & 0.1507 & 95.0\% & 96.0\% & 65.7 & 21.7 \\
 4000  & 0.0991 & 0.1148 & 95.4\% & 96.4\% & 81.9 & 28.2 \\
 8000  & 0.0721 & 0.0901 & 95.5\% & 96.2\% & 94.4 & 35.9 \\
\midrule
\multicolumn{7}{c}{(B) Leave-One-Out} \\
\midrule
  500  & 0.3314 & 0.3584 & 94.8\% & 96.0\% & 50.5 & 17.1 \\
 1000  & 0.1966 & 0.2055 & 96.2\% & 96.0\% & 54.9 & 17.8 \\
 2000  & 0.1394 & 0.1510 & 95.3\% & 95.7\% & 64.9 & 21.4 \\
 4000  & 0.1004 & 0.1150 & 95.4\% & 96.4\% & 82.0 & 28.3 \\
 8000  & 0.0726 & 0.0906 & 95.6\% & 96.3\% & 94.4 & 35.9 \\
\bottomrule
\end{tabular}

\vspace{0.5em}
\parbox{\textwidth}{\footnotesize
``US'': undersmoothing with $\tilde{K} = \lceil\sqrt{\log n}\cdot\hat{K}\rceil$ and normal CV $z_{0.975}=1.96$.
``BA'': bias-aware CI at rate-optimal $\hat{K}$ with inflated CV $z_{0.975} + \hat{A}\cdot\vartheta^*_{1-\hat{\alpha}}$.
}
\end{table}

\end{document}